\documentclass[11pt, letterpaper]{book}
\usepackage{tikz-cd}
\usetikzlibrary{calc}
\usepackage[utf8]{inputenc}
\usepackage{appendix}
\usepackage{pst-all}
\usepackage{pstcol}
\usepackage{etex}
\usepackage{hyperref}
\usepackage{tikz}
\usepackage[draft,inline,nomargin]{fixme} \fxsetup{theme=color}
\FXRegisterAuthor{cp}{cpp}{\color{teal}CP}
\def\fxnote#1{}
\def\fxwarning#1{}
\renewcommand{\blue}{}

\usepackage{mathptmx}

\usepackage[format=hang,indention=-1.0cm,margin=10pt]{caption}
\usepackage{graphicx}
\usepackage{lettrine}
\usepackage{amsfonts}
\usepackage{amsmath}
\usepackage{amsthm}
\usepackage{amssymb}
\usepackage{bbm}
\usepackage{fancyhdr}                    
\usepackage[hmarginratio=1:1,textwidth=12.5cm,totalheight=20.5cm,includehead]{geometry}
\usepackage{subfigure}
\usepackage{setspace}

\makeatletter
\def\cleardoublepage{\clearpage\if@twoside \ifodd\c@page\else%
    \hbox{}%
    \thispagestyle{empty}
    \newpage%
    \if@twocolumn\hbox{}\newpage\fi\fi\fi}
\makeatother

\renewenvironment{thebibliography}[1]{\begin{oldthebibliography}{#1}%
    \setlength{\parskip}{0cm}\setlength{\itemsep}{.2cm}}%
{\end{oldthebibliography}}

\def\>{\rangle}
\def\<{\langle}
\newcommand{\mcO}{\mathcal{O}}
\newcommand{\mcU}{\mathcal{U}}
\newcommand{\mcH}{\mathcal{H}}

\newcommand{\rmi}{\mathrm{i}} 
\newcommand{\tr}{\mathop{\mathrm{tr}}\nolimits}
 
\newcommand{\mcA}{\mathcal{A}}
\newcommand{\mcG}{\mathcal{A}}
\newcommand{\mcI}{\mathcal{I}}

\newcommand{\mcN}{\mathcal{N}}

\newcommand{\mcE}{\mathcal{E}}

\newcommand{\mcB}{\mathcal{B}}
\newcommand{\mcL}{\mathcal{L}}
\newcommand{\mcF}{\mathcal{F}}
\newcommand{\mcT}{\mathcal{T}}
\newcommand{\mcS}{\mathcal{S}}
\newcommand{\mcD}{\mathcal{D}}
\newcommand{\diag}{\text{diag}}
\newcommand{\ie}{i.e.}

\newcommand{\eref}[1]{eq.~(\ref{#1})} 
\newcommand{\sref}[1]{sec.~\ref{#1}}
\newcommand{\fref}[1]{fig.~\ref{#1}}

\newcommand{\Eref}[1]{Eq.~(\ref{#1})}

\newcommand{\Div}{\ensuremath{{\sf C}^\text{div}}}
\newcommand{\pDiv}{\ensuremath{{\sf C}^\text{P}}}
\newcommand{\cpDiv}{\ensuremath{{\sf C}^\text{CP}}}
\newcommand{\LDiv}{\ensuremath{{\sf C}^\text{L}}}

\newcommand{\Ind}{\ensuremath{\overline{{\sf C}^{\rm div}}}}
\newcommand{\InftyDiv}{\ensuremath{{\sf C}^\infty}}

\newcommand{\InfDiv}{\ensuremath{{\sf C}^\text{Inf}}}
\newcommand{\cptp}{\ensuremath{\sf C}}
\newcommand{\hp}{HP}
\newcommand{\DivFunc}{\delta}
\newcommand{\one}{\mathbbm{1}}
\newcommand{\gqc}{GQC}
\newcommand{\sgqc}{SGQC}
\newcommand{\dgqc}{$\delta\text{GQC}$}
\newcommand{\gf}{GF}
\newcommand{\tp}{TP}
\newcommand{\Imi}{\imath}
\newtheorem{definition}{Definition}
\newtheorem{corollary}{Corollary}
\newcommand{\gs}{GS}
\usepackage{pdfpages}

\newtheorem{theorem}{Theorem}
\newtheorem{proposition}{Proposition}
\newtheorem{lemma}{Lemma}
\newtheorem{example}{Example}

\newcommand{\ket}[1]{{\vert #1 \rangle}}
\newcommand{\bra}[1]{{\langle #1 \vert}}
\newcommand{\proj}[2]{{\vert #1 \rangle \langle #2 \vert}}

\newcommand{\fig}[1]{fig.\ref{#1}}

\newcommand{\id}{\text{id}}
\newcommand{\Jami}{Choi-Jamio\l{}kowski}
\newcommand{\scho}{Schr\"odinger}

\newhsbcolor{huezero}{0 1 1}
\newhsbcolor{huepone}{0.1 1 1}
\newhsbcolor{huepfour}{0.4 1 1}
\newhsbcolor{huepsix}{0.6 1 1}
\newhsbcolor{huepseven}{0.7 1 1}
\newhsbcolor{hueoneosix}{0.1666 1 1}
\newhsbcolor{hueoneothree}{0.3333 1 1}
\newhsbcolor{hueoneotwo}{0.5 1 1}
\newhsbcolor{huetwoothree}{0.6666 1 1}
\newhsbcolor{huefiveosix}{0.8333 1 1}

\providecommand{\openone}{\leavevmode\hbox{\small1\kern-3.8pt\normalsize1}}
\graphicspath{{images_arxiv/}}
\begin{document}
\begin{titlepage}
\begin{center}
\includegraphics[draft=false,width=.3\textwidth]{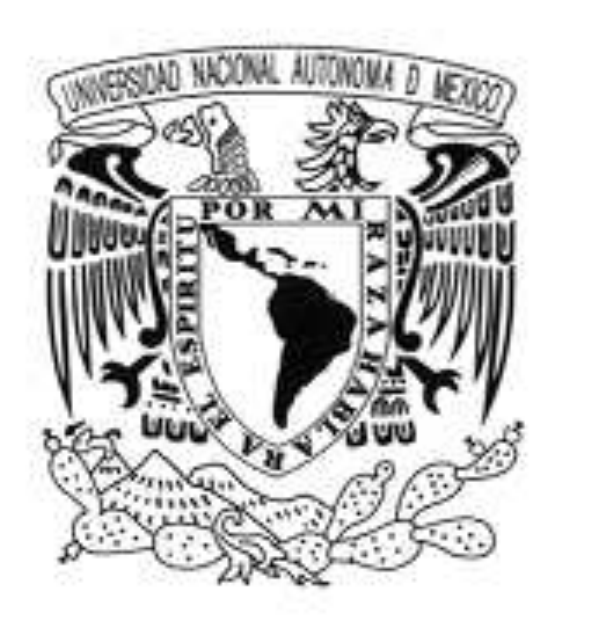}\hfill\includegraphics[draft=false,width=.3\textwidth]{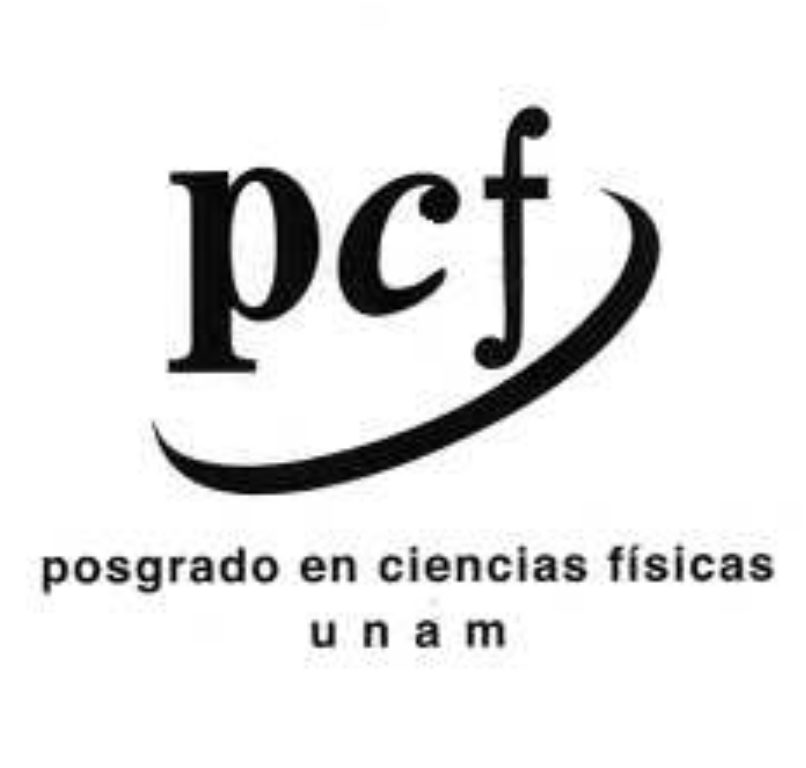}\\
\large UNIVERSIDAD NACIONAL AUT\'ONOMA DE M\'EXICO\\
POSGRADO EN CIENCIAS F\'ISICAS\\[1cm]
\huge Divisibility classes of qubit maps and singular Gaussian channels\\[.5cm]\large TESIS\\[1.2cm]

\begin{tabular}{rl}
 Que para obtener el grado de: & Doctor en Ciencias (F\'isica)\\
 Presenta: & David D\'avalos Gonz\'alez\\
 Director de tesis:& Dr. Carlos Francisco Pineda Zorrilla\\
 Codirector: & Dr. Mario Ziman \\
 \end{tabular}

\vspace*{\stretch{1}}
{\small Miembros del Comit\'e Tutoral: \\
Dr. Carlos Pineda, Dr. Luis Benet,  y Dr. Thomas H. Seligman}\\
\includegraphics[draft=false,scale=0.4]{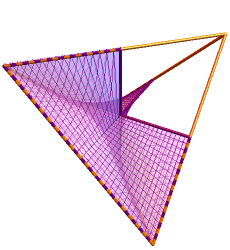}
\includegraphics[draft=false,scale=0.4]{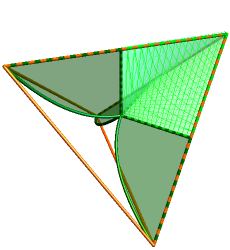}\\
M\'exico D. F., 2019
\end{center}
\cleardoublepage
\pagestyle{empty}
\begin{center}
\includegraphics[draft=false,width=.3\textwidth]{logo_unam}\hfill\includegraphics[draft=false,width=.3\textwidth]{logo_pcf}\\
\large UNIVERSIDAD NACIONAL AUT\'ONOMA DE M\'EXICO\\
POSGRADO EN CIENCIAS F\'ISICAS\\[1cm]
\huge Divisibility classes of qubit maps and singular Gaussian channels\\[.5cm]\large THESIS\\[1.2cm]

\begin{tabular}{rl}
 To obtain the degree: & Doctor en Ciencias (F\'isica)\\
 Presents: & David D\'avalos Gonz\'alez\\
 Director:& Dr. Carlos Francisco Pineda Zorrilla\\
 Co-director:& Dr. Mario Ziman \\
 \end{tabular}

\vspace*{\stretch{1}}
{\small Members of the Tutorial Committee: \\
Dr. Carlos Pineda, Dr. Luis Benet,  and Dr. Thomas H. Seligman}\\
\includegraphics[draft=false,scale=0.4]{CPportada.png}
\includegraphics[draft=false,scale=0.4]{Lportada.png}\\
M\'exico D. F., 2019
\end{center}
\end{titlepage}

\cleardoublepage

\setcounter{page}{1}\pagenumbering{roman} \setcounter{page}{1}
\vspace*{\stretch{1}}
\begin{flushright}
\textit{Truth is ever to be found in simplicity, and not in the multiplicity and confusion of things.}\\
Isaac Newton
\end{flushright}

\vspace*{\stretch{1}}
\thispagestyle{empty}
\cleardoublepage

{ \setlength\parindent{0cm} \setlength\parskip{0.1cm}

\begin{center} \Huge Gracias \end{center}
La idea que nació cuando cursaba la secundaría, la de convertirme algún día en científico, no habría sido posible de no haber nacido en el seno de una familia estable, funcional y de sólidos valores. Por eso les agradezco infinitamente a mis queridos Padres. A mi Madre, Sara, por dedicarme tanto de su tiempo y energía, por ser una Madre muy amorosa, llena de valores y por ser la persona mas paciente del mundo. Le agradezco a mi Padre, Juan Manuel, por siempre estar atento a que fuera una persona de principios y un buen ciudadano, por ser un padre amoroso y por darme su confianza. Le agradezco que siempre se haya preocupado por tener una computadora en casa y por ser un entusiasta de la tecnología, eso aportó fundamentalmente a quien soy hoy. A mis hermanas y hermanos por preocuparse por mi y por regalarme tantas veces su tiempo. 
A mi compañera de vida, a mi esposa Lorena, gracias por tenerme tanta paciencia, por creer en mi y por quererme tanto. 
A mi tutor y amigo, Carlos, le agradezco su paciencia, sus valiosas enseñanzas, su gran apoyo y su amistad. A Thomas Seligman y Luis Benet por siempre apoyarme.
Le agradezco a mis amigos Luis Juárez, Arturo Carranza, Thomas Gorin, Mario Ziman, Mauricio Torres, Roberto León, François Leyvraz, Pablo Barberis, Diego Wisniacki, Ignacio García, Juan Diego Urbina, Peter Rap\v can, Tomá\v s Rybár, Edgar Aguilar, David Amaro, Alvaro Díaz, Miguel Cardona, Chayo Camarena, Antonio Rosado, Sergio Sánchez, Nephtalí Garrido, Sergio Pallaleo, Samuel Rosalio, Alejandro Reyes, Afra Montero y Daniel Garibay.
A mis gatitos Lola y Dalí por hacerme feliz el poco tiempo que estuvieron en este mundo. A mi querida gatita Lulú por hacer de mi hogar siempre un lugar feliz. Les agradezco a todos los seres queridos que hicieron de esta parte de mi trayectoria algo memorable.

\par
Agradezco el apoyo brindado por los proyectos PAPIIT número IG100518 y CONACYT CB-285754.
}
\vfill
\vfill
\thispagestyle{empty}
\cleardoublepage
\begin{center} {\bf SYNOPSIS}\\ \end{center}
We present two projects concerning the main part of my PhD work. In the first
one we study quantum channels, which are the most general operations mapping
quantum states into quantum states, from the point of view of their
divisibility properties. We introduced tools to test if a given quantum channel
can be implemented by a process described by a Lindblad master equation. This
in turn defines channels that can be divided in such a way that they form a
one-parameter semigroup, thus introducing the most restricted studied
divisibility type of this work. Using our results, together with the study of other
types of divisibility that can be found in the literature, we characterized the
space of qubit quantum channels. We found interesting results connecting the
concept of entanglement-breaking channel and infinitesimal divisibility.
Additionally we proved that infinitely divisible channels are equivalent to the ones that
are implementable by one-parameter semigroups, opening this question for
more general channel spaces.
In the second project we study the functional forms of one-mode Gaussian
quantum channels in the position state representation, beyond Gaussian
functional forms. We perform a \textit{black-box} characterization using
complete positivity and trace preserving conditions, and report the existence
of two subsets that do not have a functional Gaussian form. The study covers as
particular limit the case of singular channels, thus connecting our results
with the known classification scheme based on canonical forms. Our full
characterization of Gaussian channels without Gaussian functional form is
completed by showing how
Gaussian states are transformed under these operations, and by deriving the conditions for the existence 
of master equations for the non-singular cases.
\vfill
{ \setlength\parindent{0cm} \setlength\parskip{0.1cm}
}
\thispagestyle{empty}
\cleardoublepage
\begin{center} {\bf RESUMEN}\\ \end{center} 
En esta tesis se presentan dos proyectos realizados durante mis estudios de doctorado. En el
primero se estudian los canales cuánticos, que son las 
operaciones más generales que
transforman estados cuánticos en estados cuánticos,
desde el punto de vista de sus
propiedades de divisibilidad. Introducimos herramientas
para probar si un canal cuántico dado puede ser implementado por un proceso
descrito por una ecuación maestra de Lindblad. Ésto a su vez define a los
canales que pueden ser divididos de tal manera que ellos forman semigrupos de
un parámetro, introduciendo entonces el tipo más restringido de divisibilidad
estudiado de este trabajo. Usando nuestros resultados, junto con el estudio de
otros tipos de divisibilidad que pueden ser encontrados en la literatura,
caracterizamos el espacio de canales cuánticos de un qubit. Encontramos
resultados interesantes que conectan el concepto de canales que rompen el
entrelazamiento (del sistema con cualquier sistema auxiliar) y el de
divisibilidad infinitesimal. Además probamos que el conjunto de canales
infinitamente divisibles es equivalente al de los canales implementables por
semigrupos de un parámetro. Ésto abre la pregunta sobre si esto sucede para
espacios de canales más generales.
En el segundo proyecto estudiamos las formas funcionales de canales Gaussianos
de un solo modo, más allá de la forma funcional Gaussiana. Se hace una
caracterización de \textit{caja negra} utilizando las condiciones de completa
positividad y preservación de la traza, y se reporta la existencia de dos
subconjuntos que no poseen forma funcional Gaussiana. El estudio cubre en
particular el límite de los canales singulares, conectando entonces nuestros
resultados con la la clasificación basada en formas canónicas. Nuestra
caracterización de canales Gaussianos sin forma funcional Gaussiana es
completada mostrando como los estados Gaussianos se transforman bajo esas
operaciones, así como al derivar las condiciones para la existencia de
ecuaciones maestras para los casos no singulares.
\vspace{.5in}

{ \setlength\parindent{0cm} \setlength\parskip{0.2cm}
\newcommand{\ignore}[1]{}
\newcommand{\ignorenot}[1]{#1}

} 
\thispagestyle{empty}
\cleardoublepage

\tableofcontents
\setlength\parindent{0cm} \setlength\parskip{0.3cm}
\mainmatter
\chapter{Introduction}
\begin{flushright}
\textit{In questions of science, the authority of a thousand is not worth the humble reasoning of a single individual.}\\
Galileo Galilei
\end{flushright}
The advent of quantum technologies opens questions aiming for deeper
understanding of the fundamental physics beyond the idealized case of isolated
quantum systems. Also the well established Born-Markov approximation used to
describe open quantum systems (e.g. relaxation process such as spontaneous
decay and decoherence) is of limited use and a more general framework of open
system dynamics is required. Recent efforts in this area have given rise to
relatively novel research subjects - non-markovianity and divisibility.

A central object of study in quantum information theory and open quantum
systems are quantum channels, also called quantum operations. They describe,
for instance, the noisy communication between Alice and Bob or the changes
that an open quantum system undergoes at some fixed time. They can also be seen
as the basic building blocks of time-dependent quantum processes (also called
quantum dynamical maps). Conversely, families of quantum channels arise
naturally given a quantum dynamical map.

Given a quantum channel, for instance a spin flip or the approximation of the
universal NOT gate, one can wonder about how it can be implemented.
The latter in the sense of, being quantum channels discrete operations, can we
find a continuous time-dependent process that at some time it implements the
given channel?; or is there a process such that
we ``just wait for a relaxation of the physical system'' to implement such
channel? 
It turns out that this question is related with the one of finding
simpler operations such that their concatenation equals the given quantum
channel~\cite{cirac}. Such operations are simpler in the sense
that they are closer to the subset of unitary operations, or even ``smaller''
in the sense that they are closer to the identity channel.


This thesis encompasses the results of two works developed during my PhD. 

The first and the most extended one was devoted to study the divisibility
properties of quantum channels (discrete evolutions of quantum systems), for
the particular case of qubits. We revise the divisibility types introduced in
the seminal paper by Wolf et al.~\cite{cirac} and derived several useful
relations to decide each type of divisibility. In particular, we characterize
channels that can be divided in such a way that they belong to one-parameter
semigroups (dynamics described by Lindblad master equations), and extended the
analysis of~\cite{Wolf2008} for channels with negative eigenvalues. We did this
using the results by Evans et al.~\cite{Evans1977} and
Culver~\cite{Culver1966}.

Beyond the mentioned characterization tools, the principal aim of the work was to understand the forms of non-markovianity standing behind the
observed quantum channels. The non-markovianity character describes
the back-action of the system's environment on the system's future time
evolution. Such phenomena is identified as emergence
of memory effects~\cite{rivasreview,breuerreview,ourmeasure}.
On the other side, divisibility questions the possibility of splitting
a given quantum channel into a concatenation of other quantum channels.
In this work we will investigate the relation between these two notions.
Thus, we related features of of continuous time evolutions of quantum
systems, and the concept of divisibility of quantum maps, which are discrete
evolutions. A very first example of this is the well known identification of
one-parameter semigroups with Lindbladian dynamics~\cite{lindblad}.

The second project is devoted to representation theory of continuous-variable
quantum systems, which is a central topic of study given its role in the
description of physical systems like the electromagnetic
field~\cite{doi:10.1142/p489}, solids and nano-mechanical
systems~\cite{RevModPhys.86.1391} and atomic
ensembles~\cite{RevModPhys.82.1041}. In this theory the simplest states, both
from a theoretical and experimental point of view, are the so-called Gaussian
states. An operation that transforms such family of states into itself is
called a Gaussian quantum channel (\gqc{}).  Even though Gaussian states and
channels form small subsets among general states and channels, they have proven
to be useful in a variate of tasks such as quantum
communication~\cite{Grosshans2003}, quantum
computation~\cite{PhysRevLett.82.1784} and the study of quantum entanglement in
simple~\cite{RevModPhys.77.513} and complicated
scenarios~\cite{PhysRevA.98.022335}.  In this project we study the possible
functional forms that one-mode Gaussian quantum channels can have in the
position state representation, and characterize the particular case of singular
channels. Although they are already characterized by their action on the first
and second moments of Gaussian states~\cite{Holevo2007,Reviewquantuminfo}, we
connect our framework to such known results. Additionally we give an insight of
the possible functional forms of, for instance, Gaussian unitaries.


The thesis is organized as follows: 
In chapter~\ref{chap:open_quantum_systems} we discuss the most widely adopted
scheme to study open quantum systems, introducing the formalism of bipartite
systems and useful tools for it. Later on we present the general setting for
system plus reservoir dynamics and its formal solution. As a paradigmatic
example of open system dynamics, we present briefly the microscopic derivation
of the Lindblad master equation using the well known Born-Markov approximation, and discuss the properties of the generator of the dynamics. Subsequently we
introduce the formalism of quantum channels, being the most general operations
over quantum systems (excluding post-selection), by introducing some useful
mathematical definitions and contrasting with its classical analog.
Additionally we discuss briefly the concept of \textit{local operations and
classical communications} (LOCC), also known as filtering operations.
Finally we give a very brief introduction to continuous variable systems,
giving special attention to Gaussian states and channels.

In chapter~\ref{chap:reps} we discuss the different available representations
for quantum channels and their relation with the concept of complete
positivity. In particular we introduce the well known Kraus representation and
discuss the \Jami{} theorem which in turn defines a very useful representation
to study quantum channels and their divisibility properties. Later on we
introduce various matrix representations of quantum channels, paying special
attention to hermitian and traceless bases types (without taking into account
the component proportional to identity). Furthermore we introduce useful
decompositions of qubit channels into unitary conjugations and one-way
stochastic local operations, and classical communication, both being analogous
to the well known singular value decomposition. Finally we give an introduction
to representations of Gaussian channels and a detailed derivation of the
position-state representations for Gaussian channels without Gaussian
functional form.

In chapter~\ref{chap:div} we give the definition of divisible quantum channel,
as well as the definition of various subclasses of divisible channels concerning
additional properties. In particular we discuss the concepts of infinitesimal
and infinitely divisible channels and some relations and inclusions between
them. Among infinitesimal divisible channels we identify two subclasses, being
the set of infinitesimal divisible channels in complete positive and positive
(but not complete positive) maps. Later on we introduce the concept of
L-divisible channels, defining the set of channels which are members of
one-parameter semigroups. We show that the set of infinitely divisible channels
is the same of the L-divisible Pauli channels.

In chapter~\ref{chap:singular} we study one-mode Gaussian quantum channels in
continuous-variable systems by performing a \textit{black-box} characterization
using complete positivity and trace preserving conditions, and report the
existence of two subsets that do not have a functional Gaussian form. Our study
covers as particular limit the case of singular channels, thus connecting our
results with their known classification scheme based on canonical forms. Our
full characterization of Gaussian channels without Gaussian functional form is
completed by showing how Gaussian states are transformed under these
operations, and by deriving the conditions for the existence of master
equations for the non-singular cases. 

In chapter~\ref{chap:summary} we give a summary of the two projects introduced
in this work and conclusions.

Finally, in the appendix~\ref{sec:unbounded} we prove that the exact reduced
dynamics of an open quantum system never follow a Lindblad master equation
unless they are unitary, given a bounded global Hamiltonian. In
appendix~\ref{sec:normal_form} we give an example that shows that the set of
Lorentz normal forms introduced in the literature, is incomplete.

\chapter{Open quantum systems and quantum channels}
\begin{flushright}
\textit{When we talk mathematics, we may be discussing a secondary language built on the primary language of the nervous system.}\\ 
John Von Neumann
\end{flushright}
\label{chap:open_quantum_systems}
In this chapter we introduce the usual scheme to study open quantum systems,
the widely known Born-Markov approximation and the concept of CP-divisibility.
Later on and based on the idea of \textit{(classical) stochastic map}, we
discuss the axiomatic formulation of quantum channels and its connection with
the usual construction of open quantum systems. Finally, for continuous
variable systems, we discuss the paradigmatic example of Gaussian channels.

\section{Introduction to the scheme of open quantum systems} 
The most widely used scheme to study open quantum systems is based on the idea of
study a closed system composed by the \textit{central system} and its
\textit{environment}, see \fref{fig:open_scheme}
 for an schematic
explanation. Thus, concepts as bipartite Hilbert spaces, density matrix and
partial trace are useful tools to study open systems. In what follows we give a
brief review of them. 
\paragraph{Bipartite Hilbert space.} Consider a bipartite closed quantum system
described by a Hilbert space with the structure $\mcH=\mcH_\text{S}\otimes
\mcH_\text{E}$, where $\mcH_\text{S}$ is the Hilbert space of the open system
and $\mcH_\text{E}$ is the Hilbert space of the \textit{environment}. If
$\lbrace \ket{\phi^\text{S}_i} \rbrace_{i=1}^{\dim\left(\mcH_S \right)}$ and
$\lbrace \ket{\phi^\text{E}_i} \rbrace_{i=1}^{\dim\left(\mcH_E \right)}$ are
basis for the spaces $\mcH_\text{S}$ and $\mcH_\text{E}$, respectively, a
basis for $\mcH$ is simply $\lbrace \ket{\phi^\text{S}_i}\otimes
\ket{\phi^\text{E}_j}\rbrace_{i=1,j=1}^{\dim\left(\mcH_\text{S}\right),
\dim\left(\mcH_\text{E}\right)}$. It is typical that for finite dimensional
systems one has that $\dim(\mcH_\text{E})\gg \dim(\mcH_\text{S})$ as the
environment is usually ``bigger'' than the central system.

To describe the states of open quantum systems it is necessary to model the
\textit{ignorance} that the observer has with respect to  the open system. 
Since the experimentalist cannot access the degrees of freedom of the
environment, they are simply ignored. To do this we need the two following
concepts.

\paragraph{Density matrix.} Let a quantum system that has probability $p_i$ to
be in the state $\ket{\phi_i}$, and let the operator $A$ an observable over
such system. Using the average formula $\langle A \rangle =\sum_i p_i \bra
{\phi_i} A \ket{\phi_i}$ it is straightforward to show that $\langle A \rangle
= \tr \left( A \rho \right)$ with 
\begin{equation}
\rho=\sum_i p_i \proj{\phi_i}{\phi_i},
\label{eq:density_matrix_def}
\end{equation}
and $\sum_i p_i=1$. $\rho$ is called \textit{density operator}
or \textit{density matrix}. Note that $\rho$ is a positive-semidefinite operator
given that $p_i\geq 0$, and the states $\ket{\phi_i}$ do not need to be
orthogonal. Also note that since $\rho$ is hermitian, together with the
positive-semidefiniteness, implies that we can always write any density matrix
as a convex combination of orthogonal pure states.
Thus, every operator $\rho$ acting on a Hilbert space $\mcH$, fulfilling
$\varrho \geq 0$, $\rho=\rho^\dagger$ and $\tr(\varrho)=1$ is a density matrix.
The set of density matrices will be denoted along this work as $\mcS(\mcH)$. 

Comparing the notion of density
matrices with the notion of state vectors in the Hilbert space $\ket{\psi} \in \mcH$, density matrices describe physical systems where the observer has
an \textit{incomplete} knowledge of the system's state. Thus, while state vectors are
naturally equipped with \textit{intrinsic} or \textit{quantum} probabilities,
density operators are additionally equipped with \textit{classical}
probabilities. 
The density matrices enjoying the form $\rho=\proj{\psi}{\psi}$, or equivalently $\rho^2=\rho$, \ie{} projectors, are \textit{pure states}. It is clear that in this case the system is prepared in the state $\ket{\psi}$ with probability one.

A useful quantity to characterize quantum states is the \textit{purity}, defined as 
\begin{equation}
P(\rho)=\tr\left( \rho^2\right).
\end{equation}
It ranges from $\dim(\mcH)^{-1}$ to $1$; $1$ is obtained for pure
states and $\dim(\mcH)^{-1}$ for the complete mixture $\one /\dim(\mcH)$.

Additionally the set $\mcS$ is convex, \ie{} any convex combination of density
matrices is another density matrix, in the same way as classical distributions
do. In fact, mixed states ($P(\rho)<1$) can be written always as convex
combinations of pure states, see~\eref{eq:density_matrix_def}. Furthermore the
set $\mcS(\mcH)$ is a subset of the bigger set of \textit{trace-class}
operators, $\mcT(\mcH)$, defined as the ones containing operators with finite
trace norm. The latter is defined as $|\Delta|_\text{tr}=\tr \sqrt{A^\dagger
A}$. This set is in turn a subset of the set of bounded operators $\mcB(\mcH)$,
containing operators with finite \textit{operator norm}, defined as
$|A|_\text{op}=\sup_{\ket{\psi}} |A\ket{\psi}| $, where
$|A\ket{\psi}|=\sqrt(\langle A\psi | A\psi \rangle)$, \ie{} the standard
Hilbert space norm, with normalized vectors $\ket{\psi}$.

It is worth to note that for the finite
dimensional case, bounded operators always have finite trace norm and vice versa, thus
$\mcT(\mcH)=\mcB(\mcH)$. But the identification of such sets is relevant for
infinite dimensional systems, where counter-examples of the non-equivalence of
such sets exist~\cite{zimansbook}. Additionally $\mcB(\mcH)$ is the dual space
of $\mcT(\mcH)$ under the Hilbert-Schmidt product, defined as $\langle
A,B\rangle=\tr(A^\dagger B)$~\cite{Holevobook}.

Now, to ignore the degrees of freedom of the unaccessible part of the system,
we have to perform an operation in a very analogous way as computing marginal
distributions in classical probability theory. For density operators this
introduces the concept of partial trace.

\paragraph{Partial trace.} Let $\rho\in\mcS(\mcH_\text{A}\otimes\mcH_\text{B})$
and $\mcH_\text{A,B}$ the Hilbert spaces of systems $A$ and $B$. Thus, $\rho$
describes a state of a bipartite system composed by $A$ and $B$. If we want to
know the state of the system $A$ alone, one performs a partial trace over $B$
defined as 
\begin{equation*}
\rho_A=\tr_B (\rho_{AB}) = 
\sum_{i=1}^{d_\text{B}}
   \left(\one\otimes\bra{\phi^\text{B}_i}\right) \rho_{AB} \left(\one\otimes\ket{\phi^\text{B}_i}\right),
\end{equation*}
where $\lbrace \ket{\phi^\text{B}_i} \rbrace_{i=1}^{d_\text{B}}$ is  a complete orthonormal
basis on $\mcH_\text{B}$. The resulting operator $\rho_\text{A}$ is a density
matrix describing the state of the system $A$ alone. It is trivial to show that
it is a density operator. A similar formula holds for $\rho_\text{B}$. An
alternative definition is $\tr_B \left(A\otimes B \right)=A\tr\left(B\right)$
plus linearity.

In general for composite systems, in a pure state, knowing the reduced states
(for instance for bipartite systems, $\rho_A$ and $\rho_B$) is in general not enough to
know the whole state of a system. This captures the
\textit{non-local} nature of quantum correlations, demanding simultaneous
measurements on both parts of the system. 
In such case we say that the
subsystems $A$ and $B$ are entangled. To see this, consider the example of the
Bell state $\ket{\Omega}=1/\sqrt{2} \left( \ket{00}+\ket{11}\right)$, where
$\lbrace \ket{0},\ket{1} \rbrace$ is an orthogonal basis of a qubit system. It
is trivial to show that $\ket{\Omega}$ cannot be written as $\ket{\phi}\otimes
\ket{\psi}$, a factorizable state, prohibiting the observer to know the state
of the whole system only by non-simultaneous measurements on $A$ and $B$
(described by reduced density matrices). In fact it is easy to show that
$\rho_{A,B}=\one/2$ are the reduced density matrices, appearing also when the
total state is $\rho_{AB}=\one/4$. For composite systems in mixed states the
situation is quite different. In this case simultaneous measurements are needed
to access classical correlations. To see this consider the state
\begin{equation}
\rho_\text{AB}=\sum_i p_i \rho^i_\text{A} \otimes \rho^i_\text{B},
\label{eq:mixed_convex_comp}
\end{equation} 
being a convex combination of factorizable mixed states. This state is a mixed
separable state~\cite{Horodecki}, \ie{} subsystems A and B are not entangled.
Notice now that performing only local non-simultaneous measurements, the
accessible reduced states are $\rho'_{A,B}=\sum p_i \rho^i_{A,B}$. This state
also arises when the total system is in the factorizable state
$\rho'_\text{A}\otimes \rho'_\text{B}$. Therefore local simultaneous
measurements are needed.


\subsection{System plus reservoir dynamics} 
The most widely used scheme to study
open quantum systems is to consider a bipartite system, where the central
system S, is interacting with its environment, E. The full system
$\text{S}+\text{E}$ undergoes a closed system evolution, \ie{} Hamiltonian
dynamics, see \fref{fig:open_scheme}. The \textit{total Hamiltonian} $H$,
describing the whole system, has the following general structure
\begin{equation}
H=H_\text{S} +H_\text{E}+V,
\end{equation}
where $H_\text{S,E}$ are the free Hamiltonians of the central system and the
environment, respectively, and $V$ is the interaction Hamiltonian among them.
Now let $\rho_\text{SE}(0)$ be the state of the total system at the time $t=0$.
Thus, the state of the system S at the time $t$ is simply:
\begin{equation}
\rho_\text{S}(t)=\tr_\text{E} \left(U(t) \rho_\text{SE}(0)U^\dagger (t) \right),
\label{eq:global_unitary_no_factorized_state}
\end{equation}
where $U(t)=e^{-\rmi H t}$ (taking $\hbar=1$) and $\tr_\text{E}$ is the partial
trace over the environmental degrees of freedom. Note that for a general
initial state $\rho_\text{SE}(0)$, where one allows classical and quantum
correlations, $\rho_\text{SE}(t)$ depends in general on initial information
about the environment and its correlations with the central system $S$. Thus,
to compute the dynamics of the central system such that we end up to universal
reduced dynamics, \ie{} the same for every initial state and independent of the
initial information in the environment, we  take a factorized initial state
$\rho_\text{SE}(0)=\rho_\text{S}(0)\otimes\rho_\text{E}$~\cite{breuerbook,rivasbook}.
We do not write explicitly the time-dependence of the environmental state since
one is not usually interested on its evolution. With the choice of a
factorizable total initial state and using
equation~\eref{eq:global_unitary_no_factorized_state}, we have the following
expression for the evolution of the central system,
\begin{equation}
\rho_\text{S}(t)=\tr_\text{E} \left[ U(t) \left(\rho_\text{S}(0) \otimes \rho_\text{E} \right) U^\dagger (t)\right].
\label{eq:rho_t_open_systems}
\end{equation}
Therefore we have that the dynamics over S only depends on the total Hamiltonian $H$ and the environmental initial state $\rho_E$, whereas $\rho_\text{S}(t)$ depends only on its initial condition. 
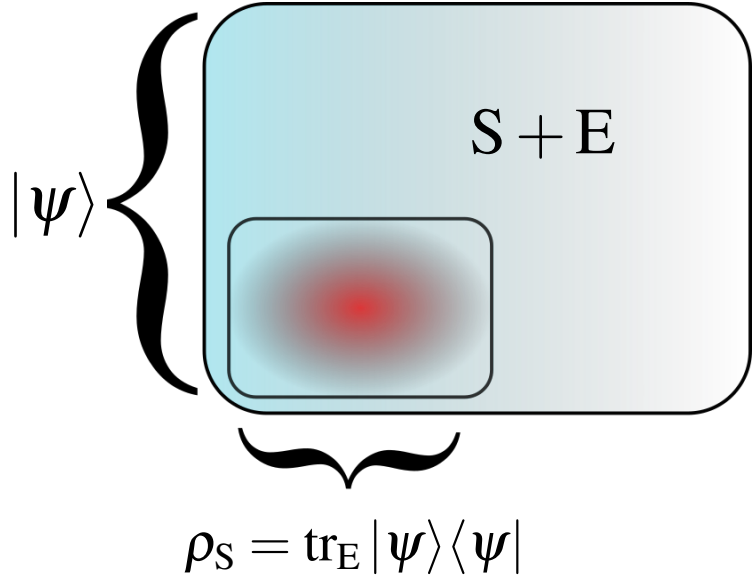
\begin{figure}
\centering
\begin{tikzpicture}
\node at (0,0) {\includegraphics[draft=false,scale=0.13]{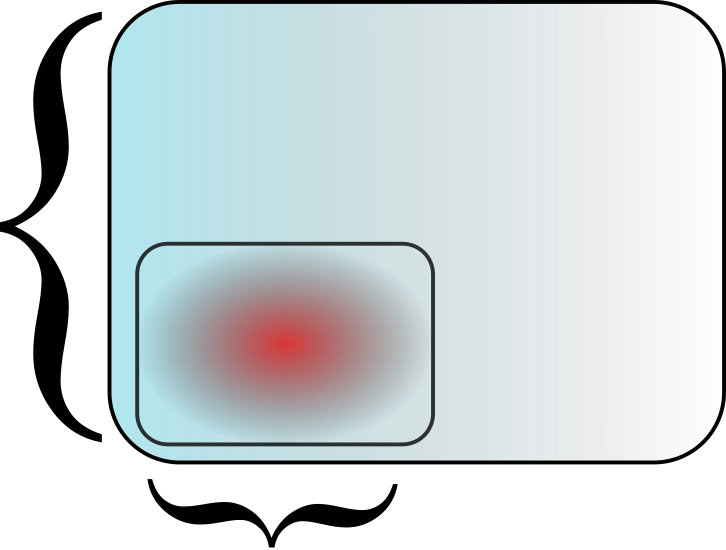}};
\node at (1.5,1.5) {\Huge $\text{S}+\text{E}$};
\node at (-5,0.55) {\Huge $\ket{\psi}$};
\node at (-1,-4) {\huge $\rho_\text{S}=\tr_\text{E} \proj{\psi}{\psi}$};
\end{tikzpicture}
\caption{Diagram of the scheme to study open quantum systems. The letters S and
E state for the \textit{open (or central) system} and \textit{environment}
parts of the total closed system, $\text{S}+\text{E}$. The latter  is described
(typically) by a pure state $\ket{\psi}\in\mcH_\text{S}\otimes \mcH_\text{E}$
and the central system is described by the reduced state computed using the
partial trace over the environmental degrees of freedom, see main text.
\label{fig:open_scheme}}
\end{figure}

Hence the equation~\eref{eq:rho_t_open_systems} defines a dynamical map,
$\mcE_t$, parametrized by $t$. \fxnote{Francois:no parece tener
sentido: aclarar
}\fxwarning{David: frase sin sentido removida}. Thus, we have
\begin{equation}
\mcE_t[\rho(0)]=\tr_\text{E} \left[ U(t) \left(\rho_\text{S}(0) \otimes \rho_\text{E} \right) U^\dagger (t)\right].
\label{eq:open_systems_map}
\end{equation} 
Such map possesses all the information concerning the dynamics of the system S,
thus knowing $\mcE_t$ one can know entirely the evolution of the system $S$.
The map $\mcE_t$ can be obtained numerically or experimentally (depending on
the context) by measuring only the system S by \textit{quantum process
tomography}~\cite{chuangbook}. In~\fref{fig:dilatation_diagram} we present a
schematic description of the two equivalent schemes under which the system S
evolves, and their connection throughout $\tr_\text{E}$.
\begin{figure} 
\centering
\begin{tikzcd}
\rho_\text{S}(0) \otimes \rho_\text{E} \arrow{d}{\tr_\text{E} (\cdot)} \arrow{r}{U(t) \cdot U^{\dagger}(t)}
& \rho_\text{SE}(t) \arrow{d}{\tr_\text{E}(\cdot)} \\
\rho_\text{S}(0) \arrow{r}{\mcE_t}
& \rho_\text{S}(t)
\end{tikzcd}
\caption{Scheme of the equivalences between the concept of dynamical map and
the theory of open quantum systems.
\label{fig:dilatation_diagram}}
\end{figure} 

\Eref{eq:open_systems_map} can be reduced, by writing $\rho_\text{E}=\sum_j p^\text{E}_j \ket{\phi^\text{E}_j}\bra{\phi^\text{E}_j}$, in the following way,
\begin{equation}
\mcE_t[\rho_\text{S}(0)]=\sum_{i,j} K(t)_{i,j} \rho_\text{S}(0) K(t)^\dagger_{i,j},
\label{eq:kraus_sum_time_dependent}
\end{equation}
where the operators $K(t)_{ij}=\sqrt{p^\text{E}_j} \bra{\phi^\text{E}_i}
U(t)\ket{\phi_j^\text{E}}$ are called Kraus operators and act upon the system S
alone~\cite{rivasbook}. The expression of \eref{eq:kraus_sum_time_dependent} is
called \textit{sum represention}, also called Kraus representation of the map $\mcE_t$, this will be
retaken on chapter~\ref{chap:reps}.

Now let us discuss the differential equation for the density matrix of an open quantum system. The total state of the system evolves according to the \textit{Von Neumann equation}~\cite{breuerbook},
\begin{equation}
\frac{d\rho_\text{SE}}{dt}=-\rmi [H,\rho_\text{SE}],
\label{eq:von_neumann}
\end{equation}
which is the analog of the Liouville equation describing the evolution
of a classical distribution in the phase space. 

Taking the partial trace on both sides of \eref{eq:von_neumann} one arrives to
the following:
\begin{align}
\frac{d\rho_\text{S}}{dt}&=-\rmi \tr_\text{E}[H,\rho_\text{SE}]\nonumber\\
&=L_t [\rho_\text{S}],
\label{eq:von_neumann_reduced}
\end{align}
where $L_t$ is the generator of the master equation of the system S. Integrating time in both sides from $\tau=0$ to $\tau=t$, we arrive to the equivalent integral equation:
\begin{equation}
\rho_\text{S}(t)=\rho_\text{S}(0)+\int_0^t d\tau L_\tau [\rho_\text{S}(t)].
\label{eq:integral_equation}
\end{equation}
To compute the formal solution of this equation, we use the method of
successive approximations. This consists on substituting the whole expression
for $\rho_\text{S}(t)$ defined by the right hand side of 
\eref{eq:integral_equation}. A first iteration leads to 
\begin{equation}
\rho_\text{S}(t)
  =\rho_\text{S}(0)+\int_0^t d\tau_1 L_{\tau_1}[\rho_\text{S}(0)]
                   +\int_0^t d\tau_1\int_0^t d\tau_2 L_{\tau_1}[L_{\tau_2}[\rho_\text{S}(t)]].
\label{eq:second_iteration_Neumann}
\end{equation}
Repeating this procedure infinite times, \ie{} substituting $\rho_\text{S}(t)$
defined by the right hand side of the last equation in its second integrand several times,
we arrive to a power series solution for
$\rho_\text{S}(t)$ (powers of $L_t$). This leads to the well known Dyson series for $L_t$. Compactly,
\begin{equation}
\rho(t)=\vec{\text{T}} \exp\left(\int_0^t ds L_s\right)\rho(0)
\label{eq:formal_solution_master}
\end{equation}
with $\vec{\text{T}}$ the time-ordering operator, defined as
$$\vec{\text{T}}[H(\tau_1)H(\tau_2)]=\theta(\tau_1-\tau_2)H(\tau_1)H(\tau_2)+
\theta(\tau_2-\tau_1)H(\tau_2)H(\tau_1),$$ 
with $\theta(x)$ the Heaviside step
function. \Eref{eq:formal_solution_master} constitutes the formal solution
to the Von Neumann equation with generator $L_t$, and we can easily identify
$\mcE_t=\vec{\text{T}} \exp\left(\int_0^t ds L_s\right)$.

\subsection{Born-Markov approach: microscopic derivation} 
In general the form of the generator $L_t$, given a global Hamiltonian, can be
quite involved~\cite{breuerbook}, but in the limit of \textit{weak coupling}
and \textit{short memory} we can perform the very well known
\textit{Born-Markov approximation}. A brief discussion is presented in this
subsection.

The Born-Markov approximation leads to the \textit{Lindblad master equation}.
We will briefly overview its usual textbook derivation. The first step is to use the
interaction picture, hence the total Hamiltonian becomes
$H_I(t)=e^{\rmi H_0 t} H e^{-\rmi H_0 t}$, where $H_0=H_\text{S}+H_\text{E}$ is
the \textit{free} Hamiltonian.
Assuming that the dimension of $\mcH_\text{E}$ is big compared with the
dimension of $\mcH_\text{S}$, the weak coupling limit leads to negligible
changes in the environmental state. Thus, at time $t$ we can approximate
$$
\rho_\text{SE}(t)\approx \rho_\text{S}(t)\otimes \rho_E.
$$
In other words, the state of the total system is left always approximately uncorrelated,
while the state of the environment is never updated. Therefore the environment 
\textit{forgets} any information about the central system, while the state of
the latter undergoes a non-trivial evolution. Additionally to simplify
the derivation we choose $\rho_\text{E}$ a stationary state of $H_\text{E}$,
\ie{} $[H_\text{E},\rho_\text{E}]=0$~\cite{rivasbook}. $\rho_\text{E}$  is typically
chosen as a thermal state of the environmental Hamiltonian, $\rho_\text{E}
\propto \exp\left(-\beta H_\text{E}\right)$, with $\beta=1/\left(k_\text{B}
T\right)$, $k_\text{B}$ the Boltzmann constant and $T$ the environment
temperature.

Now, in the interaction picture the Von Neumann equation becomes
\begin{equation}
\frac{d\rho_\text{S}}{dt}=-\rmi \tr_\text{E} [V_\text{I}(t),\rho_\text{S}],
\label{eq:von_neumann_interaction_picture}
\end{equation}
where $V_\text{I}(t)=e^{\rmi H_0 t} V e^{-\rmi H_0 t}$ and the state
$\rho_\text{S}(t)$ are now written in the interaction picture. Inserting
$\rho_\text{S}(t)$ from its integral equation \eref{eq:integral_equation}
in 
the differential equation~(\ref{eq:von_neumann_interaction_picture}) and
assuming $\tr_\text{E}[V_\text{I}(t),\rho_\text{S}\otimes
\rho_\text{E}]=0$~\cite{breuerbook}, we obtain
\begin{equation}
\frac{d\rho_\text{S}}{dt}=-\int_0^t d\tau \tr_\text{E} [V_\text{I}(t),[V_\text{I}(\tau),\rho_\text{S}(\tau)\otimes\rho_\text{E}]].
\end{equation}
If we assume that the dynamics of the state of the central system does not
depend on its past, we can change $\rho_\text{S}(\tau)$ to $\rho_\text{S}(t)$,
this is called the \textit{Markovian approximation}. 
Additionally doing the
variable change $\tau'=t-\tau$, we arrive to
\begin{equation}
\frac{d\rho_\text{S}}{dt}=-\int_0^t d\tau' \tr_\text{E} [V_\text{I}(t),[V_\text{I}(t-\tau'),\rho_\text{S}(t)\otimes\rho_\text{E}]],
\end{equation}
this equation is known as Redfield equation~\cite{Redfield1965} and it is local
in time~\cite{breuerbook}. Assuming that the time scale on which the central
system varies appreciably is much larger than the time on which the
correlations of the environment decay (say $\tau_\text{E}$), the integrand
decays to zero rapidly for $\tau'\gg \tau_\text{E}$. Then we can safely replace
$t$ by $\infty$ in the integrand limits, obtaining
\begin{equation}
\frac{d\rho_\text{S}}{dt}=-\int_0^\infty d\tau' \tr_\text{E} [V_\text{I}(t),[V_\text{I}(t-\tau'),\rho_\text{S}(t)\otimes\rho_\text{E}]].
\label{eq:born_markov}
\end{equation}

Up to this point, \eref{eq:born_markov} has in general fast oscillating
terms coming from the explicit dependence on $V_\text{I}(t)$, this in turn can
bring a generator that leads to a quantum process that violates complete
positivity~\cite{rivasreview,rivasbook}. In order to get rid of such fast
oscillations, one uses the aforementioned assumption that the environment is
initialized in a stationary state, and perform the so
called \textit{secular approximation}~\cite{rivasreview}. A detailed derivation is outside of the scope of
this thesis, but it can be consulted on references~\cite{breuerbook,rivasbook}.
After performing the Markov, Born and secular approximations and changing back
to the \scho{} picture, the resulting master equation can be written in the
following forms
\begin{align}
\frac{d\rho_\text{S}}{dt}&= \rmi [\rho_\text{S},\tilde H_\text{S}]
 +\sum_{i,j=1}^{d_\text{S}^2-1} G_{ij} 
     \left( 
         F_{i}\rho_\text{S} F^{\dagger}_{j}
             -\frac{1}{2} \lbrace F^{\dagger}_{j} F_{i},\rho_\text{S} \rbrace 
     \right),\\
     &= \rmi [\rho_\text{S},\tilde H_\text{S}]
 +\sum_{j=1}^{d_\text{S}^2-1} \gamma_j
     \left( 
         A_{j}\rho_\text{S} A^{\dagger}_{j}
             -\frac{1}{2} \lbrace A^{\dagger}_{j} A_{j},\rho_\text{S} \rbrace 
     \right),\\
     &=L[\rho_S].
\label{eq:lindblad_from_deri}
\end{align}
$F_j$ ($j=0,\cdots,d_\text{S}^2-1$) are operators acting on the central system
that additionally form an orthonormal basis under Hilbert-Schmidt inner
product, such that $F_{0}=\one/\sqrt{d_\text{S}}$ and $\tr F_j=0$ $\forall j >0$ (this will be revised in subsection~\ref{sec:herm_and_trace_less});
the matrix $G$ is called dissipator matrix. In the second inequality we have
used the singular value decomposition of matrix $G$, thus operators $A_i$ are
linear combinations of $F_i$. The scalars $\gamma_j>0$ are called
\textit{relaxation rates} and the operator $\tilde H_\text{S}$ is the
shifted free Hamiltonian of the central system. The first term on both
equations, the commutator, is called \textit{Hamiltonian part}, while the
second, the superoperator defined with the summations, is called
\textit{dissipator}. Note that if $\gamma_j=0$ $\forall j$ (uncoupled limit),
one recovers the Hamiltonian dynamics over the system S. The operator $L$ is 
called \textit{Lindblad generator}
or \textit{Lindbladian} and \eref{eq:lindblad_from_deri} is called
\textit{Lindblad master equation}. We will use along the work the notation
$L$ for Lindblad operators. 

Note that $L$ is independent of time, hence the formal solution of the master
equation~\eref{eq:lindblad_from_deri} equation is simply the exponentiation of
$L$ [see~\eref{eq:formal_solution_master}], \ie{} 
\begin{equation}
\rho_\text{S}(t)=e^{Lt}\rho_\text{S}(0).
\end{equation}
Therefore the dynamics is homogeneous in time and, together with the fact that
$\mcE_t=\exp(L t)$, we have $\mcE_{t+s}=\mcE_t\mcE_s$, \ie{} the quantum
process $\mcE_t$ resulting from a Lindblad master equation forms a
one-parameter semigroup. In fact, Lindblad has proven the converse for norm continuous semigroups~\cite{lindblad}.  Here we write the theorem for the finite dimensional case that is trivially norm continuous,
\begin{theorem}[One-parameter quantum semigroups]
Let $\mcE_t$ with $\mcE_0=\id$ and $t\geq 0$ a finite dimensional quantum process, it is a one-parameter quantum semigroup if and only if it has a generator with the form presented in~\eref{eq:lindblad_from_deri}.
\end{theorem}
A proof is given in Ref.~\cite{Alicki2007}.
It is worth to point out that starting from global dynamics governed by a finite dimensional Hamiltonian, the reduced dynamics are never of Lindblad form. This can be stated as the following,
\begin{theorem}[Exact dynamics with Lindblad master equation]
Let $\mcE_t=e^{tL}$ a quantum process generated by a Lindblad operator $L$. The equation 
$$\mcE_t[\rho]=\tr_\text{E} \left[ e^{-\rmi H t} \left(\rho \otimes \rho_\text{E}\right) e^{\rmi H t}\right],$$ 
where $H$ has finite dimension, holds if and only if $\mcE_t$ is an unitary conjugation for every $t$.
\end{theorem}
A proof made jointly with Sergey Filippov is given in the~appendix~\ref{sec:unbounded}. It was made using an specific matrix representation for operators that will be introduced in the next chapter. But a more general proof can be found in~\cite{exner}.

Let us point out that this is not the case for Hamiltonians with continuum spectrum, they can lead to Lindblad master equations for the reduced dynamics. This is shown below together other illustrative examples.

\paragraph{Examples.}
To illustrate Lindblad dynamics we present several examples. The first one,
depolarizing dynamics, is constructed via a continuous and
monotonic contraction of the Bloch sphere. The second one corresponds to a system for which the exact reduced
dynamics have Lindblad generator.
\begin{example}[Dephasing dynamics]
Let $\rho(0)=\left( \begin{array}{cc}
\rho_{00} & \rho_{01} \\ 
\rho_{01}^* & \rho_{11}
\end{array}  \right)$ be the initial state, written in a basis called
\textit{decoherence basis}, of a system that undergoes depolarizing dynamics. This is, only coherence terms (in this basis) are modified in the following way:
$$\mcE_t:\rho(0)\mapsto \left( \begin{array}{cc}
\rho_{00} & \rho_{01} e^{-\gamma t} \\ 
\rho_{01}^* e^{-\gamma t} & \rho_{11}
\end{array}  \right)=:\rho(t),$$
with $\gamma>0$. It is trivial to check that $\mcE_t$ is a one-parameter
semigroup with $\mcE_0=\id$. For $t \to \infty$, we get
$\rho(0)\to\text{diag}\left( \rho_{00},\rho_{11} \right)$. For this process it
is easy to prove, by taking $0< t \ll 1$, that its generator is $L[\rho]=\gamma/2
\left( \sigma_z \rho \sigma_z -\rho \right)$, which has Lindblad form. It has
null Hamiltonian part and only one operator $A_0=\sigma_z$ and one relaxation
ration, $\gamma/2$.
\label{example:dephasing}
\end{example}

\begin{example}[Dynamics from global Hamiltonian with continuous spectrum]
Consider a bipartite system composed by a qubit interacting with a particle in a
line, with global Hamiltonian $H= \sigma_z \otimes \hat x$, where
$\hat x$ is position operator. Notice that $H$ is unbounded since the configuration space of the particle
is the entire real line.
Initializing the environment in the state $\ket{\psi}$ with 
\begin{equation*}
\langle x |\psi \rangle=\sqrt{\frac{\gamma}{\pi}}\frac{1}{x+\rmi \gamma},
\end{equation*}
it can be shown that the exact reduced dynamics for the qubit, without any
approximation, is $L[\rho]=\gamma/2 \left( \sigma_z \rho \sigma_z -\rho
\right)$ ~\cite{exactlindblad}. The same generator as in the first example.
\end{example}

\section{Quantum channels} 
In this section we give a brief introduction to classical stochastic processes,
this motivates the definition of quantum channel. We first give an overview of stochastic processes; based
on this we review the construction steps of quantum channels and discuss
several of their properties. Additionally we introduce the simplest example of
local operations and classical communication. Later on one we discuss the
definition of CP-divisible processes based on the definition of classical
Markovianity. Finally we give a brief revision of Gaussian quantum states and
channels.
\subsection{A classical analog} 
\label{sec:classical_analog}
The classical analog of quantum channels are the widely known \textit{stochastic
matrices} or \textit{stochastic maps} which propagate classical probability
distributions.  
To introduce them consider, for sake of simplicity, a finite
dimensional stochastic system whose state $x_t$ (at time $t$) is described by
the probability distribution (or probability vector) $\vec p(t)$, \ie{}
$x_t\sim\vec p(t)$ [with $\sum_i p_i (t)=1$ and $p_i(t)\geq 0$].
Note that probability vectors form a convex
space in the very same way that density matrices do.
The distribution $\vec p(t)$ is the classical analogous object to density matrices. They serve as the tool to model the accessible information of the observer about 
the state of the classical stochastic system.

Consider now the most general linear transformation on probability vectors that
takes, for instance $\vec p (0)$ to $\vec p(t)$ and let us write it explicitly
as a matrix multiplication, $\vec p(t)=\Lambda_{(t,0)}\vec p(0)$. We have to
impose further constrictions over $\Lambda_{(t,0)}$ in order to preserve the
normalization of $\vec p(t)$ and the non-negativity of its elements. Since $p_i
(t)=\sum_j \left(\Lambda_{(t,0)}\right)_{ij} \vec p_j (0)$, simple algebra
leads us to note that $\sum_{i}\left(\Lambda_{(t,0)}\right)_{ij}=1\ \ \forall
j$ and $\left(\Lambda_{(t,0)}\right)_{ij}\geq 0$. Matrices that fulfill these
conditions are widely known as \textit{stochastic matrices}, and form a
convex set following the convexity of the space of probability distributions.

A remarkable property of stochastic maps is that they are contractive with respect to the Kolmogorov distance.
\begin{theorem}[Contractivity of stochastic maps]
The matrix $\Lambda$ is a stochastic matrix if and only if 
\begin{equation}
\mcD_\text{K}\left(\Lambda \vec p, \Lambda \vec q \right)\leq \mcD_\text{K}\left(\vec p,\vec q\right),
\end{equation}
where $\mcD_\text{K}\left(\vec p, \vec q \right)=\sum_k |p_k-q_k|$ is the
Kolmogorov distance.
\label{thm:classical_contractivity}
\end{theorem}
It is worth to note Kolmogorov distance is a measure of distinguishability between classical distributions. A detailed proof of this theorem can be found in Ref.~\cite{rivasreview}. 

A particular and interesting class of stochastic matrices are
\textit{bistochastic matrices}. They are defined as the transformations that
leave invariant the probability distribution with maximum entropy, given by
$\vec m=\left(1/N,\dots,1/N\right)^\text{T}$, where $N$ is the number that the system can have. Therefore a bistochastic matrix fulfills $\vec m=\Lambda_{(t,0)}\vec
m$. Doing simple algebra leads us to note that bistochastic matrices
additionally fulfill $\sum_{j}\left(\Lambda_{(t,0)}\right)_{ij}=1\ \ \forall
i$. This implies that they are also stochastic matrices \textit{acting from the
right}, \ie{} mapping row probability vectors. This is also the origin of the
name bistochastic.

In the previous section we have introduced the concept of Markovianity in the
context of open quantum systems, the so called Markovian approximation. It
consisted on assuming that the system 'forgets the information about its
previous states'. This concept comes from the theory of classical stochastic
processes. Let us introduce the following
definition~\cite{breuerbook,rivasreview},

\begin{definition}[Classical Markovian process]
Let $x_t$ be the state of a stochastic system where $t\in [0,\tau]$, and
$\chi=\lbrace t_0, \dots t_n \rbrace$ any ordered set of times such that
$0<t_0< t_1 <\dots <t_n < \tau$, the process is Markovian if
\begin{equation}
P(x_{t_n},t_n|x_{t_{n-1}},t_{n-1};\dots;x_{t_0},t_0)=P(x_{t_n},t_n|x_{t_{n-1}},t_{n-1}) \ \ \forall n>0,
\label{eq:classical_markov}
\end{equation}
where $P(\cdot | \cdot)$ denotes conditional probability.
\end{definition}
According to this definition, the conditional probability of the system to be
at the state $x_{t_{n}}$ at the time $t_n$, given the history of events
$\lbrace x_{t_{n-1}},t_{n-1};\dots;x_{t_0},t_0\rbrace$, depends only on the
previous state. This definition captures the memoryless character of Markovian
processes.

Consider now a stochastic process and $\lbrace\Lambda_{(t,0)}\rbrace_{t\in
\chi}$ a set of stochastic matrices given some ordered set of times $\chi$.
If the process is Markovian then the matrices $\Lambda_{(t_m,t_n)}$ are
stochastic matrices for any $\chi$, where $t_m>t_n \in \chi$. 
The converse is
not true~\cite{rivasreview,breuerbook}. This condition 
implies that the map
$\Lambda_{(t,0)}$ is divisible in the sense that it can always be written as 
\begin{equation}
\Lambda_{(t_1,t_0)}=\Lambda_{(t_1,s)}\Lambda_{(s,t_0)} \ \ \forall \ \ t_1>s>t_0,
\label{eq:classical_divisibility}
\end{equation}
with $\Lambda_{(t_1,s)}$, $\Lambda_{(s,t_0)}$ and $\Lambda_{(t_1,t_0)}$
stochastic matrices, the latter two by definition. Intermediate maps can be
constructed as $\Lambda_{(t_1,s)}=\Lambda_{(t_1,t_0)}\Lambda^{-1}_{(s,t_0)}$ if
$\Lambda^{-1}_{(s,t_0)}$ exists. Note that
theorem~\ref{thm:classical_contractivity} implies that Markovian stochastic
processes do not increase the Kolmogorov distance.
\subsection{Construction of quantum channels} 
The concept of quantum channel, also known as quantum operation, captures the
idea of stochastic map in the quantum setting. Thus, being the density matrices
the analogous objects to probability vectors, we seek for linear operations
that transform density matrices into density matrices. The operations that do
such job are defined as follows:
\begin{definition}[Positive and trace preserving linear operations (PTP)]
A linear operation $\mcE:\mcT(\mcH)\to \mcT(\mcH)$ is positive and trace
preserving if, for all $\Delta \in \mcH$, we have the following
\begin{itemize}
\item $\mcE[\Delta]\geq 0 \ \ \forall \Delta \geq 0,$
\item $\tr \left(\mcE[\Delta]\right)=\tr \left( \Delta\right)$.
\end{itemize}
\end{definition}
A remarkable property of linear positive maps is that they are contractive
respect to the trace norm~\cite{rivasreview}. This leads to a decrease of
the distinguishability of quantum states, similar to the classical case.
\begin{theorem}[Contractivity of positive maps]
A linear map $\mcE$ is PTP if and only if
$\left|\mcE[\Delta]\right|_\text{tr}\leq \left| \Delta \right|_\text{tr} \ \ \forall \Delta^\dagger=\Delta \in \mcB(\mcH)$.
\label{thm:contractivity}
\end{theorem}
A simple proof for the finite dimensional case can be found in Ref.~\cite{rivasreview}.

Now, given that any hermitian operator can be written as 
\begin{align*}
\Delta &=\left(\tr\Delta\right) H_p\text{, for $\tr \Delta \neq 0$,}\\
\Delta &=\tr\Delta^+\left(\rho_1-\rho_2\right)\text{, for $\tr \Delta =0$,}
\end{align*}
where $H_p=p\rho_1-(1-p)\rho_2$ a
Helstrom matrix and $p\in [0,1]$, by theorem~\ref{thm:contractivity} the
generalized trace distance defined as $\mcD_p(\rho_1,\rho_2)=|p \rho_1
-(1-p)\rho_2|_\text{tr}$ decreases after the application of a positive map
$\mcE$, \ie{} 
$$\mcD_p\left(\mcE[\rho_1],\mcE[\rho_2]\right)\leq
\mcD_p\left(\rho_1,\rho_2 \right).$$
It is worth to point out that this is
directly related to the two-state discrimination problem
where we have, for
instance, probability $p$ of erroneously identify $\rho_1$ with
$\rho_2$~\cite{chuangbook,rivasreview}.
In this setting the probability of
failing with such identification is
\begin{equation*}
P_\text{err}=\frac{1-\mcD_p(\rho_1,\rho_2)}{2}.
\end{equation*}
Therefore if the distance is zero, the probability of correctly identify $\rho_1$ is the same as choosing randomly between $\rho_1$ and $\rho_2$, but if it is $1$, we identify $\rho_1$ from $\rho_2$ with certainty.
For $p=1/2$ we recover the standard unbiased trace distance.

It is well known that any quantum system can be entangled with another, for instance a central system can be entangled with its environment. Thus, in the context of quantum operations we must handle this fact carefully. Let us define the following:
\begin{definition}[$k$-positive operations]
A linear map $\mcE$ is $k$-positive if
$$\id_k\otimes \mcE [\tilde \Delta]\geq 0  \ \ \forall \tilde \Delta \geq 0 \in \mcB(\mcH_k\otimes\mcH),$$
with $k$ a positive integer, being the dimension of $\mcH_k$ and $\id_k$ the identity map in that space.
\end{definition}
Therefore a positive map is $k$-positive if the expended map $\id_k\otimes \mcE$ is positive, the trace preserving of $k$-positive maps follows immediately from the trace preserving of $\mcE$. Such maps transform properly density matrices of the extended system (with ancilla of dimension $k$) into density matrices, apart from the fact that they transform properly the density matrices of the system, hence handling quantum entanglement correctly for this ancilla. 

Since the dimension of any other quantum system is arbitrary, being for example the rest of the universe, one must have that quantum maps must transform quantum states for every positive integer $k$. Therefore one defines \textit{complete positive} and trace preserving linear maps as the following,
\begin{definition}[Complete positive and trace preserving operations (CPTP)]
A trace preserving linear operation $\mcE:\mcT(\mcH)\to\mcT(\mcH)$ is complete positive if
$$\id_k\otimes \mcE [\tilde \Delta]\geq 0  \ \ \forall \tilde \Delta \geq 0 \in \mcB(\mcH_k\otimes\mcH), \forall k \in \mathbb{Z}_0^+,$$
where $\mathbb{Z}^+$ is the set of the positive integers.
\end{definition}
It will be shown
later
in chapter~\ref{chap:reps}, that deciding complete positivity is
straightforward using the so called Choi matrix.

It is trivial to check that unitary operations, $\mcU[\rho]=U \rho U^\dagger$,
are CPTP maps as expected. Additionally they leave invariant the maximally
mixed state, $\one/\dim(\mcH)$. In fact, unitary operations belong to a wider
class of CPTP maps called \textit{unital quantum maps}, similar to its
classical counterpart. The set of unital channels is defined simply as the one
containing CPTP maps $\mcE$ that additionally fulfill $\mcE[\one]=\one$.

Additionally notice that due to the trace preserving property the adjoint operator of $\mcE$ is always unital. The adjoint is defined in the usual way, 
\begin{equation}
\langle A, \mcE[B] \rangle=\langle \mcE^*[A],B\rangle,
\label{eq:adjoint}
\end{equation} 
where the inner product is the Hilbert-Schmidt product and $A\in \mcT(\mcH)$
and $B\in \mcB(\mcH)$~\cite{Holevobook,zimansbook}. Now, $\forall \Delta\in
\mcT(\mcH)$   we write the trace preserving condition as $\tr \Delta=\tr
\mcE[\Delta]=\langle \one, \mcE[\Delta] \rangle=\langle
\mcE^*[\one],\Delta\rangle$, therefore $\mcE^*[\one]=\one$.

Let us now illustrate the connection of the concept of quantum channel with the scheme of open quantum systems introduced above. Consider the following theorem~\cite{stinespring}:
\begin{theorem}[Stinespring dilation theorem]
Let $\mcE$ a CPTP map, there exist an environmental Hilbert space $\mcH_\text{E}$ and $\rho_E\in \mcS(\mcH_E)$ such that
$$\mcE[\rho]=\tr_\text{E}\left[U \left(\rho \otimes \rho_\text{E}\right) U^\dagger \right],$$
with the unitary matrix $U:\mcH\otimes \mcH_\text{E}\to \mcH\otimes \mcH_\text{E}$.
\label{thm:stinespring}
\end{theorem}
The unitary $U$ and the state $\rho_E$ are not unique~\cite{zimansbook}.
Stinespring theorem is an important result given that one can always understand
a CPTP operation as a Hamiltonian evolution in a bigger space, such that we
recover the given operation at some fixed time and by performing a partial
trace over the environmental degrees of freedom. Later in this chapter we
will discuss an important implication of this theorem for Markovian processes.

Along the work we will also denote the set of CPTP linear maps simply as \cptp{}.

A remarkable property of \cptp{} is its convexity. To show this consider the
following convex combination of CPTP maps: $\mcE=p\mcE_1+(1-p)\mcE_2$, acting
upon the density matrix $\rho_0$. By linearity we have
$\mcE[\rho_0]=p\mcE_1[\rho_0]+(1-p)\mcE_2[\rho_0]$. Defining the density
matrices $\rho_i=\mcE_i[\rho_0]\in \mcS(\mcH)$, it follows from the convexity
of $\mcS(\mcH)$ that $\mcE$ is another CPTP map. Therefore the set \cptp{} is
convex.

Unitary maps
are extremal channels of \cptp{} i.e.
they cannot be written as
convex combinations of other channels, but they can be used to construct other
maps, see \fig{fig:scheme_cptp_slice}. For instance
consider a simple convex combination of unitary maps
$\mcE[\rho]=\sum_i p_i U\rho U^\dagger$, with $\sum_i p_i=1$ and $p_i\geq 0$.
This channel is a more general example of a unital channel, in fact it turns out that
every unital qubit channel has such form. This can be shown easily using the
Ruskai's decomposition that will be introduced in the next chapter. Convex
combinations of unitary channels can be implemented in the laboratory, for
instance choosing unitaries randomly by tossing a die. 
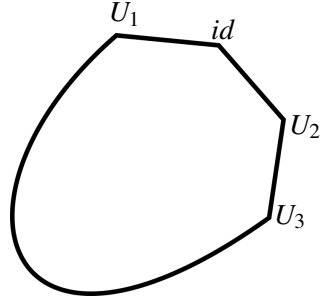
\begin{figure}
\centering
\begin{tikzpicture}
\node at (0,0) {\includegraphics[draft=false,scale=0.2]{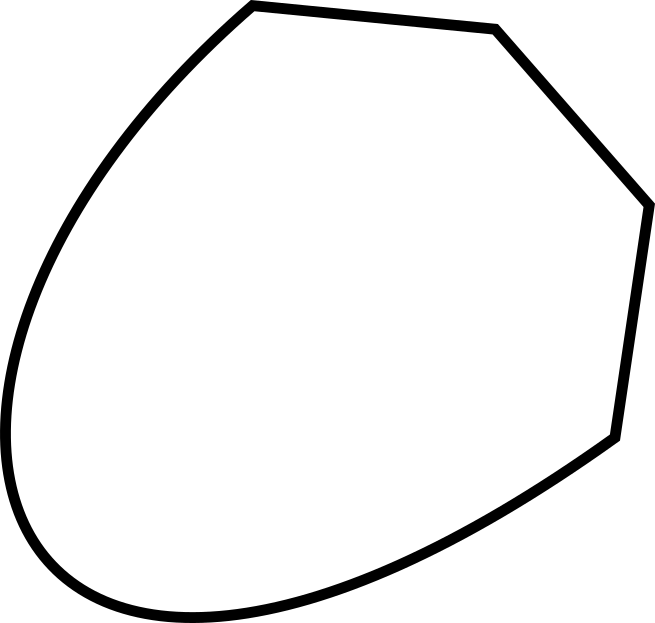}};
\node at (1,1.8) {$id$};
\node at (-0.3,2) {$U_1$};
\node at (2.1,0.5) {$U_2$};
\node at (1.9,-0.7) {$U_3$};
\end{tikzpicture}
\caption{The figure shows an schematic slice of CPTP maps, one can see the identity map and other extremal channels. The straight lines are convex combinations of those channels, the curve contains channels in the boundary that cannot be written as convex combinations of unitary channels.\label{fig:scheme_cptp_slice}}
\end{figure}

Regarding the algebraic properties of the set \cptp{}, it enjoys the structure
of a semigroup. It is closed under the composition operation, \ie{} $\mcE_1
\mcE_2\in \cptp{}$, $\forall \mcE_1,\mcE_2 \in \cptp{}$, and is associative,
$\left(\mcE_1 \mcE_2\right)\mcE_3=\mcE_1 \left(\mcE_2\mcE_3\right)$.
Additionally it contains an identity element. \cptp{} does not contain the
inverse elements, this captures the irreversible character of general quantum
operations, being only the unitaries the ones their inverse elements in
\cptp{}. 
Furthermore, \cptp{} contains another remarkable convex structure,
\fxnote{Francois:No es una frase completa. Debe decir: An entanglement breaking channel is...}\fxwarning{Editado, no quedo exactamente asi pero ya debe de tener sentido.}
\begin{definition}[Entanglement-breaking channels]
A map $\mcE\in\cptp{}$ is entanglement-breaking if it breaks the entanglement of the system with any ancilla, \ie{} $ \forall k\in \mathbb{Z}^+$ and $\forall \sigma\in \mcS(\mcH_k\otimes \mcH)$, the state $\left(\id_k \otimes \mcE\right) [\sigma]$ is separable.
\label{def:eb}
\end{definition} 
This set is convex given that convex combinations of separable states is separable~\cite{Horodecki}.

Quantum channels can be seen as the basic building of time-dependent quantum processes, also called quantum dynamical maps.
\begin{definition}[Quantum dynamical maps]
A continuous family of channels $\lbrace\mcE_t \in \cptp{}:t\geq 0 , \mcE_0=\id \rbrace$
is called quantum dynamical map.
\end{definition}
Given some interval $\mcI$, if the family is smooth respect to $t\in \mcI$ and invertible, it admits a master equation 
$$\dot \rho(t)=A_t[\rho(t)] \ \ \text{with} \ \ A_t=\dot \mcE_t \mcE_t^{-1}.$$
An schematic description is shown in~\fref{fig:dynamical_map_scheme}.
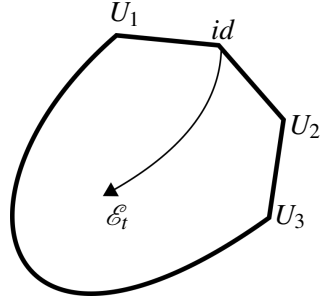
\begin{figure}
\centering
\begin{tikzpicture}
\node at (0,0) {\includegraphics[draft=false,scale=0.2]{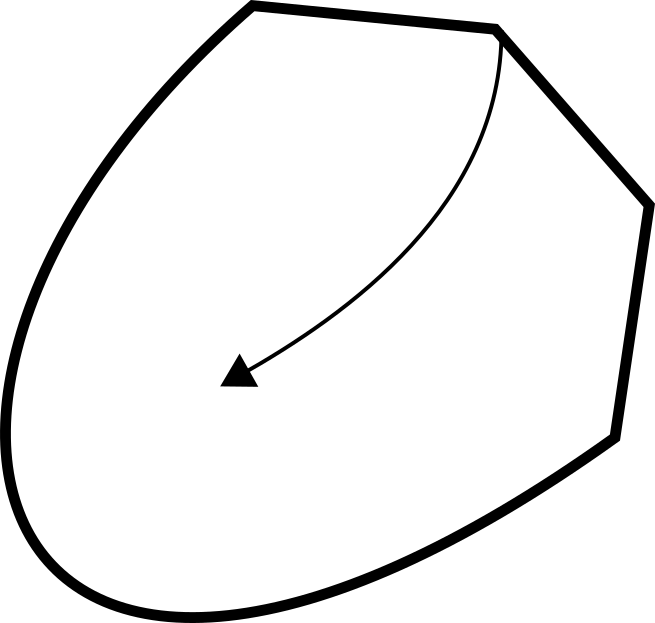}};
\node at (1,1.8) {$id$};
\node at (-0.3,2) {$U_1$};
\node at (2.1,0.5) {$U_2$};
\node at (1.9,-0.7) {$U_3$};
\node at (-0.4,-0.7) {$\mcE_t$};
\end{tikzpicture}
\caption{Scheme of an smooth dynamical map inside a slice of the set \cptp{}.\label{fig:dynamical_map_scheme}}
\end{figure}
Note that the standard scheme open quantum systems, introduced at the beginning
of this chapter, leads to quantum dynamical maps.
\subsection{Non-Linear CPTP operations} 
Notice that the set \cptp{} does not contain everything that can be performed
on a quantum system; it contains only linear operations. Therefore \cptp{} does
not contain  \textit{postseletion} procedures, \ie{} updating the state once a
measurement is done and the result is known. For instance, let $\rho$ the state
of some system and $\lbrace M_i \rbrace$ a collection of \textit{measurement
operators} over it, where the index $i$ refers to the measurement outcome. The
probability of measuring $i$ is $p(i)=\tr \left(M_i \rho
M^{\dagger}_i\right)$, while the operation performed over the state is
\begin{equation*}
\rho \mapsto \frac{M_i \rho  M^{\dagger}_i}{\tr \left(M_i \rho  M^{\dagger}_i\right)}.
\end{equation*}
This operation is explicitly non-linear but it is trivially complete positive
and trace preserving. Note that if the action of the measurement apparatus is
performed but the experimentalist does not read the outcome, or it is simply
forgotten, the resulting map belongs to \cptp{}~\cite{chuangbook}. This is
shown by noting that the operation $M_i\rho M^\dagger_i$ is applied with
probability $p(i)$, then the performed operation is $\sum_i p(i) M_i\rho
M_i^\dagger$ and it is linear and CPTP by construction. Complete positivity
follows immediately from the complete positivity of $\rho \mapsto M_i \rho M_i^\dagger$ and
the trace preserving property from the weighted summation.


A more general set of operations including measurements, postselection and
exchange of classical information will be introduced in the next subsection.


\subsection{Local operations and classical communication} 

Several types of quantum operations can be found and studied, in particular in
Ref.~\cite{zimansbook} there is a classification mainly based on its
\textit{locality}. A paradigmatic and widely studied type are the so called
\textit{local operations and classical communication}~\cite{Horodecki}.  A
surprising feature of these operations is that they can increase the
entanglement of entangled states of a system (at the cost of throwing away some
members of the ensemble), but cannot create them from non-entangled
ones~\cite{Verstraete2001,Horodecki}.

In this work we are particularly interested in \textit{one-way stochastic local
operations and classical communication} channels (1wSLOCC). 
Consider a
bipartite system where one part is controlled by Alice and the other by Bob.
Alice performs an operation which includes measurements with postselection,
and then she communicates its outcome to Bob. Then Bob performs a local operation
that can be again a measurement with postselection, finishing the protocol. The
stochasticity comes from the fact that this operation, for each particular set
of measurement outcomes, has a certain probability generally less than $1$ of
occurrence. And the one-way comes from the fact that no feedback is given to
Alice and no more operations and classical communications are performed. These
operations can be written in the following way:
\begin{equation}
\rho\mapsto \rho'=
  \frac{\left(X\otimes Y\right) \rho \left(X\otimes Y\right)^\dagger
    }{
    \tr \left[ 
           \left(X\otimes Y\right) \rho \left(X\otimes Y\right)^\dagger \right]}.
\label{eq:slocc}
\end{equation}
Additionally we will consider $\det X\neq0$ and $\det Y\neq 0$, this is the
usual choice as projective measurements destroy
entanglement~\cite{Verstraete2001}.  These operations
are completely positive and trace preserving, but non-linear
unless $X$ and $Y$ are unitaries. Additional notice that given $\rho$ and
$\rho'$, the matrices $X$ and $Y$ can always be chosen such that $\det X=\det
Y=1$ (for the invertible case). Therefore for two-level systems it is enough to consider $X,Y \in \text{SL}(2,\mathbb{C})$~\cite{wuki},
where the latter is the \textit{special linear group} of $2\times 2$ matrices
with complex entries. Furthermore notice that the operation 
\begin{equation}
\rho \mapsto\left(X\otimes Y\right) \rho \left(X\otimes Y\right)^\dagger
\label{eq:det_press}
\end{equation}
preserves the determinant, \ie{} $\det \rho=\det \rho'$. In the next chapter we
will exploit this to show that there is a correspondence between 1wSLOCC and
Lorentz transformations. We use this to introduce a decomposition analogous to
the singular value decomposition, but using the Lorentz metric instead of the
Euclidean, enjoying an useful physical meaning.
\section{Quantum channels of continuous variable systems} 
Many of the definitions and tools introduced in the previous sections are also
relevant for the infinite dimensional case. Although we can always choose countable
basis for the Hilbert space as long it is separable~\cite{zimansbook}, it is
often of interest to consider non-countable bases, typically phase-space
variables. This introduces the theory of continuous variable systems.
It is a central topic of study given that
they appear naturally in the description of many physical systems. A 
few examples are the electromagnetic
field~\cite{doi:10.1142/p489}, solids and nano-mechanical
systems~\cite{RevModPhys.86.1391} and atomic
ensembles~\cite{RevModPhys.82.1041}.  
In particular, in this section we introduce and discuss a set of continuous variable channels called \textit{Gaussian quantum channels}. 
\subsection{Gaussian quantum states} 
To introduce the definition of Gaussian quantum channel, consider first the
simplest state type of quantum states in continuous variable, both from a
theoretical and experimental point of view, the so-called \textit{Gaussian
states}.  The operations that transform such family of states into itself are
called 
Gaussian quantum channels (\gqc{}).
Even though Gaussian states and
channels are small subsets of all possible states/channels, they 
have proven to be useful in a very wide variate of tasks such as 
quantum communication~\cite{Grosshans2003}, quantum computation~\cite{PhysRevLett.82.1784}
and the study 
of quantum entanglement in simple~\cite{RevModPhys.77.513} and complicated
scenarios~\cite{PhysRevA.98.022335}.

\fxnote{Francois:no está muy claro. Podría ser útil una fórmula. Cuál es la cara del Gaussian state que da lugar a la matriz de correlación que describes?
}\fxwarning{David:Editado}
{\blue Gaussian states are defined as those having Gaussian Wigner function. In
particular, for one-mode the Wigner function is 
\begin{equation}
W\left( \vec u \right)=\frac{1}{2 \pi \sqrt{\det \sigma}}e^{-\frac{1}{2}\left( \vec u -\vec d \right)^\text{T}\sigma^{-1}\left( \vec u -\vec d \right)},
\end{equation}
where $\vec u=\left( q,p \right)^\text{T}$~\cite{Cerf}. The \textit{mean vector} $\vec d$ and the \textit{covariance matrix} $\sigma$ are the first and second moments, respectively. They are given by
\begin{align*}
\sigma&=\left( \begin{array}{cc}
\langle \hat q^2 \rangle-\langle \hat q \rangle^2 & \frac{1}{2}\langle \hat q \hat p+\hat p \hat q\rangle -\langle \hat q \rangle \langle \hat p \rangle \\ 
\frac{1}{2}\langle \hat q \hat p+\hat p \hat q \rangle -\langle \hat q \rangle \langle \hat p \rangle  & \langle \hat p^2 \rangle -\langle \hat p \rangle^2
\end{array} \right),\\
\vec d &=\left(\langle \hat q \rangle, \langle \hat p \rangle
\right)^\text{T}.
\end{align*}
The observables $\hat q$ and $\hat p$ are the standard
canonical conjugate position and momentum variables. As for any other Gaussian variable, Gaussian quantum states are characterized completely
by first and second probabilistic moments. Therefore a Gaussian state $S$
can be denoted as $S=S\left(\sigma,\vec d\right)$.
}
\subsection{Gaussian quantum channels} 
To start with, we recall the following definition~\cite{Reviewquantuminfo}:
\begin{definition}[Gaussian quantum channels]
A quantum channel is Gaussian (\gqc{}) if it transforms Gaussian
states into Gaussian states.
\label{def:gqc}
\end{definition}
This definition is strictly equivalent to the statement that any \gqc{}, say $\mcG$, can be written as
\begin{equation}
\mcG[\rho]=\tr_\text{E} \left[ U \left(\rho \otimes \rho_\text{E} \right) U^{\dagger} \right]
\label{def:def_2}
\end{equation}
where $U$ is a unitary transformation, acting on a combined global state
obtained from enlarging the system with an  environment $\text{E}$, that is
generated by a quadratic bosonic Hamiltonian (\ie{} $U$ is a \textit{Gaussian
unitary})~\cite{Reviewquantuminfo}. The environmental initial state $\rho_E$ is
a Gaussian state and the trace is taken over the environmental
degrees of freedom.

Following definition~\ref{def:gqc}, a GQC is fully characterized by its action
over Gaussian states, and this action is in turn defined by \textit{affine
transformations}~\cite{Reviewquantuminfo}. Specifically, $\mcG=\mcG\left(
\mathbf{T}, \mathbf{N},
\vec \tau \right)$ is
given by a tuple $\left( \mathbf{T}, \mathbf{N},
\vec \tau \right)$ where $\mathbf{T}$ and $\mathbf{N}$ are $2\times 2$ real
matrices with $\mathbf{N}=\mathbf{N}^\text{T}$~\cite{Reviewquantuminfo} acting on Gaussian states according to  
$$\mcG\left( \mathbf{T},
\mathbf{N}, \vec \tau
\right)\left[S\left(\sigma,\vec d\right)\right]=S\left( \mathbf{T}\sigma
\mathbf{T}^{\text{T}} +\mathbf{N}, \mathbf{T} \vec d +\vec \tau \right).$$
In
the particular case of closed systems we have $\mathbf{N}=\mathbf{0}$ and $\mathbf{T}$ is a
symplectic matrix.
The particular form and properties of Gaussian quantum channels in the continuous variable representations, as well as their connection with the mentioned affine transformations, will be given in chapter~\ref{chap:reps}.

In this work we explore \gqc{}s {\blue without Gaussian functional form in the position state representation.}\fxnote{Francois: Gaussian beyond Gaussian?!??
}\fxwarning{David: Nos referimos a formas funcionales Gaussianas, edito para que este mas claro} In particular we
study channels that can arise when singularities on the coefficients of
Gaussian forms \gf{} occur (they will be denoted by \dgqc{}).  Such channels
can lead immediately to singular Gaussian operations.  Thus, we characterize
which forms in \dgqc{} lead to valid quantum channels, and under which
conditions singular operations lead to valid \textit{singular Gaussian quantum channels}
(\sgqc{}).

Let us note that although channels with Gaussian form trivially transform
Gaussian states into Gaussian states, the definition
goes beyond \gf{}. We will use
the typical \textit{difference} and \textit{sum} coordinates, $x=q_2-q_1$ and
$r=(q_1+q_2)/2$, respectively. Defining $\rho(x,r)=\left.\left\langle
r-\frac{x}{2} \right. \right| \left. \hat \rho \left|  r+\frac{x}{2} \right.
\right\rangle$, a quantum channel in this representation is defined such that
\begin{equation}
\rho_f\left(x_f,r_f\right)
   =\int_{\mathbb{R}^2} dx_i dr_i J(x_f,x_i;r_f,r_i)\rho_i\left(x_i,r_i\right),
\label{eq:propagacion}
\end{equation}
where $\hat \rho_i$ and $\hat \rho_f$ are the initial and final states,
respectively, and $J(x_f,x_i;r_f,r_i)$ is the representation of the quantum
channel in the aforementioned variables.  An example of a channel without \gf{}
can be constructed from the general form of Gaussian quantum channel with
\gf{}~\cite{PazSupplementary}:
\begin{multline}
J_\text{G}(x_f,x_i;r_f,r_i)
   =\frac{b_3}{2 \pi} \exp\Big[\Imi\Big( 
     b_1 x_f r_f  +b_2x_f r_i  +b_3x_ir_f \\  +b_4x_ir_i +c_1x_f+c_2x_i \Big) 
          -a_1 x_f^2-a_2x_fx_i-a_3x_i^2 \Big],
\label{eq:gf}
\end{multline}
where all coefficients are real and no quadratic terms in $r_{i,f}$ are
allowed. 
Choosing 
$$a_n=\alpha_n \epsilon^{-1} + \tilde a_n $$
and 
$$b_n=\beta_n\epsilon^{-1/2}+\tilde b_n,$$
 with $\epsilon>0$, $\alpha_n,\beta_n,\tilde a_n, \tilde b_n \in \mathbb{R} \ \ \forall n$ and $\tilde b_3=0$. Taking the limit $\epsilon\to 0$ and using the formula
\begin{equation}
\delta(x)=\lim_{\epsilon\to 0} \frac{1}{2 \sqrt{\pi \epsilon}}{e^{\frac{-x^2}{4 \epsilon}}},
\label{eq:limit_delta}
\end{equation}
we arrive to 
\begin{equation}
\lim_{\epsilon \to 0}
J_\text{G}(x_f,x_i;r_f,r_i) =
\mcN \delta(\alpha x_f-\beta x_i) e^{\Sigma'(x_f,x_i;r_f,r_i)},
\label{eq:typeII}
\end{equation}
where $\alpha$, $\beta \in \mathbb{R}$ and $\Sigma'(x_f,x_i;r_f,r_i)$ is  a
quadratic form that now admits quadratic terms in $r_{i,f}$, arising from the
completion of the square of the exponent of~\eref{eq:gf} to take the limit
of~\eref{eq:limit_delta}.  This is the first example of a \dgqc{}.  This
channel is still a \gqc{} according to the definition. A physical, but
complicated realization occurs in the system of one Brownian quantum particle
with harmonic potential and linearly coupled to the bath. In such system,
channels with the functional form of \eref{eq:typeII} are realized at isolated
points in time, see equations 6.71-75 of Ref.~\cite{Gert}.

Since the form of \eref{eq:typeII} admits quadratic terms in $r_{i,f}$ in the
exponent, it suggest that a form with two deltas exist and can be defined using
the same limit, see~\eref{eq:limit_delta}. In fact, the identity map is a
particular case; it is realized setting
$J(x_f,x_i;r_f,r_i)=\delta(x_f-x_i)\delta(r_f-r_i)$. In any case, to
avoid working with such limits, it is convenient to perform a
\textit{black-box} characterization of general forms involving Dirac's deltas,
which will be done in the next chapter. This will lead to explicit relations
between position state representation and affine representations of Gaussian
quantum channels without Gaussian functional form.


\chapter{Representations of quantum channels}
\begin{flushright}
\textit{Simplicity is the ultimate sophistication.
}\\
Leonardo da Vinci
\end{flushright}

In this chapter we introduce several and useful representations of quantum
channels for the finite dimensional case. We start with the Kraus
representation, already mentioned in the previous chapter, but additionally we will show that
quantum channels always have this form. Later on we introduce Choi's theorem
(and the so called \Jami{} representation) which is cornerstone tool to study
many properties of quantum channels. We also discuss operational
representations by introducing two types of basis. These representations are useful to prove
several results in this work. Next, we apply the introduced tools to the
qubit case. Additionally we discuss two decomposition of qubit channels,
leading to two \textit{normal forms} that are essential to study divisibility
properties of quantum channels.
\label{chap:reps} 
\section{Kraus representation} 
In the previous chapter we have shown that starting from the usual scheme of
open quantum systems, we arrive to the \textit{Kraus representation}, see
\eref{eq:kraus_sum_time_dependent}. Later on, using the Stinespring dilation
theorem, see Theorem~\ref{thm:stinespring}, we showed that CPTP maps can always fit in
the scheme of open quantum systems for some global unitary evolution. Since the
latter scheme always has a Kraus representation, one concludes that CPTP maps
always have a Kraus representation. It turns out that the converse also holds~\cite{kraus}.
\begin{theorem}[Kraus]
A linear operation $\mcE:\mcT(\mcH) \to \mcT(\mcH)$ belongs to \cptp{}
if and only if there exist a set of bounded operators $\lbrace K_i \rbrace$
such that 
$$\mcE[\Delta]=\sum_i K_i \Delta K^\dagger{}_i \ \  \forall \Delta \in \mcT(\mcH),$$
with $\sum_i K_i^\dagger{} K_i=\one$.
\label{thm:kraus}
\end{theorem}
\begin{proof}
The 'only if' part is already commented in the main text and follows
the logic: every $\mcE\in \cptp{}$ has a dilation such that it has
the familiar form of the open quantum systems dynamics, \ie{} there exists
$U$ and $\rho_\text{E}$ such that 
$\mcE[\rho]=\tr_\text{E}\left[U \left(\rho \otimes \rho_\text{E}\right)
U^\dagger \right]$. We already showed that writing $\rho_\text{E}$ in terms of
its eigenbasis, the latter expression leads to the Kraus representation,
see~\eref{eq:kraus_sum_time_dependent}. To prove the 'if' part, we only have
to construct the extended map to test its complete positivity. Let $k>0\in \mathbb{Z}$ and $\tau_k=\left(
\id_k \otimes \mcE \right)[ \tilde \Delta_k]$, where $\tilde \Delta_k \in
\mcB(\mcH_k\otimes \mcH)$ and $\tilde \Delta_k \geq 0$, using Kraus decomposition and
evaluating $\bra{\phi} \tau_k \ket {\phi}$ with $\ket{\phi}\in \mcH_k \otimes
\mcH$, one arrives to
\begin{align*}
\bra{\phi}\tau_k \ket{\phi}
   &=\sum_i \bra{\phi}
      \left(\one_k \otimes K_i\right) \tilde \Delta_k \left(\one \otimes K_i^\dagger{} \right)\ket{\phi}\\
   &=\sum_i \bra{\phi_i} \tilde \Delta_k \ket{\phi_i}\\
   &\geq 0.
\end{align*}
The latter follows immediately from the positive-semidefinitiveness of $\tilde
\Delta_k$, \ie{} $\bra{\phi_i} \tilde \Delta_k \ket{\phi_i} \geq 0$. The
condition $\sum_i K_i^\dagger{} K_i=\one$ comes from the trace-preserving
of $\mcE$ and the cyclic property of the trace,
\begin{align*}
\tr \mcE[\Delta]&=\sum_i \tr \left[ K_i \Delta K_i^\dagger{}  \right]\\
&=\sum_i \tr \left[ K_i^\dagger{} K_i \Delta \right]\\
&=\tr\left[ \left( \sum_i K_i^\dagger{} K_i \right) \Delta \right]\\
&= \tr \Delta,
\end{align*}
\end{proof}
Therefore $\sum_i K_i^\dagger{} K_i=\one$. It is worth to note that Kraus operators are not unique. Defining a new set of
operators, $A_k=\sum_l u_{k l} K_l$, it is
easy to show that $\sum_i K_i \Delta K_i^\dagger{}=\sum_k A_k \Delta
A_k^\dagger{}$ if and only if $u_{kl}$ are the components of an unitary matrix.
Therefore different Kraus representations are related by unitary conjugations.
\section{\Jami{} representation} 

The \Jami{} representation arises as part of a very useful theorem in quantum
information theory, the so called Choi's theorem~\cite{choi,zimansbook}.
\begin{theorem}[Choi]
Let $\mcE:\mathbb{C}^{n\times n}\to \mathbb{C}^{m\times m}$ be a linear map. The following statements are 
equivalent: \\
i) $\mcE$ is $n-$positive.\\
ii) The matrix
$$C_\mcE=\sum_{i,j=1}^{n} \ket{\varphi_i}\bra{\varphi_j}\otimes \mcE[\ket{\varphi_i}\bra{\varphi_j}] \in \mathbb{C}^{n\times m} \otimes \mathbb{C}^{n\times m}$$
\fxnote{Francois: no es acaso $\mathbb{C}^{n\times n} \otimes \mathbb{C}^{m\times m}$  ? Claro que es isomorfo.}\fxwarning{Si es isomorfo, sin embargo dado que despues descomponemos justamente $\mathbb{C}^{n \times m}$, decidi dejar el producto de espacios como está. Adicionalmente, así como está, sugiere un espacio de matrices cuadradas, $C_\mcE$ es justamente cuadrada.}
is positive-semidefinite with $\lbrace \ket{\varphi_i}\rbrace_{i=1}^n $ an orthonormal basis in $\mathbb{C}^{n}$.\\
iii) $\mcE$ is completely positive.\\
\end{theorem}
\begin{proof}
The proof of $iii)\to i)$ is trivial, if $\mcE$ is completely positive then it
is $n-$positive. The implication $i)\to ii)$ can be proved easily by noticing
that normalizing $C_\mcE\to C_\mcE/n=:\tau_\mcE$, where $\tau_\mcE$ can be
obtained as the application $\tau_\mcE=\left(\id_n \otimes
\mcE\right)[\omega]$, where $\omega=\proj{\Omega}{\Omega}$ with
$\ket{\Omega}=1/\sqrt{n}\sum_i^{n}\ket{\varphi_i}\otimes \ket{\varphi_i}$ a
Bell state between two copies of $\mathbb{C}^n$. Therefore, by the
$n-$positivity of $\mcE$ it follows that $\tau_\mcE$ is positive-semidefinite.

What remains to prove is $ii)\to iii)$. To do this observe that the space
$\mathbb{C}^{n\times m}$ is isomorphic to the direct sum of $n$ copies of
$\mathbb{C}^m$, \ie{} $\mathbb{C}^{n\times m}\cong \mathbb{C}^m_1\oplus
\mathbb{C}^m_2\oplus \dots \oplus \mathbb{C}^m_n$\fxwarning{Gracias por notar que tenia $\otimes$ en lugar de $\oplus$ !}, and define the projector
into the $k$th copy as $P_k=\bra{\varphi_k}\otimes \one$, such that $P_k C_\mcE
P_l=\mcE[\proj{\varphi_k}{\varphi_l}]$. Now, given that $C_\mcE$ is
positive-semidefinite, it can be written as $C_\mcE=\sum_i^{nm}
\proj{\Psi_i}{\Psi_i}$, where $\ket{\Psi_i}\in \mathbb{C}^{n\times m}$ are
generally unnormalized vectors. Thus, we have that
$\mcE[\proj{\varphi_k}{\varphi_l}]=\sum_i P_k\proj{\Psi_i}{\Psi_i}P_l$, where
$P_k\ket{\Psi_i}\in \mathbb{C}^m_k$. Defining the operators $\lbrace
K_i:\mathbb{C}^n\to \mathbb{C}^m\rbrace_i$ via the equation $P_k \ket{\Psi_i}=K_i
\ket{\varphi_k}$, where choosing for example $\ket{\varphi_k}$ as the canonical
basis, the columns of $K_i$ contain the $n$ projections of $\ket{\Psi_i}$ into
the copies of $\mathbb{C}^m$. Finally we arrive to
$\mcE[\proj{\varphi_k}{\varphi_l}]=\sum_i K_i \proj{\varphi_k}{\varphi_l} K_i
^\dagger{} \ \ \forall k,l=1,\dots,n$. In conclusion, since $\lbrace \proj{\varphi_k}{\varphi_l} \rbrace_{k,l}$ is a complete basis of $\mathbb{C}^{n\times n}$, by linearity and by
theorem~\ref{thm:kraus}, the map $\mcE$ is completely positive.
\end{proof}
The matrix $C_\mcE$ is commonly known as Choi matrix and $\tau_\mcE$ as \Jami{}
state. Both define the \Jami{} representation, in this work labeled as
$\tau_\mcE$ since it is normalized.

Choi's theorem provides a simple test of complete positivity, which I find beautiful. 
For instance, if we want to know if a given PTP map $\mcE$ is a valid quantum map, we just have to
consider two copies of our system in only one state, the Bell state, then apply
$\mcE$ to one of the copies and check if the result, $\tau_\mcE=\left(\id_n
\otimes \mcE\right)[\omega]$, is a density matrix.

The \Jami{} representation enjoys other useful properties, if $\mcE$ preserves the trace, the matrix 
\begin{equation}
\tau_\mcE=\frac{1}{n}\left(
\begin{array}{ccc}
\mcE[\proj{\varphi_1}{\varphi_1}] & \dots & \mcE[\proj{\varphi_1}{\varphi_n}] \\ 
\vdots & \ddots & \vdots \\ 
\mcE[\proj{\varphi_n}{\varphi_1}] & \dots & \mcE[\proj{\varphi_n}{\varphi_n}]
\end{array} 
\right)
\label{eqn:jami_block}
\end{equation}
has blocks of trace $1/n$ and $0$, since $\tr \mcE[\proj{\varphi_i}{\varphi_j}]=\delta_{ij}$. 
This property additionally means that not every density matrix in $\mathbb{C}^{n\times m}\otimes \mathbb{C}^{n\times m}$ has a corresponding CPTP map.

The matrix rank of $\tau_\mcE$ coincides with the so called \textit{Kraus rank}, \ie{} the number of linearly independent Kraus operators required to write the channel. This can be shown easily noticing that computing $\tau_\mcE$ from the Kraus sum, one arrives to the equality $\ket{\Psi_i}=1/\sqrt{n} \one \otimes K_i \ket{\Omega}$, therefore the linear independence of $\lbrace \ket{\Psi_i} \rbrace_i$ follows immediately from the linear independence of $\lbrace K_i\rbrace_i$.
Therefore the maximum Kraus rank is $mn$ and the minimum $1$. Channels with Kraus rank equal to 1 are trivially unitary channels given that $\mcE[\Delta]=K \Delta K^\dagger{}$ with $K^\dagger{} K=\one$. Channels with the maximum rank are called \textit{full Kraus rank channels}.

Another interesting property is that if $\tau_\mcE$ is {\blue separable (i.e. not entangled)} \fxwarning{Francois: qué quiere decir separable?} \fxwarning{David: no entrelazado, lo edite para que quede mas claro}, then $\mcE$ is entanglement-breaking, see definition~\ref{def:eb}. For qubit channels it is enough to test that the concurrence is zero~\cite{Rybar2012}.

\section{Operational representations} 
It has been shown that the \Jami{} representation is useful to test several
properties of quantum channels. In this section we will introduce other
representations, this time with operational meanings. They are basically
operator basis that give matrix and vector forms to channels and density
matrices, respectively.

The vectorization of density matrices can be achieved
simply ``making them flat'', this is,
$$\left( \begin{array}{ccc}
\rho_{11} & \dots & \rho_{1d} \\ 
\vdots & \ddots & \vdots \\ 
\rho_{d1} & \dots & \rho_{dd}
\end{array}  \right)\mapsto \left( \begin{array}{c}
\rho_{11} \\ 
\rho_{12} \\ 
\vdots \\ 
\rho_{dd}
\end{array}  \right)=:\vec \rho.$$
Using this mapping, the matrix form of operators acting on $\mcT(\mcH)$ is build using the simple rule~\cite{vector} 
\begin{equation}
A\rho B \mapsto \left(A\otimes B^\text{T}\right) \vec \rho.
\label{eq:flatten}
\end{equation}
For instance applying this rule to a commutator, $[H,\rho]\mapsto
\left(H\otimes \one-\one\otimes H^\text{T}\right)\vec \rho$. This
representation is useful to prove various results involving operators acting on
the space of density matrices, see for instance the
appendix~\ref{sec:unbounded}.
Additionally it is simple to prove that the Hilbert-Schmidt inner product is
mapped to $\langle \gamma, \rho \rangle \mapsto \vec \gamma^\dagger{} \vec
\rho$.


One can use other operator basis accordingly to our purposes. In general we
have the following, consider $\lbrace A_i \rbrace_i$ an orthonormal operator
basis in the space $\mcT(\mcH)$, the components of the density matrix are 
$$\alpha_i = \langle A_i,\rho\rangle=\tr \left[A_i^\dagger{} \rho\right],$$ 
so 
$$\rho=\sum_i \alpha_i A_i.$$ 
Correspondingly, the components of operators acting on $\mcB(\mcH)$, for
instance $\mcE$, are simply $$\hat \mcE_{ij}=\langle A_i,\mcE[A_j]
\rangle=\tr\left[ A_i^\dagger{} \mcE[A_j]\right].$$ Using this equation it is
easy to prove that the representation of the adjoint operator of $\mcE$,
see~\eref{eq:adjoint}, is simply $\hat \mcE^*=\hat \mcE^\dagger{}$.

\subsection{Hermitian and traceless basis} 
\label{sec:herm_and_trace_less}
Two types of basis are specially useful in this work, the first one are the
hermitian basis. This is, every orthonormal basis $\lbrace A_i \rbrace_i$ that
fulfills $A_i=A_i^\dagger{}, \ \forall i$. To show the utility of this kind of
basis, let us introduce the following definition,
\begin{definition}[Hermiticity preserving operators]
A linear operator $\mcE:\mcB(\mcH)\to \mcB(\mcH)$ preserves hermiticity if
$$\mcE[\Delta]^\dagger{}=\mcE[\Delta^\dagger{}], \ \forall \Delta\in \mcB(\mcH).$$
\end{definition}
Using the Kraus representation is trivial to prove that linear CPTP maps
preserve hermiticity, using complete positivity.
Furthermore, hermiticity preserving maps enjoy an hermitian \Jami{}
representation, \ie{} $\tau_\mcE=\tau_\mcE^\dagger{}$~\cite{wolfnotes}.

Using an hermitian basis it is straightforward to prove the following,
\begin{proposition}[Representation with real entries]
Let $\mcE$ be a linear and hermiticity preserving map. Its matrix
representation using an hermitian basis $\lbrace A_i \rbrace$ has real entries.
\end{proposition}
\begin{proof}
Let $\overline{\hat \mcE_{ij}}=\overline{\tr\left[ A_i \mcE[A_j] \right]}$, where the line over denotes complex conjugation. Distributing the latter inside the argument of the trace and using the hermiticity of $A_i$, we get $\overline{\hat \mcE_{ij}}=\tr\left[ A_i \mcE[A_j]^\dagger{}\right]$, finally stressing that $\mcE[A_j]^\dagger{}=\mcE[A_j^\dagger{}]=\mcE[A_j]$, we arrive to $\overline{\hat \mcE_{ij}}=\hat \mcE_{ij}$. 
\end{proof}
This simple property will be used later to prove the equivalence of the problem of finding channels that can be written as $\mcE=\exp(L)$, with $L$ a Lindblad generator.

The second useful type of basis are the so called traceless bases. They are
defined as follows. Let $\lbrace F_i\rbrace_{i=0}^{d^2-1}$ be an orthogonal basis, where we have
indicated the dimension of the space $\mcT(\mcH)$ as $d^2$ with
$d=\dim(\mcH)$, it is traceless if $F_{0}=\one/\sqrt{n}$ and $\tr F_i=0 \ \
\forall i> 0$. The traceless property comes from the fact that only one
element has non-zero trace, it is easy to prove that it must be proportional to
the identity, given that one can write the identity matrix using such basis.

This basis is useful to prove that generators of quantum dynamical maps, $L_t$, defined with $\mcE_{(t+\epsilon,t)}[\rho]=\rho+\epsilon L_t[\rho]+\mcO(\epsilon^2)$, have the following specific structure.
\begin{theorem}[Specific form of generators of dynamical maps]
Let $L:\mcT(\mcH)\to\mcT(\mcH)$ be a linear operator fulfilling $L[\Delta]^\dagger{}=L[\Delta^\dagger{}]$ and $\tr\left[L[\Delta]\right]=0$ (or equivalently $L^*[\one]=0$), then it has the following form, 
\begin{equation}
L[\rho]= \rmi [\rho,H]
 +\sum_{i,j=1}^{d^2-1} G_{ij} 
     \left( 
         F_{i}\rho F^{\dagger}_{j}
             -\frac{1}{2} \lbrace F^{\dagger}_{j} F_{i},\rho \rbrace 
     \right),
\label{eq:lindblad_from_theorem_with_ccp}
\end{equation}
where $d=\dim(\mcH)$, $H\in \mathbb{C}^{d\times d}$ and $G\in \mathbb{C}^{\left(d^2-1\right)\times \left(d^2-1\right)}$ are hermitian, and $\lbrace F_i\rbrace_{i=0}^{d^2-1}$ is an orthonormal traceless basis of $\mcB(\mcH)$.
\label{thm:linblad_simple_form}
\end{theorem}
Notice that Lindblad generators enjoy such form with the
additional condition that $G\geq 0$, see~\eref{eq:lindblad_from_deri}. A proof
of this is given in Ref.~\cite{Evans1977} for the infinite dimensional case
using technicalities beyond this work. Here we will prove it for the finite
dimensional case, using the notation of an incomplete proof given in
Ref.~\cite{Wolf2008}.
\begin{proof}
Since $L$ preserves hermiticity, it has an hermitian \Jami{} matrix, $\tau_L
\in \mathbb{C}^{d^2\times d^2}$. We can write such matrix always as
\begin{equation}
\tau_L=\tau_\phi-\proj{\Psi}{\Omega}-\proj{\Omega}{\Psi},
\label{eq:general_hermitian_form}
\end{equation}
where $\ket{\Psi}=-\omega_{\perp}\tau_L\ket{\Omega}-\left(\lambda/2\right) 
\ket{\Omega}$, $\lambda=\bra{\Omega}\tau_L\ket{\Omega}$, 
$\omega_\perp \tau_{L} \omega_\perp=\omega_{\perp}\tau_\phi\omega_{\perp}=\tau_\phi$ and $\omega_\perp = \one - \omega$.  Observe that choosing the traceless operator basis $\lbrace F_i\rbrace_{i=0}^{d^2-1}$, it is simple to prove that the matrix $\tau_\phi$ can be understood also as the \Jami{} matrix of the following operator:
\begin{equation}
\phi\left[\rho\right]=\sum_{i,j=1}^{d^2-1} G_{i j} F_{i}\rho F^{\dagger}_{j},
\label{eq:phi_super}
\end{equation}
with $G$ hermitian, \ie{} $\tau_\phi=\left(\id_{d^2}\otimes \phi\right)[\omega]$. This can be shown noticing that the summation $\sum_{i,j=1}^{d^2-1}$ goes over only traceless operators, therefore the projections into the one-dimensional space of $\ket{\Omega}$ of the Choi matrix of $\phi$ are null,
\[ \omega \tau_\phi
   = \sum_{i,j=1}^{d^2-1}G_{ij}\frac{1}{d} \text{tr} 
     \left(F_{i}\right)\omega \left(\one\otimes F^{\dagger}_{j}\right)
   = 0 \; , \]
and similarly for
\[ \tau_\phi\omega=0
   \; .  \] 
For the second and third terms of Eq.~\ref{eq:general_hermitian_form}, it is easy to show that the corresponding operator is simply 
$\rho \mapsto-\kappa \rho -\rho \kappa^\dagger$, where we identify $\ket{\Psi}=\left(\one\otimes \kappa \right)\ket{\Omega}$.

Up to now we have shown that hermiticity preserving generators have the form 
\begin{equation}
L\left[\rho\right]=\phi\left[\rho\right]-\kappa\rho -\rho \kappa^{\dagger}.
\label{eq:hermiticity_preserving_simple}
\end{equation}
Using the condition $L^*[\one]=0$, we have that 
$$\kappa+\kappa^\dagger=\phi^*[\one],$$
\ie{} the hermitian part of $\kappa$ is given by $\frac{1}{2} \sum_{i,j=1}^{d^2-1} G_{i j} F_j^\dagger F_{i}$. Simply writing the antihermitian part as $i H$ we end up with
$$\kappa=\rmi H+ 
\frac{1}{2} \sum_{i,j=1}^{d^2-1} G_{i j} F_j^\dagger F_{i}.$$
Substituting this expression and~\eref{eq:phi_super}, in~\eref{eq:hermiticity_preserving_simple}, we arrive to the desired form, see~\ref{eq:lindblad_from_theorem}.
\end{proof}

Notice that the operator $\phi$ is completely positive if and only if $G\geq 0$~\cite{zimansbook}, thus, $G\geq 0 \Leftrightarrow \tau_\phi\geq 0$. In such case $L$ has exactly a Lindblad form. This condition will be introduced later as \textit{conditional complete positivity}~\cite{Evans1977,Wolf2008}. 

The following is a central and useful result for our work.
\begin{proposition}[Conditional complete positivity]
An hermiticity preserving linear operator $L:\mcT(\mcH)\to\mcT(\mcH)$ fulfilling $\tr \left[L^*[\one]\right]=0$, has Lindblad form if and only if
$$\omega_\perp \tau_L \omega_\perp\geq 0.$$
\label{prop:ccp}
\end{proposition}

Additionally choosing an arbitrary basis of the Hilbert space to write operators $\lbrace F_i\rbrace_{i=1}^{d^2}$, it is easy to prove that $G$ and $\omega_\perp \tau_L \omega_\perp$ are related by an unitary conjugation~\cite{montesgorin}.

\section{Qubit channels} 
We shall devote some time to the most simple but non-trivial
quantum system, the qubit. This case turns out to be rich enough to
use and test the tools provided by the literature and the ones developed here, in the context of
divisibility. We recall a particular representation
and a couple of decomposition theorems for qubit channels.

\subsection{Pauli representation and Ruskai's decomposition} 
In the case of qubit channels we can have at the same time an
hermitian, traceless and unitary basis, it is the simple Pauli basis
$\frac{1}{\sqrt{2}}\lbrace \mathbf{1},\sigma_x,\sigma_y,\sigma_z\rbrace$. This induces
a simple $4\times 4$ representation with real entries given by
\begin{equation}
\hat \mcE=\left(\begin{array}{cc}
1 & \vec 0^{T} \\ 
\vec t & \Delta
\end{array} \right),
\label{eq:qubit_channel_pauli_basis}
\end{equation}
where $\Delta$ is a $3\times 3$ matrix with real entries and $\vec t$ a column
vector.  This describes the action of the channel in the Bloch sphere picture
in which the points $\vec{r}$ are identified with density matrices
$\varrho_{\vec{r}}=\frac{1}{2}(\one+\vec{r}\cdot\vec{\sigma})$~\cite{ruskai}.
Therefore the action of the channel is  described by
$\mcE(\rho_{\vec{r}})=\rho_{\Delta\vec{r}+\vec{t}}$.

In order to study qubit channels with simpler expressions, we will consider a
decomposition in unitaries such that
\begin{equation}
\mcE=\mcU_1 \mcD \mcU_2.
\label{eq:EUDU}
\end{equation}
\fxnote{Aclarar: si entiendo bien, en la SVD, se requieren propiedades de positividad de los lambda's. Esto a veces requiere usar para los R1 y R2 matrices con determinante negativo. Aquí, se insiste en usar únicamente rotaciones propias, y por ende no se insiste en la positividad de los lambda's? Si esto es correcto, debiera formularse más explícitamente}\fxwarning{David: Es correcto, lo edit\'e.}
This can be achieved using Ruskai's decomposition~\cite{ruskai}, which can be
performed by decomposing $\Delta$ in rotation matrices,
\ie{} $\Delta=R_1 D R_2$, where $D={\rm diag}(\lambda_1,\lambda_2,\lambda_3)$
is diagonal and the rotations $R_{1,2}\in\text{SO}(3) $ (of
the Bloch sphere) correspond to the unitary channels $\mcU_{1,2}$. {\blue Notice that as $D$ is not required to be positive-semidefinite, Ruskai's decomposition must not be confused with the singular value
decomposition.} The latter allows decompositions that include total
reflections. Such operations do not correspond to unitaries over a qubit, in
fact they are not CPTP. An example of this is the
universal NOT gate defined by $\rho \mapsto \one -\rho$, it is PTP but not
CPTP. The resulting form from Ruskai's decomposition is stated in the following
theorem,
\begin{theorem}[Special orthogonal normal form]
For any qubit channel $\mcE$, there exist two unitary conjugations , $\mcU_1$ and $\mcU_2$, such that $\mcE=\mcU_1\mcD\mcU_2$, where $\mcD$ has the following form in the Pauli basis,
\begin{equation}
\hat \mcD=\left( \begin{array}{cc}
1 & \vec 0^{T} \\ 
\vec{\gamma} & D
\end{array} \right)\,,
\label{eq:orthogonalform}
\end{equation}
and is called special orthogonal normal form of $\mcE$.
\label{thm:orthogonal_normal_form}
\end{theorem}

Here, $R_1^\text{T}\Delta R_2^\text{T}=D$ and $\vec{\gamma}=R_1^T\vec{t}$. The latter describes the
shift of the center of the Bloch sphere under the action of $\mcD$.
The parameters $\vec{\lambda}$ determine the length of semi-axes
of the Bloch ellipsoid, being the deformation of Bloch sphere under
the action of $\mcE$. In particular $\det \hat \mcD=\det \hat \mcE=\lambda_1 \lambda_2 \lambda_3$.

To develop geometric intuition in the space determined by the possible
values of the three parameters of $\vec \lambda$, consider the \Jami{}
representation of the special orthogonal normal form of an arbitrary channel in 
the basis that diagonalises $D$,
\begin{equation}
\tau_\mcD=\frac{1}{4}\left(
\begin{array}{cccc}
 \gamma _3+\lambda _3+1 & \gamma _1-i \gamma _2 & 0 & \lambda _1+\lambda _2 \\
 \gamma _1+i \gamma _2 & -\gamma _3-\lambda _3+1 & \lambda _1-\lambda _2 & 0 \\
 0 & \lambda _1-\lambda _2 & \gamma _3-\lambda _3+1 & \gamma _1-i \gamma _2 \\
 \lambda _1+\lambda _2 & 0 & \gamma _1+i \gamma _2 & -\gamma _3+\lambda _3+1 \\
\end{array}
\right).
\label{eq:choi_qubit}
\end{equation}
Complete positivity is determined by the non-negativity of its eigenvalues,
given that it is hermitian, but it turns that for the general case they have
complicated expressions. To overcome this problem we use the fact that if
$\mcD$ is a channel, then its \textit{unital part}, defined by taking $\vec
\gamma=\vec 0$, is a channel too~\cite{wolfnotes}.
Therefore the set of the possible
values of $\vec \lambda$ for the general case is contained in the set arising
from the unital case. The complete positivity conditions for the latter are
\begin{equation}
1+\lambda_i\pm(\lambda_j+\lambda_k) \geq 0,
\label{eq:complete_positivity_qubit}
\end{equation}
with $i$, $j$ and $k$ all different,
this implies that the possible set of lambdas lives inside the tetrahedron with
corners $(1,1,1)$, $(1,-1,-1)$, $(-1,1,-1)$ and $(-1,-1,1)$, see
\fref{fig:simple_tetra}. For unital channels, all points in the tetrahedron are
allowed. The corner $\vec \lambda = (1,1,1)$
corresponds to the identity channel, $\vec \lambda = (1,-1,-1)$ to $\sigma_x$
$\vec \lambda = (-1,1,-1)$ to $\sigma_y$  and $\vec \lambda = (-1,-1,1)$ to
$\sigma_z$ (Kraus rank 1 operations). Points in the edges correspond to Kraus
rank 2 operations, points in the faces to Kraus rank 3 operations and in the
interior of the tetrahedron to Kraus rank 4 operations.
In particular, this tetrahedron defines the set of \textit{Pauli channels}, which are defined to have diagonal special orthogonal normal form.
\begin{definition}[Pauli channels]
A qubit channel $\mcE$ is a Pauli channel if
\begin{equation}
\mcE[\rho]=\sum_{i=0}^3 p_i \sigma_i \varrho \sigma_i,
\end{equation}
with $\sigma_0:= \one$, $p_i\geq 0$ and $\sum_{i=0}^3 p_i=1$.
\end{definition}

 For non-unital channels more restrictive conditions arise, an example will be given later.
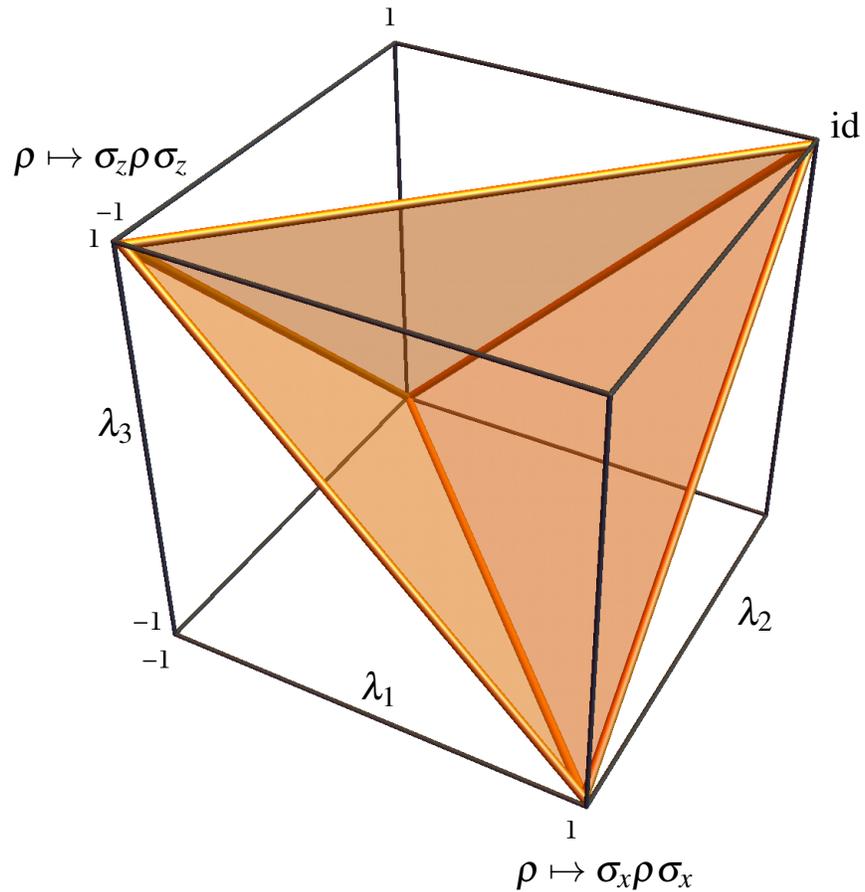
\begin{figure}
\centering
\begin{tikzpicture}
\node at (0,0) {\includegraphics[draft=false,scale=1.2]{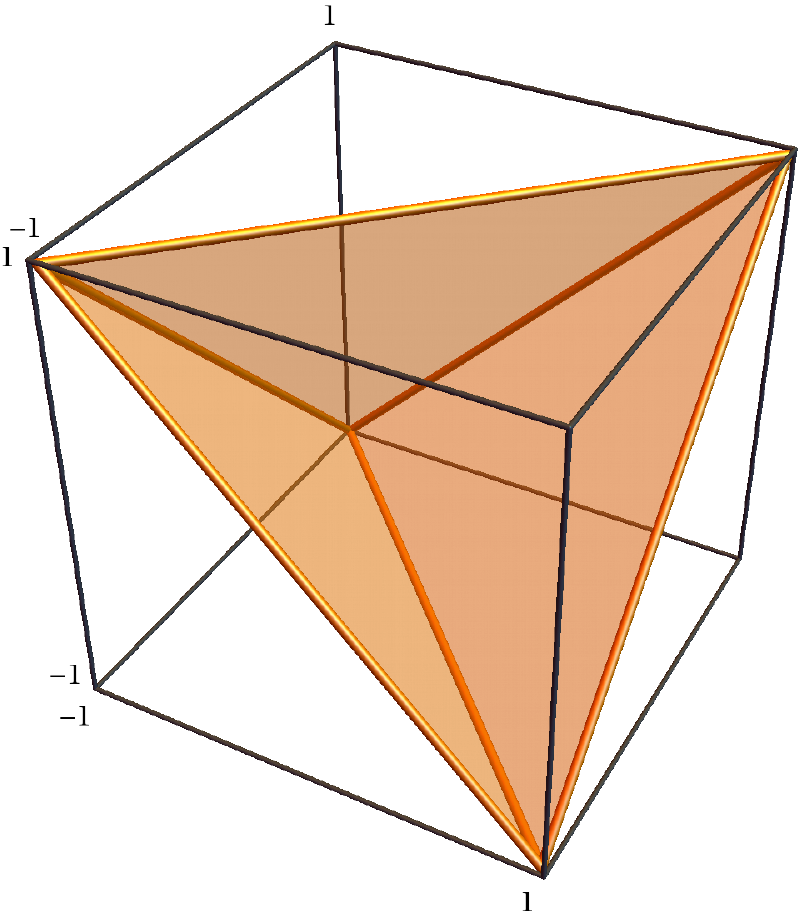}};
\node at (-1,-3.5) {\Large $\lambda_1$};
\node at (-4.5,0) {\Large $\lambda_3$};
\node at (4,-2.5) {\Large $\lambda_2$};
\node at (5.2,4) {\Large id};
\node at (2,-6) {\Large $\varrho\mapsto \sigma_x \varrho \sigma_x$};
\node at (-4.7,3.5) {\Large $\varrho\mapsto \sigma_z \varrho \sigma_z$};
\end{tikzpicture}
\caption{
Set of the possible values of $\vec \lambda$. This set has the shape of a
tetrahedron where the corners are the Pauli unitaries ($\one$, $\sigma_x$ and
$\sigma_z$ are indicated in the figure, while $\sigma_y$ lies behind). The rest
of the body contains convex combinations of Pauli unitaries.
Unital qubit channels can be obtained by concatenating Pauli channels with
unitary conjugations, see theorem~\ref{thm:orthogonal_normal_form}.
\label{fig:simple_tetra}}
\end{figure}

\label{sec:qubitchan}

\subsection{1wSLOCC and singular value decomposition using the Lorentz metric} 
\label{sec:lorentz}

There is another parametrization for qubit channels called \textit{Lorentz
normal decomposition}~\cite{Verstraete2001} which is specially useful to
characterize infinitesimal divisibility \InfDiv{}.
To introduce it, let us
resort to chapter~\ref{chap:open_quantum_systems} where we discussed local
operations and classical communication. For the two-qubit case, the operations
that Alice and Bob apply for their reduced states are
\begin{align}
\rho_{A} &\mapsto X \rho_A X^\dagger\nonumber \\
\rho_{B} &\mapsto Y \rho_B Y^\dagger,
\label{eq:locc_alice_bob}
\end{align}
where we have shown that it is enough to consider $X,Y\in
\text{SL}(2,\mathbb{C})$ for $X$ and $Y$ invertible, see
chapter~\ref{chap:open_quantum_systems} 
Now we are going to show
that such operations can be understood as proper orthochronous Lorentz
transformations in the Pauli representation.

Consider an arbitrary hermitian operator $\Delta$ and its representation in the
Pauli basis,
$$\Delta=\left(\one \tr \Delta + \vec r_\Delta \cdot \vec \sigma\right)/2$$
with $ \det \Delta=(\tr \Delta)^2-|\vec r_\Delta|^2$. 
Now observe that $\det \Delta$ can be understood as the squared Lorentz norm of the four-vector $r_\Delta=\left(\tr \Delta, \vec r_\Delta \right)^\text{T}$, lying in the Minkowski vector space, denoted as $(\mathbb{R}^4,\eta)$, where $\eta=\diag\left(1,-1,-1,-1 \right)$ is the \textit{Lorentz metric}. Therefore we have
\begin{equation}
\det \Delta=|r_\Delta|^2_\text{Lorentz}=\langle r_\Delta,\eta r_\Delta \rangle,
\label{eq:norm_Lorentz}
\end{equation}
 with $\langle \cdot \rangle$ the standard inner product.
Then $\tr \Delta$ is a time-like component and $\vec r_\Delta$ space-like components. 

Given that operations shown in~\eref{eq:locc_alice_bob} preserve the
determinant, they are isometries in the Minkowski space. That is, they preserve
the norm shown in~\eref{eq:norm_Lorentz} for any vector $r_\Delta$ (with
$\Delta$ hermitian).
Additionally, due to linearity of ~\eref{eq:locc_alice_bob}, these operations
belong to $\text{SO}(3,1)$ (the Lorentz group). In fact, due to the positivity
of quantum operations, they do not change the sign of the trace (the time-like
component); therefore the transformations are orthochronous. Also notice that
$\text{SL}(2,\mathbb{C})$ contains the identity transformation, therefore the
set of one-way stochastic local operations and classical communication is
identified with the proper orthochronous Lorentz group,
$\text{SO}^{+}(1,3)$~\cite{wolfnotes,wuki}.
However, since $-X$ and $X$ give the same result, see~\eref{eq:locc_alice_bob},
and both belong to $\text{SL}(2,\mathbb{C})$, one says that the latter is a
double cover of $\text{SO}^{+}(1,3)$. This map is also called \textit{spinor
map}.

Given this map, it is expected that the operations mentioned in
\eref{eq:locc_alice_bob} are explicitly Lorentz matrices, when writing them in
the Pauli basis. Also notice that unitary conjugations are particular cases of
them~\cite{ruskai}, therefore one can think of a different 
decomposition using the Lorentz metric instead of the three dimensional
Euclidean metric, used in Ruskai's decomposition.

The Lorentz normal form was introduced first for two-qubit states by writing
them as 
\begin{equation}
\tau=\frac{1}{4} \sum_{ij} R_{ij} \sigma_i \otimes \sigma_j,
\label{eq:choi_R}
\end{equation}
 where we
have used the notation of \Jami{} states for convenience. 
This decomposition is derived from the theorem 3 of Ref.~\cite{Verstraete2001},
which essentially states that the matrix $R$ can be decomposed as 
\begin{equation}
R=L_1\Sigma L^\text{T}_2.
\label{eq:lorentz_decomp_states}
\end{equation}
Here $L_{1,2}$ are proper orthochronous Lorentz transformations
and $\Sigma$ is either $\Sigma=\diag \left( s_0, s_1, s_2, s_3 \right)$ with
$s_0 \ge s_1 \ge s_2 \ge |s_3|$, or 
\begin{equation}
\Sigma=\left(
\begin{array}{cccc}
 a & 0 & 0 & b \\
 0 & d & 0 & 0 \\
 0 & 0 & -d & 0 \\
 c & 0 & 0 & -b+c+a 
\end{array}
\right).
\label{eq:state_normal_form_singular}
\end{equation}
Note that $\Sigma$ corresponds to an unnormalized state, with trace $\tr
\Sigma=a$. Thus, the normalization constant  is $\alpha=a^{-1}$.

To introduce the Lorentz normal decomposition of qubit channels, let us first
introduce the following. Let $\mcE$ a qubit channel and $\hat \mcE$ its matrix
representation using the Pauli basis. The latter is related with the matrix
$R$, which defines its \Jami{} state, see~\eref{eq:choi_R}, 
\begin{equation}
\hat \mcE\Phi_\text{T}=R,
\label{eq:choi_pauli}
\end{equation}
where $\Phi_\text{T}=\diag\left(1,1,-1,1 \right)$. This can be shown by
defining a generic Pauli channel, computing its Choi matrix and extracting $R$
using the Hilbert-Schmidt inner product with the basis $\lbrace \sigma_i\otimes
\sigma_j \rbrace_{i,j}$. Now, defining the decomposition for channels
throughout decomposing the \Jami{} state, we can easily compute the
corresponding Lorentz transformations using equations~(\ref{eq:choi_pauli}) and
(\ref{eq:lorentz_decomp_states}):
\begin{align}
R \Phi_\text{T}&=\alpha L_1 \Sigma L_2^\text{T} \Phi_\text{T}\nonumber\\
\hat \mcE &= \alpha L_1 \Sigma L_2^T \Phi_\text{T}\\
&= L_1 \left(\alpha \Sigma \Phi_\text{T} \right) \Phi_\text{T} L_2^T \Phi_\text{T}\\
&= L_1 \hat{\tilde{\mcE}} \tilde L_2^T,
\label{eq:lorentz_form_for_channels}
\end{align}
where $\hat{\tilde{\mcE}} =\alpha \Sigma \Phi_\text{T}$ is the Lorentz normal form of $\hat \mcE$. Also notice that $\tilde L_2^\text{T}=\Phi_\text{T} L_2^\text{T}\Phi_\text{T}$ is proper and orthochronous, given that its determinant is positive and $\Phi_\text{T}$ is proper. Therefore the possible Lorentz normal forms for channels are
$\hat{\tilde{\mcE}}=\diag \left( s_0, s_1, -s_2, s_3 \right)$ with
$s_0 \ge s_1 \ge s_2 \ge |s_3|$, or 
\begin{equation}
\hat{\tilde{\mcE}}=\left(
\begin{array}{cccc}
 a & 0 & 0 & b \\
 0 & d & 0 & 0 \\
 0 & 0 & d & 0 \\
 c & 0 & 0 & -b+c+a 
\end{array}
\right).
\label{eq:channel_normal_form_singular}
\end{equation}
Using this, the authors of~Ref.~\cite{Verstraete2002} introduced a theorem (theorem 8 of the reference) defining the Lorentz normal form for channels by forcing
$b=0$, in order to have normal forms proportional to trace-preserving operations. The latter is equivalent to say that the
decomposition of \Jami{} states leads to states that are also
\Jami{}. We didn't find a good argument to justify such assumption, and
found a counterexample that shows that Lorentz normal forms with $b\neq 0$ exist (see appendix~\ref{sec:normal_form}). Therefore in
general we can find a $\Sigma$ with form
of~\eref{eq:state_normal_form_singular} with $b\neq 0$. The consequence of this
is that the theorem 8 of Ref.~\cite{Verstraete2002} is incomplete, but given
that form of~\eref{eq:state_normal_form_singular} is Kraus rank deficient (it
has rank three for $b\neq c$ and two for $b=c$), the full Kraus rank case is
still useful. Thus, we propose a restricted version of their theorem:
\begin{theorem}[\textbf{Restricted Lorentz normal form for qubit quantum channels}]
For any full Kraus rank qubit channel $\mcE$ there exist rank-one completely
positive maps $\mcT_1,\mcT_2$ such that $\mcT=\mcT_1\mcE\mcT_2$ is proportional
to 
\begin{equation}
\left( \begin{array}{cc}
1 & \vec 0^\text{T} \\ 
\vec 0 & \Lambda
\end{array}  \right),
\label{eq:form_fullKraus_Lorentz}
\end{equation}
where
$\Lambda={\rm
        diag}(s_1,s_2,s_3)$ with $ 1\geq s_1 \geq s_2 \geq |s_3|$.
\label{thm:Lorentz}
\end{theorem}

The channel $\mcT$ is called the Lorentz normal form of the channel $\mcE$. For unital qubit channels $D$ coincides with $\Lambda$.


\section{Representation of Gaussian quantum channels} 
\label{sec:gqc}
In this section we start from two ans\"atze, that put together with the
Gaussian functional form considered in Ref.~\cite{PazSupplementary}, lead to
the complete set of functional forms in position state representation of
one-mode Gaussian channels.

We will show that only two possible forms of \dgqc{} hold according to
\textit{trace preserving} (\tp{}) and \textit{hermiticity preserving} (\hp{})
conditions. The one corresponding to~\eref{eq:typeII} is one of these, as
expected.  Later on we will impose \textit{complete positivity} in order to
have valid \gqc{}, \ie{} \textit{complete positive and trace preserving}
(\cptp{}) Gaussian operations.

Following definition~\ref{def:gqc}, those channels can be characterized by how
they act over Gaussian states.  It is well known that the action of GQCs on
Gaussian states is described by \textit{affine
transformations}~\cite{Reviewquantuminfo}. Let $\mcG$ be a GQC defined by a
tuple such that $\mcG=\mcG\left( \mathbf{T}, \mathbf{N}, \vec \tau \right)$,
where $\mathbf{T}$ and $\mathbf{N}$ are $2\times 2$ real matrices with
$\mathbf{N}=\mathbf{N}^\text{T}$~\cite{Reviewquantuminfo}. The transformation
acts on Gaussian states according to 
$$
\mcG\left( \mathbf{T},
\mathbf{N}, \vec \tau
\right)\left[S\left(\sigma,\vec d\right)\right]=S\left( \mathbf{T}\sigma
\mathbf{T}^{\text{T}} +\mathbf{N}, \mathbf{T} \vec d +\vec \tau \right).
$$
In
the particular case of closed systems, where the system is governed by a
Gaussian unitary, we have that $\mathbf{N}=\mathbf{0}$ and $\mathbf{T}$ is a
symplectic matrix.
\subsection{Possible functional forms of \dgqc{} operations} 
\label{sec:ptp}     
Let us introduce the ans\"atze for the possible forms of
\gqc{} in the
position representation, to perform the black-box characterization. Following
\eref{def:def_2} and taking the continuous variable representation of
difference and sum coordinates, the trace becomes an integral over position
variables of the environment. Then we end up with a Fourier transform of a
multivariate Gaussian. Since the Fourier transform of a Gaussian is
again a Gaussian (unless there are singularities in the coefficients, as in the
example of \eref{eq:typeII}), the result of the Fourier transform for one mode
can have the following structures: a Gaussian form [\eref{eq:gf}], a Gaussian
form multiplied with one-dimensional delta or a  Gaussian form multiplied by a
two-dimensional delta. No more deltas are allowed given that there are only two
integration variables when applying the channel, see \eref{eq:propagacion}.
Thus, in order to start with the black-box
characterization, we shall propose the following general Gaussian operations
with one and two deltas, respectively
\begin{align}
J_\text{I}(x_f,r_f;x_i,r_i)&=
  \mcN_\text{I} \delta(\vec \alpha^{\text{T}}
     \vec v_f+\vec \beta^{\text{T}} \vec v_i)e^{\Sigma(x_f,x_i;r_f,r_i)}, 
 \label{eq:deltaop1} \\
J_\text{II}(x_f,r_f;x_i,r_i)&=\mcN_\text{II} \delta (\mathbf{A}\vec{v}_f-\mathbf{B} \vec v_i)e^{\Sigma(x_f,x_i;r_f,r_i)}, 
\label{eq:deltaop2}
\end{align}
with $\vec{v}_{i,j}=(r_{i,j},x_{i,j})$, and $\mcN_\text{I,II}$ are
normalization constants. Coefficient arrays $\mathbf{A}$, $\mathbf{B}$,
$\vec \alpha$, and $\vec \beta$ have real entries since initial and final
coordinates must be real. Finally, the exponent reads:
\begin{multline*}
\Sigma(x_f,x_i;r_f,r_i)=\Imi \Big( b_1 x_f r_f+b_2x_f r_i+b_3x_ir_f
 +b_4x_ir_i+c_1x_f
+c_2x_i
  \Big)\\
-a_1 x_f^2-a_2 x_fx_i-a_3x_i^2 -e_1r_f^2-e_2r_fr_i-e_3r_i^2 -d_1r_f-d_2r_i.
\end{multline*}
They provide, together with \eref{eq:gf} all possible ans\"atze for \gqc{}.
\subsection{Hermiticity and trace preserving conditions}  
Before studying CPTP conditions it is useful to simplify expressions of
equations (\ref{eq:deltaop1}) and (\ref{eq:deltaop2}). To do this we use the fact that linear CPTP operations preserve hermiticity and trace. 
For channels of continuous variable systems in the position state representation, $J(q_f,q_f';q_i,q_i')$, \hp{} condition is derived as follows,
\begin{align}
\rho_f(q'_f,q_f)^*&=\int_{\mathbb{R}^2}dq_i dq'_i J(q_f',q_f;q_i,q_i')^*\rho_i(q_i,q_i')^*\nonumber\\
&=\int_{\mathbb{R}^2} dq_i dq_i' J(q_f',q_f;q_i',q_i)^*\rho_i(q_i,q_i')\nonumber\\
&=\rho_f(q_f,q'_f),
\end{align}
where the last equality holds if
$$J(q_f,q'_f;q_i,q'_i)=J(q'_f,q_f;q'_i,q_i)^*.$$
Using sum and difference coordinates, \hp{} becomes
\begin{equation}
J(-x_f,r_f;-x_i,r_i)=J(x_f,r_f;x_i,r_i)^{*}.
\label{eq:hp_gaussian}
\end{equation}
Following this equation and comparing exponents of the both sides of the last equations, it is easy to note that the coefficients $a_n$, $b_n$,
$c_n$, $e_n$ and $d_n$ must be real.
Concerning the delta factors, in~\eref{eq:hp_gaussian} we end up with
expressions like 
$$\delta\left(\alpha_1 x_f +\alpha_2 x_i+\beta_1 r_f +\beta_2 r_i\right)=\delta\left(-\alpha_1 x_f -\alpha_2 x_i+\beta_1 r_f +\beta_2 r_i\right)$$
for both cases. Therefore the equality holds for
\eref{eq:deltaop1} only for two possible combinations of variables: i)
$\delta(\alpha x_f-\beta x_i) $ and
ii) $\delta(\alpha r_f-\beta r_i) $. For the case of \eref{eq:deltaop2}, equality holds only for iii) $
\delta(\gamma r_f-\eta r_i)\delta(\alpha x_f-\beta x_i)$.
Let us now analyze the trace preserving condition (\tp{}), since the trace of $\rho_f$ in sum and difference coordinates is 
\begin{align*}
\tr \rho_f &=\int_\mathbb{R} dr_f' \rho_f(x_f=0,r_f')\\
&=\int_\mathbb{R} dr_f' dr_i dx_i J(x_f=0,r_f';x_i,r_i)\rho_i(x_i,r_i)\\
&=\int_\mathbb{R} dr_i \rho_i(x_i=0,r_i).
\end{align*}
To fulfill the last equality, the following must be accomplished
\begin{equation}
\int_{\mathbb{R}} dr_f' J(x_f=0,r_f';x_i,r_i)=\delta (x_i).
\label{eq:tp}
\end{equation} 
This condition immediately discards ii) from the above combinations of deltas,
thus we end up with cases i) and iii). For case i) \tp{} reads:
\begin{equation}
\mcN_\text{I}\int dr_f\delta(-\beta x_i)e^{\Sigma}
    =\frac{\mcN_\text{I}}{|\beta|}\sqrt{\frac{\pi}{e_1}}\delta(x_i)
	e^{\left(\frac{e_2^2}{4e_1}-e_3\right)r_i^2},
\end{equation}
thus, the relation between the coefficients assumes the form
\begin{equation}
\frac{e_2^2}{4e_1}-e_3=0, d_1=0,d_2=0,
\label{eq:trace_preserving_form1}
\end{equation} 
and the normalization constant $\mcN_\text{I}=|\beta |\sqrt{\frac{e_1}{\pi}}$ with
$\beta\neq 0$ and $e_1>0$.  For case iii) the trace-preserving condition reads
\begin{equation*}
\mcN_\text{II}
\int dr_f \delta(\gamma r_f-\eta r_i)\delta(-\beta x_i)e^{\Sigma} \\
	=\frac{\mcN_\text{II}}{|\beta \gamma|}\delta(x_i)
	e^{-\left(e_1(\frac{\eta}{\gamma})^2+e_2\frac{\eta}{\gamma}+e_3\right)r_i^2-\left(d_1\frac{\eta}{\gamma}+d_2\right) r_i}.
\end{equation*}
Thus, the following relation between $e_n$ and $d_n$ coefficients must be fulfilled:
\begin{equation}
e_1\Bigg(\frac{\eta}{\gamma}\Bigg)^2+e_2\frac{\eta}{\gamma}+e_3=0, \ \ d_1\frac{\eta}{\gamma}+d_2=0,
\label{eq:restricII}
\end{equation}
with $\gamma,\beta \neq 0$ and $\mcN_\text{II}=|\beta \gamma|$. In the particular case of $\eta=0$, \eref{eq:restricII} is reduced to $e_3=d_2=0$. 
As expected from the analysis of limits above, we showed that \dgqc{}'s admit quadratic terms in $r_{i,j}$. 

\subsection{Complete positivity conditions} 
Up to this point we have \textit{hermitian and trace preserving Gaussian
operations}; to derive the remaining CPTP conditions, it is useful to write its
Wigner's function and Wigner's characteristic function.
The representation of the Wigner's characteristic function reads
\begin{equation}
\chi(\vec k)=\exp\left[ -\frac{1}{2}\vec k^{\text{T}}\left( \Omega \sigma \Omega^{\text{T}}\right)\vec k-\Imi\left( \Omega \langle \hat x \rangle \right)^{\text{T}} \vec k \right]
\label{eq:cha}
\end{equation}
and its relation with Wigner's function:
\begin{align}
W(\mathbf{x})&=\int_{\mathbb{R}^{2}} d\vec x e^{-\Imi \vec{x}^{\text{T}}\Omega \vec k}\chi\left(\vec k\right)\\
&=\int_{\mathbb{R}}e^{\Imi p x} dx \left.\left\langle  r-\frac{x}{2} \right. \right| \left. \hat \rho \left|  r+\frac{x}{2} \right. \right\rangle,
\label{eq:wigners}
\end{align}
where $\vec k=\left(k_1,k_2 \right)^{\text{T}}$, $\vec
x=\left(r,p\right)^{\text{T}}$ and $\hbar=1$ (we are using natural units).
Using the previous equations to construct Wigner and Wigner's
characteristic functions of the initial and final states, and substituting them
in the equation~\ref{eq:propagacion}, it is straightforward to get the propagator in the Wigner's
characteristic function representation:
\begin{align}
\tilde J \left( \vec k_f,\vec k_i \right)=\int_{\mathbb{R}^6} d\Gamma K(\vec l) J(\vec v_f,\vec v_i),
\label{eq:changeofrepresentation}
\end{align}
where the transformation kernel reads
$$K(\vec{l})= \frac{1}{(2\pi)^3}e^{  \left[ \Imi \left( k_2^f r_f-k_1^f p_f -k_2^i r_i+k_1^i p_i -p_i x_i +p_f x_f \right)  \right]  },$$ 
with 
$d\Gamma=d p_f d p_i dx_f dx_i dr_f dr_i$ and 
$\vec l=\left(p_f, p_i, x_f, x_i, r_f,r_i \right)^{\text{T}}$.
By elementary integration of ~\eref{eq:changeofrepresentation} one can
show that for both cases
\begin{equation}
\tilde J_{\text{I,III}}\left( \vec k_f, \vec k_i \right)= \delta \left ( k_1^i -\frac{\alpha}{\beta} k_1^f  \right) \delta \left( k_2^i -\vec\phi^{\text{T}}_{\text{I,III}} \vec k_f  \right) e^{P_\text{I,III}(\vec k_f)},
\label{eq:formcharI}
\end{equation}
where $P_{\text{I,III}}(\vec k_f)=\sum_{i,j=1}^2 P^{(\text{I,III})}_{ij} k^f_i
k^f_j+\sum_{i=1}^2 P^{(\text{I,III})}_{0i}k^f_i$ with
$P^{(\text{I,III})}_{ij}=P^{(\text{I,III})}_{ji}$.
For case i) we obtain
\begin{align}
P^\text{(I)}_{11}&=-\left(\left(\frac{\alpha}{\beta}\right)^2\left(a_3+\frac{b_3^2}{4 e_1}\right)+\frac{\alpha}{\beta} \left(a_2+\frac{1}{2} \frac{b_1 b_3}{e_1}\right)+a_1+\frac{b_1^2}{4 e_1}\right),\nonumber\\
P^\text{(I)}_{12}&=-\left(\frac{\alpha}{\beta}\frac{b_3}{2 e_1}+\frac{b_1}{2 e_1}\right),\nonumber\\
P^\text{(I)}_{22}&=-\frac{1}{4 e_1}.
\label{eq:deltaI_Ps}
\end{align}
For case iii) we have
\begin{align}
P^\text{(III)}_{11}&= -\left(\left(\frac{\alpha}{\beta}\right)^2a_3 +\frac{\alpha}{\beta} a_2 +a_1\right),\nonumber\\
P^\text{(III)}_{12}&=P^\text{(III)}_{22}=0.
\label{eq:deltaII_Ps}
\end{align}
And for both cases we have $P^{(\text{I,III})}_{01}=\Imi \left( \frac{\alpha}{\beta}c_2 +c_1\right)$ and $P^{(\text{I,III})}_{02}=0$.
Vectors $\vec \phi$ are given by
\begin{align}
\vec\phi_{\text{I}}&=\left(\frac{\alpha}{\beta}\left( b_4-\frac{b_3 e_2}{2 e_1}\right)-\frac{b_1 e_2}{2 e_1}+b_2,-\frac{e_2}{2 e_1}\right)^\text{T},\nonumber\\
\vec\phi_{\text{III}}&=\left( \frac{\alpha}{\beta}\frac{\eta}{\gamma}b_3+\frac{\alpha}{\beta}b_4 +\frac{\eta}{\gamma}b_1 +b_2,\frac{\eta}{\gamma} \right)^\text{T}.
\label{eq:phis}
\end{align}

We are now in position to write explicitly the conditions for complete positivity. 
Having a Gaussian operation
characterized by
$\left( \mathbf{T}, \mathbf{N}, \vec \tau \right)$, 
the CP condition can be expressed in terms of the matrix
\begin{equation}
\mathbf{C}=\mathbf{N}+\Imi\Omega -\Imi\mathbf{T}\Omega \mathbf{T}^{\text{T}},
\label{eq:ccp}
\end{equation}
where $\Omega=\left( \begin{array}{cc}
0 & 1 \\ 
-1 & 0
\end{array}  \right)$ is the symplectic matrix. An operation $\mcG\left(
\mathbf{T}, \mathbf{N}, \vec \tau \right)$ is CP if and only if $\mathbf{C}\geq
0$~\cite{cptp,Reviewquantuminfo}. 
Applying the propagator on a test characteristic function, \eref{eq:cha}, it is
easy compute the corresponding tuples. For both cases we get:
\begin{align}
\mathbf{N}_\text{I,III}&=2\left( \begin{array}{cc}
-P_{22} & P_{12} \\ 
P_{12} & -P_{11}
\end{array}  \right),\nonumber\\
\vec \tau_{\text{I,III}}&=\left(0,\Imi{} P^{(\text{I,III})}_{01} \right)^\text{T},
\label{eq:nsandtaus}
\end{align}
while for case i) matrix $\mathbf{T}$ is given by
\begin{equation}
\mathbf{T}_\text{I}=\left(
\begin{array}{cc}
 \frac{e_2}{2 e_1} & 0 \\
 \vec \phi_\text{I,1} & -\frac{\alpha }{\beta } \\
\end{array}
\right),
\label{eq:tupleI}
\end{equation}
where $\vec \phi_\text{I,1}$ denotes the first component of vector $\vec \phi_\text{I}$, see \eref{eq:phis}.
The complete positive condition is given by the inequalities raised from the eigenvalues of matrix~\eref{eq:ccp}:
\begin{align}
\pm\frac{\sqrt{\alpha ^2 e_2^2+4 \alpha  \beta  e_2 e_1+4 \beta ^2 e_1^2
\left(4 {P_{12}^{(\text{I})}}^2+\left(P_{11}^{(\text{I})}-P_{22}^{(\text{I})}\right)^2+1\right)}}{2
\beta  e_1}
\geq P_{11}^{(\text{I})}+P_{22}^{(\text{I})}.
\label{eq:CPI}
\end{align}
For case iii) matrix $\mathbf{T}$ is
\begin{equation}
\mathbf{T}_\text{III} = \left(
\begin{array}{cc}
 -\frac{\eta }{\gamma } & 0 \\
 \vec \phi_\text{III,1}  & -\frac{\alpha }{\beta } \\
\end{array}
\right),
\label{eq:tupleII}
\end{equation}
and complete positivity conditions read:
\begin{equation}
\pm \frac{\sqrt{(\beta  \gamma -\alpha  \eta )^2+\beta ^2 \gamma ^2 {P_{11}^{(\text{III})}}^2}}{\beta  \gamma }-P_{11}^{(\text{III})}\geq 0.
\label{eq:cpII}
\end{equation}
Note that in both cases the complete positivity conditions do not depend on
$\vec \phi$.


\chapter{Divisibility of quantum channels and dynamical maps}
\label{chap:div}
\begin{flushright} 
\textit{Wine is sunlight, held together by water.  }\\ Galileo Galilei
\end{flushright}
In this chapter we introduce the formal definition of \textit{divisibility of
quantum channels}, inspired by questioning  how can we implement a given
quantum channel via the concatenation of simpler channels.  Later on we define
further types of divisibility by adding extra conditions, such as channels
being infinitesimal divisible and channels belonging to one-parameter
semigroups. These types are physically relevant since both lead to Markovian
dynamical maps~\cite{rivasreview}.  We additionally prove three theorems, which
are the central contributions of this part of the work. Finally, a complete
characterization of channels belonging to one-parameter semigroups that is
given.
\section{Divisibility of quantum maps} 
A quantum channel $\mcE$ is said to be divisible if it can be expressed as
the concatenation of two non-trivial channels,
\begin{definition}[Divisibility]
A linear map $\mcE \in \cptp{}$ is divisible if there exists a decomposition,
\begin{equation}
\mcE=\mcE_2 \mcE_1,
\end{equation}
such that $\mcE_1$ and $\mcE_2$ are both unitary or non-unitary channels.
\label{def:divisibility}
\end{definition}
Notice that this definition ensures that unitary channels are divisible, and
that non-unitary channels must be divisible in non-unitary channels. This
prevents one to consider simple changes of basis as a ``division'' of a given
quantum operation. This type of divisibility, which is the most general and
less restrictive one, defines a set that will be denoted by \Div{}. The set of
indivisible channels is the complement of \Div{} in \cptp{}, therefore it will
be denoted as \Ind{}. Notice that this definition is different to the one given
in Ref.~\cite{cirac} where unitary channels are excluded to be divisible. 

The concept of indivisible channels resembles the concept of prime
numbers, unitary channels play the role of unity (which are not
indivisible/prime), i.e.  a composition of indivisible and a unitary channel
results in an indivisible channel. 

We now introduce three results from Ref.~\cite{cirac} that shall be used
later. We only give the proof for the second for the sake of 
brevity. 
\begin{theorem}[Full Kraus rank channels]
Let $\mcE:\mcT(\mcH)\to \mcT(\mcH)$ be a quantum channel. If it has full Kraus rank, \ie{} $d^2$ with $d=\dim\left( \mcH\right)$, then it is divisible.
\label{thm:divisible}
\end{theorem}
An example of full Kraus rank channel is the total depolarizing channel $\rho\mapsto \one/\dim \left( \mcH \right)$, which maps every state into the maximal mixed one.
\begin{theorem}[Indivisible channels]
Consider the set $\cptp{}_d$ of channels acting on the space of density
matrices of $d\times d$, \ie{} $\mcE:\mcT(\mcH)\to \mcT(\mcH)$ with $d=\dim
\left( \mcH \right)$. The channel with minimal determinant, $\mcE_0\in
\cptp{}_d$, is indivisible.
\end{theorem}
\begin{proof}
To prove this we use the fact that channels with negative determinant exist~\cite{cirac} (two examples are given below), and the property of monotonicity of the determinant. 

Let $\mcE\in \cptp{}$ with $\det \mcE<0$ and $\mcE=\mcE_2 \mcE_1$ an arbitrary
division of $\mcE$ with $\mcE_1,\mcE_2\in \cptp{}$. The monotonicity of the
determinant implies the following,
$$ |\det \left( \mcE_2 \mcE_1 \right)| =|\det \mcE_2 ||\det \mcE_1|  \leq |\det
\mcE_1|. $$
Assuming, without loss of generality that $\det \mcE_1<0$ and $\det \mcE_2>0$,
we have that 
$$\det \mcE_1 \det \mcE_2\leq \det \mcE_1.$$
Multiplying both sides by $-1$ we arrive to
$$|\det \mcE_2 ||\det \mcE_1|  \geq |\det \mcE_1|. $$
Therefore, by monotonicity of the determinant, we have 
$$|\det \mcE_2 ||\det \mcE_1| = |\det \mcE_1|, $$
which implies that $\det \mcE_2=1$, \ie{} $\mcE_2$ is an unitary conjugation~\cite{cirac} and $\mcE$ has minimum determinant.
By definition~\ref{def:divisibility} $\mcE$ is indivisible.
\end{proof}
Two examples for the qubit case are the approximate NOT and the approximate
transposition maps:
\begin{align}
\rho &\mapsto \frac{\tr(\rho) \one+\rho^\text{T}}{3}\text{ (approximate transposition),}\nonumber\\
\rho &\mapsto \frac{\tr(\rho)\one -\rho}{3}\text{ (approximate NOT gate),}
\label{eq:NOT_gate}
\end{align}
both have minimal determinant corresponding to $-1/27$, which can be computed 
from their matrix representation.
\begin{theorem}[Unital Kraus rank three channels]
A unital qubit channel is indivisible if and only if it has Kraus rank equal to
three.
\label{thm:unital_indivisible}
\end{theorem}
This is a restricted version of theorem 23 of Ref.~\cite{Wolf2008}, where
authors proved the theorem for any qubit channel instead of only unital ones.
Since their proof rely on the validity of the Lorentz normal decomposition for
channels, we have written here a restricted version, where Lorentz normal form
is equivalent to the special orthogonal normal form (see theorem~\ref{thm:Lorentz} and
its discussion).

These results can be used immediately to identify the divisibility character of
unital qubit channels, see~\fref{fig:tetra}. The faces of the tetrahedron
(without edges) correspond to indivisible channels, in particular the center of
every face corresponds to channels with minimal determinant. The body (full
Kraus rank channels) contain divisible channels. 
\subsection{Subclasses of divisible maps} 
\subsubsection{Divisibility of quantum dynamical maps} 
We motivate the extra conditions to define new types of divisibility on
the concept of Markovian process.
In subsection~\ref{sec:classical_analog} we have introduced the definition of
Markovian process and its consequences at the level of propagators of one-point
probabilities, see \eref{eq:classical_divisibility}. Based on this, we
introduce the concept of CP-divisibility of quantum dynamical maps, which is often used as definition of Markovianity in the quantum realm~\cite{rivasreview}.
\begin{definition}[CP-divisible quantum dynamical maps]
Consider a quantum dynamical map $\mcE_{\left( t,0 \right)}:\mcT(\mcH)\to \mcT(\mcH)$ with $t\in \mathbb{R}^+$. It is CP-divisible in the interval $[0,t] \subset \mathbb{R}^+$ if for every decomposition of the form
$$\mcE_{\left(t,0\right)}=\mcE_{\left(t,s\right)}\mcE_{\left(s,0\right)},$$
$\mcE_{\left(t,s\right)}$ is a quantum channel for every $s\in (0,t).$
\label{def:cp_div_dyn}
\end{definition}

A remarkable theorem on CP-divisible maps is the following~\cite{kossa,kossa2,Gorini1976,lindblad,rivasreview},
\begin{theorem}[Gorini-Kossakowski-Susarshan-Lindblad]
An operator $L_t$ is the generator of a CP-divisible process if and only if it can be written in the following form:
\begin{equation}
L_t[\rho]= -\rmi [H(t),\rho]
 +\sum_{i,j} G_{i j} (t)
     \left( 
         F_{i}(t)\rho F^{\dagger}_{j}(t)
             -\frac{1}{2} \lbrace F^{\dagger}_{j}(t) F_{i}(t),\rho \rbrace 
     \right),
\label{eq:time_dependent_lindblad}
\end{equation}
where $G$ is hermitian and positive semidefinite, $H(t), F_k(t)\in \mathbb{C}^{d \times d}$ are time-dependent operators acting on $\mcH$, with $H(t)$ hermitian for every $t\in \mathbb{R}^+$, and $d=\dim\left(\mcH\right)$.
\label{thm:KKSG}
\end{theorem}
In Ref.~\cite{rivasbook} a proof is given starting from the Kraus representation of quantum dynamical maps and the definition of CP-divisibility. Here we will give a simpler proof resorting to theorem~\ref{thm:linblad_simple_form}. 
\begin{proof}
Notice that for each time $t$ we can define the ``instant'' map
$\mcE_{\left(t+\epsilon,t \right)}[\rho]=\rho +\epsilon
L_t[\rho]+\mcO(\epsilon^2)$, with $\epsilon>0$, therefore the hermiticity
preserving of $L_t$ follows from the hermiticity preserving of $\mcE_{(t,0)}$.
Also note that we can always choose the same traceless basis, $\lbrace
F_i\rbrace_{i=0}^{d^2-1}$, to write~\eref{eq:time_dependent_lindblad}, such
that the time dependence is dropped only in $G(t)\in \mathbb{C}^{d^2\times
d^2}$ and $H(t)$. By theorem~\ref{thm:linblad_simple_form}, $L_t$ has the form
stated in~\eref{eq:time_dependent_lindblad}, the only thing that remains to
prove is that $G(t)\geq 0$ for every $t$. To do this we construct the \Jami{}
matrix of the instant map, $\tau_{t,\epsilon}=\omega +\epsilon
\left(\id_{d^2}\otimes L_t\right)[\omega]+\mcO(\epsilon^2)$. We remind the reader that
$\omega=\proj{\Omega}{\Omega}$, where $\ket{\Omega}$ is the Bell state between
two copies of $\mathbb{C}^d$. Now we test positive-semidefinitiveness of
$\tau_{t,\epsilon}$,
\begin{align*}
\bra{\varphi}\tau_{t,\epsilon}\ket{\varphi}&=\langle \varphi | \Omega \rangle \langle \Omega |  \varphi \rangle+\epsilon \langle \varphi |\left( \id_{d^2}\otimes L_t\right)[\proj{\Omega}{\Omega}]  |  \varphi \rangle +\mcO(\epsilon^2)
\geq 0,
\end{align*}
$\forall \ket{\varphi} \in \mathbb{C}^{d^2}$.
The inequality always holds for any $\langle \varphi | \Omega \rangle\neq 0$
and $\epsilon>0$. For $\langle \varphi | \Omega \rangle=0$ we have that for
$\epsilon>0$ the inequality $\langle \varphi |\left( \id_{d^2}\otimes
L_t\right)[\proj{\Omega}{\Omega}]  |  \varphi \rangle\geq 0$ must be
accomplished, \ie{} $\omega_\perp \tau_L \omega_\perp \geq 0$ (conditional
complete positivity). Therefore by proposition~\ref{prop:ccp}, one has that
$G(t)\geq 0$.
\end{proof}
Analogously to CP-divisible processes, if we relax the condition of the
intermediate maps to be PTP (and not necessarily CPTP), we arrive to the
following definition:
\begin{definition}[P-divisible quantum dynamical maps]
Consider a quantum dynamical map $\mcE_{\left( t,0 \right)}:\mcT(\mcH)\to \mcT(\mcH)$ with $t\in \mathbb{R}^+$. It is P-divisible in the interval $[0,t] \subset \mathbb{R}^+$ if for every decomposition of the form
$$\mcE_{\left(t,0\right)}=\mcE_{\left(t,s\right)}\mcE_{\left(s,0\right)},$$
$\mcE_{\left(t,s\right)}$ belongs to PTP for every $s\in (0,t).$
\label{def:p_div_dyn}
\end{definition}
Unfortunately, to the best of our knowledge,  there doesn't exist a statement
similar to theorem~\ref{thm:KKSG}, nor a simple test of P-divisibility. But for
certain types of generators of dynamical maps, conditions for P-divisibility
were derived in Ref.~\cite{montesgorin}.
\subsubsection{Divisibility of quantum channels} 

Let us discuss these two types of divisibility but now from a statical point of
view. First notice that instant operations $\mcE_{(t+\epsilon,t)}$ are
arbitrarily close to the identity map as $\epsilon\to 0^+$, for both
P-divisible and CP-divisible processes. In other words, they are infinitesimal.
Consider now the idea of quantum channels divisible in infinitesimal parts,
\ie{} what is given this time is a quantum channel instead of a dynamical map.
This idea motivates the following definition~\cite{cirac},

\begin{definition}[Infinitesimal divisible channels in CPTP]
Let $\mcL_\text{CP}$ be the set containing operations $\mcE\in \cptp{}$ with
the property that $\forall\epsilon>0$ there exist a finite number of channels
$\mcE_i \in \cptp{}$ such that $\vert\mcE_i-\id\vert<\epsilon$ and
$\mcE=\prod_i \mcE_i$, see~\fref{fig:diagramatic_CP_DIV}. It is said that a
channel is infinitesimal divisible if it belongs to the closure of
$\mcL_\text{CP}$. This set is denoted as \cpDiv{}.
\label{def:infinitesimal_CP}
\end{definition}
\begin{figure} 
\centering
\begin{tikzpicture}
\node at (-2.5,2.5) {\Large $\mcE\in \cpDiv{}$};
\node at (-0.8,2) {\includegraphics[scale=0.15]{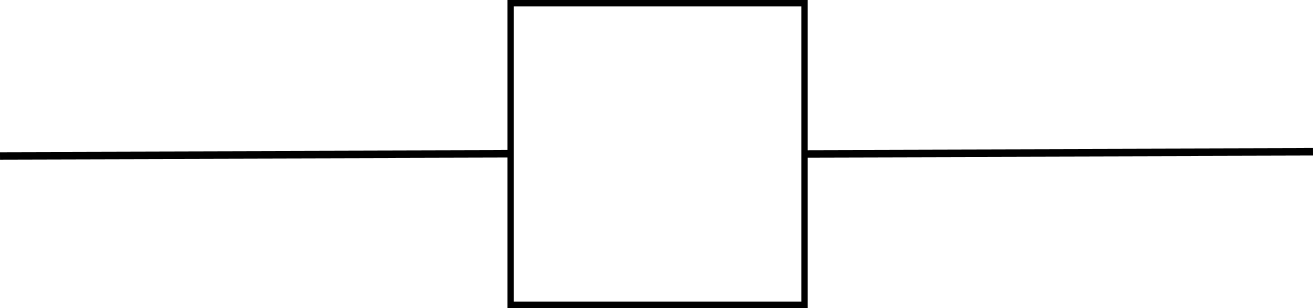}}; 
\node at (-0.8,2) { \Large $\mcE$};
\node at (2.8,2) {\Huge $=$};
\node at (0,0) {\includegraphics[scale=0.15]{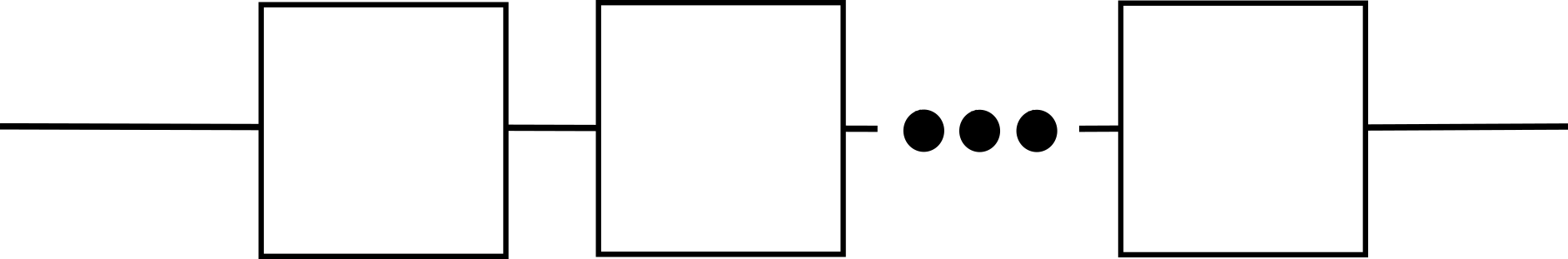}}; 
\node at (-2,0) { \Large $\mcE^1$};
\node at (-0.3,0) {\Large $\mcE^2$};
\node at (2.4,0) {\Large $\mcE^{N_\epsilon}$};
\draw (4.5,-1.8) -- (4.5,2.7);
\node at (7,0) {\includegraphics[scale=0.2]{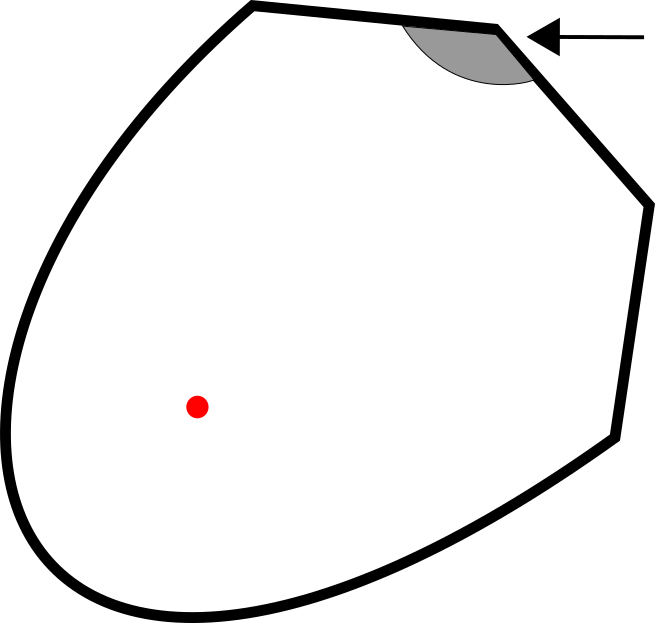}};
\node at (8,1.8) {$id$};
\node at (6.5,-0.7) {\Large $\mcE$};
\node at (9,1.5) {\Large $\mcE_i$};
\end{tikzpicture}
\caption{Diagrammatic decomposition of channels belonging to $\mcL_\text{CP}$
whose closure is \cpDiv{}, see definition~\ref{def:infinitesimal_CP}. 
We show the circuit representing the decomposition of $\mcE$ into channels (left)
arbitrarily close to the identity map (right).
\label{fig:diagramatic_CP_DIV}}
\end{figure}
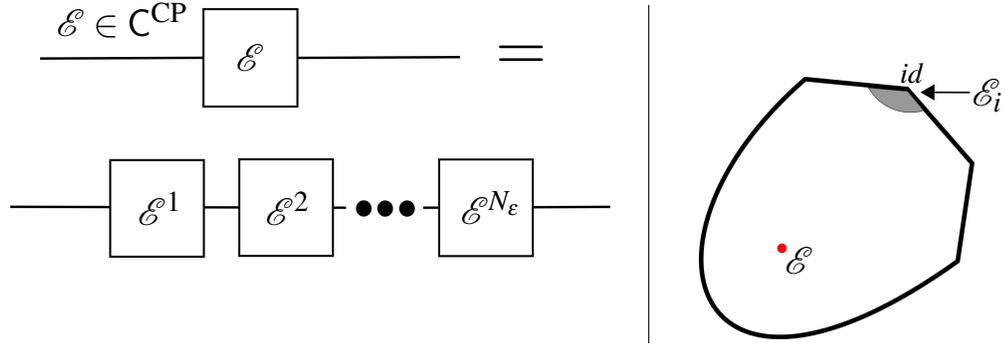 
The necessity of the closure can be motivated using the following example. Consider the qubit channel defined as follows:
\begin{equation}
\mcE_\infty: \left( \begin{array}{cc}
\rho_{00} & \rho_{01} \\ 
\rho_{01}^* & \rho_{11}
\end{array}  \right)\mapsto \left( \begin{array}{cc}
\rho_{00} & 0 \\ 
0 & \rho_{11}
\end{array}  \right).
\end{equation}
This channel is singular, \ie{} does not belong to $\mcL_\text{CP}$. Now
observe that using the dynamical process, $\mcE_t$, given in
example~\ref{example:dephasing}, one can get arbitrarily close to $\mcE_\infty$
when $t\to \infty$, \ie{} $\mcE_\infty=\lim_{t\to \infty}\mcE_t$.
Note that $\mcE_t\in \mcL_\text{CP}$ for every $t\in \mathbb{R}^+$, see
theorem~\ref{thm:KKSG}, therefore $\mcE_\infty$ is an accumulation point of
$\mcL_\text{CP}$. Thus, the closure is taken to define infinitesimal divisible
channels, to include channels such as $\mcE_\infty$.

Up to this point we have shown that CP-divisible processes are
infinitesimal divisible, \ie{} CP-divisible processes parametrize families of
channels belonging to \cpDiv{}. In Ref.~\cite{Wolf2008}, authors have shown
that channels in \cpDiv{} can always be implemented with CP-divisible
processes. This can be roughly shown as follows. 

Since \cptp{} is connected, we can understand infinitesimal channels as the
ending point of an arbitrarily small curve parametrized by $t$, \ie{} channels
$\mcE_i$ in definition~\ref{def:infinitesimal_CP} can be written approximately
as $\mcE_i\approx\id+ L_i\approx \exp\left(L_i \right)$. We have shown that
$L_i$ has Lindblad form, see theorem~\ref{thm:KKSG}. Therefore we have that if
$\mcE\in \cpDiv{}$, it can be written as 
$$\mcE=\prod_i e^{L_i}.$$
Therefore $\mcE$ can be implemented using a CP-divisible dynamical processes.
Bounds of  the convergence ratio using channels of the form $\exp(L_i)$ instead
of general infinitesimal channels, are computed in
Ref.~\cite{Wolf2008}.

Analogous to infinitesimal divisible channels in \cptp{} and its relation with
CP-divisible processes, one can also define the following set involving PTP
maps.
\begin{definition}[Infinitesimal divisible channels in PTP]
Let $\mcL_\text{P}$ be the set containing operations $\mcE\in \cptp{}$ with the
property that $\forall\epsilon>0$ there exist a finite number of channels
$\mcE_i \in \text{PTP}$ such that $\vert\mcE_i-\id\vert<\epsilon$ and
$\mcE=\prod_i \mcE_i$, see~\fref{fig:diagramatic_P_DIV}. It is said that a
channel is infinitesimal divisible in PTP if it belongs to the closure of
$\mcL_\text{P}$. This set is denoted as \pDiv{}.
\label{def:infinitesimal_P}
\end{definition}
\begin{figure} 
\centering
\begin{tikzpicture}
\node at (-2.5,2.5) {\Large $\mcE\in \pDiv{}$};
\node at (-0.8,2) {\includegraphics[scale=0.15]{circuit.png}}; 
\node at (-0.8,2) {\Large $\mcE$};
\node at (2.8,2) {\Huge $=$};
\node at (0,0) {\includegraphics[scale=0.15]{circuit_many.png}}; 
\node at (-2,0) { \Large $\mcE^1$};
\node at (-0.3,0) { \Large $\mcE^2$};
\node at (2.4,0) {\Large $\mcE^{N_\epsilon}$};
\draw (4.5,-1.8) -- (4.5,2.7);
\node at (7,0) {\includegraphics[scale=0.2]{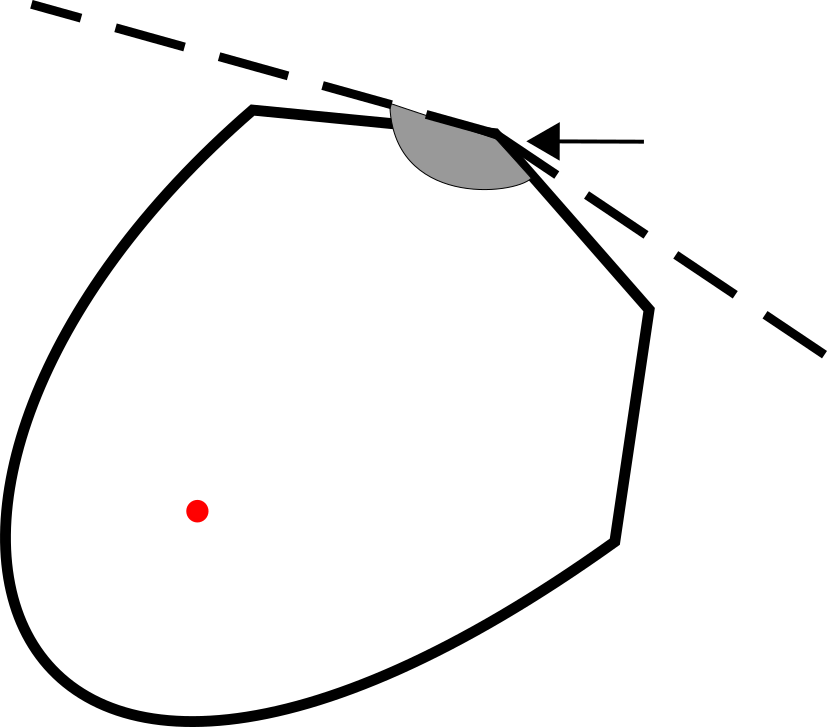}};
\node at (7.8,1.5) {$id$};
\node at (6.2,-0.8) {\Large $\mcE$};
\node at (8.55,1.3) {\Large $\mcE_i$};
\node at (6.9,0.3) {CPTP};
\node at (8.7,-0.5) {PTP};
\end{tikzpicture}
\caption{Diagramatic decomposition of channels belonging to $\mcL_\text{P}$ which closure is \pDiv{}, see definition~\ref{def:infinitesimal_P}. At the left we show the circuit representing the decomposition of $\mcE$ into channels arbitrarily close to the identity map, see figure at the right. In contrast to figure~\ref{fig:diagramatic_CP_DIV}, note that infinitesimal channels can be outside the set of CPTP maps, but inside PTP. \label{fig:diagramatic_P_DIV}}
\end{figure} 
Infinitesimal divisibility in PTP maps is interesting since this kind of maps
can arise in settings where the system is initially correlated with its
surroundings, or if the operation is correlated with the initial
state~\cite{Carteret2008}.

Infinitesimal divisible (either in CPTP and PTP) channels have
non-negative determinant due
to its continuity~\cite{Wolf2008}. To see this note that 
channels arbitrarily close to the identity map have
positive determinant; and by its multiplicative property, the
channel resulting from
the concatenation of infinitesimal channels has non-negative determinant. 
\begin{proposition}[Determinant of infinitesimal divisible channels]
If a quantum map $\mcE$ belongs either to \pDiv{} or \cpDiv{}, then $\det \mcE \geq 0$.
\end{proposition}

It turns out that a non-negative determinant is a sufficient condition for a
channel to be infinitesimal divisible in PTP, see theorem 25 of
Ref.~\cite{Wolf2008}.

Other interesting type of divisibility that in turn forms a subset of \cpDiv{} is
the following~\cite{Wolf2008,Denisov1989}.
\begin{definition}[Infinitely divisible channels]
A quantum channel $\mcE$ is infinitely divisible if
$\forall n\in \mathbb{Z}^+$ $\exists \mcE_n \in \cptp{}$ such that
$\mcE=\left(\mcE_n\right)^n$. This set is denoted as \InftyDiv{}, see~
\fref{fig:diagramatic_Infty_DIV}.
\label{def:infinitely_divisible}
\end{definition}
\begin{figure} 
\centering
\begin{tikzpicture}
\node at (-2.5,2.5) {\Large $\mcE\in \InftyDiv{}$};
\node at (-0.8,2) {\includegraphics[scale=0.15]{circuit.png}}; 
\node at (-0.8,2) { \Large $\mcE$};
\node at (2.8,2) {\Huge $=$};
\node at (0,0) {\includegraphics[scale=0.15]{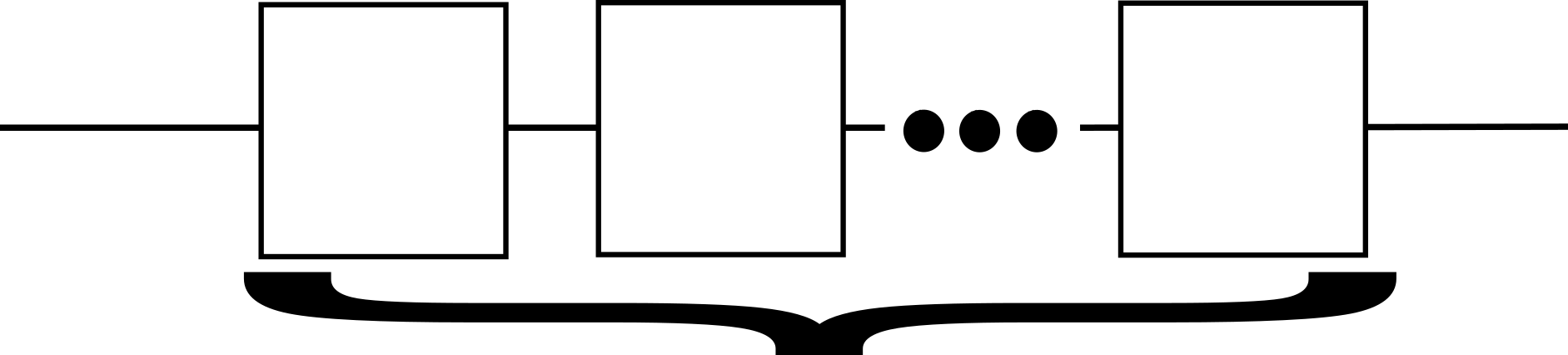}}; 
\node at (-2,0) { \Large $\mcE_n$};
\node at (-0.3,0) {\Large $\mcE_n$};
\node at (2.4,0) {\Large $\mcE_n$};
\node at (0,-1.5) {\Large $n$ times};
\end{tikzpicture}
\caption{Diagrammatic decomposition of channels belonging to \InftyDiv{}, see
definition~\ref{def:infinitely_divisible}. This set contains channels for which
every $n$-root exist and is a valid quantum channel, denoted in the
circuit as $\mcE_n$. \label{fig:diagramatic_Infty_DIV}}
\end{figure} 
This set contains channels for which every $n$-root exists and is a valid
quantum channel. Denisov has shown in
\cite{Denisov1989} that infinitely divisible channels can be
written as $\mcE=\mcE_0 \exp(L)$, 
with $L$ a Lindblad  generator, and an $\mcE_0$ idempotent operator that
fulfills $\mcE_0 L \mcE_0=\mcE_0 L$. In this work we will prove that every
infinitely divisible Pauli channel has the simple form $\exp(L)$.

Let us now introduce the most restricted type of divisibility studied in this
work,
\begin{definition}[Channels belonging to one-parameter semigroups
(L-divisibility)]
Let $\mcL_\text{L}$ be the set containing non-singular operations $\mcE\in
\cptp{}$, such that there exist at least one logarithm, denoted as $L=\log
\mcE$, such that
\begin{equation}
L[\rho]= \rmi [\rho,H]
 +\sum_{i,j} G_{ij} 
     \left( 
         F_{i}\rho F^{\dagger}_{j}
             -\frac{1}{2} \lbrace F^{\dagger}_{j} F_{i},\rho \rbrace 
     \right),
\label{eq:lindblad_from_theorem}
\end{equation}
where $H$ and $G$ are hermitian with $G\geq 0$, and $\lbrace F_i\rbrace_i$ are
bounded operators acting on $\mcT(\mcH)$.
\label{eq:lindblad_definitive}
It is said that a channel is L-divisible if it belongs to the closure of $\mcL_\text{L}$. This set is denoted as \LDiv{}.
\label{def:L_divisibility}
\end{definition}
Analogous to the relation of CP-divisible dynamical maps and its relations with
\cpDiv{}, time-independent Markovian
processes form families of L-divisible channels. The converse is true by
definition.
One of the principal objectives of this work is to construct a test to check
whether a given channel belongs to \LDiv{} or not.
\subsection{Relation between channel divisibility classes} 
\label{subsec:relations}
Let us summarize the introduced divisibility sets and the relations between
them. Since channels belonging to \cpDiv{} can be implemented with
time-dependent Lindblad master equations, and time-independent ones are a
particular case of time dependent ones, we have $\LDiv{} \subset \cpDiv{}$.
Now, since
infinitely divisible channels have the form $\mcE_0 \exp(L)$, channels with
form $\exp(L)$ are a particular case of \InftyDiv{}, therefore $\LDiv{}
\subseteq
\InftyDiv{}$. Also, given that CPTP maps are also PTP, then $\cpDiv{} \subset
\pDiv{}$. Finally, every set except \pDiv{} is subset of \Div{}, given that an
infinitesimal divisible channels in PTP is not necessarily divisible in CPTP
channels. In summary we have~\cite{cirac},
\begin{equation}
  \label{eq:set_relations}
\begin{array}{cccccc}
\InftyDiv{} & \subset & \cpDiv{} & \subset &\Div{} & \\ 
\rotatebox{90}{$\subseteq$} &  &  &  &   \\ 
\LDiv{} & \subset & \cpDiv{} & \subset & \pDiv{} 
\end{array} .
\end{equation}
The intersection of \pDiv{} and \Div{} is not empty since
$\cpDiv{}\subset \Div{}$ and $\cpDiv{} \subset \pDiv{}$, later on we will
investigate if $\pDiv{}\subseteq \Div{}$ or not. A scheme of the inclusions is given in fig.~\ref{fig:setscheme}.
\begin{figure} 
\centering
\includegraphics{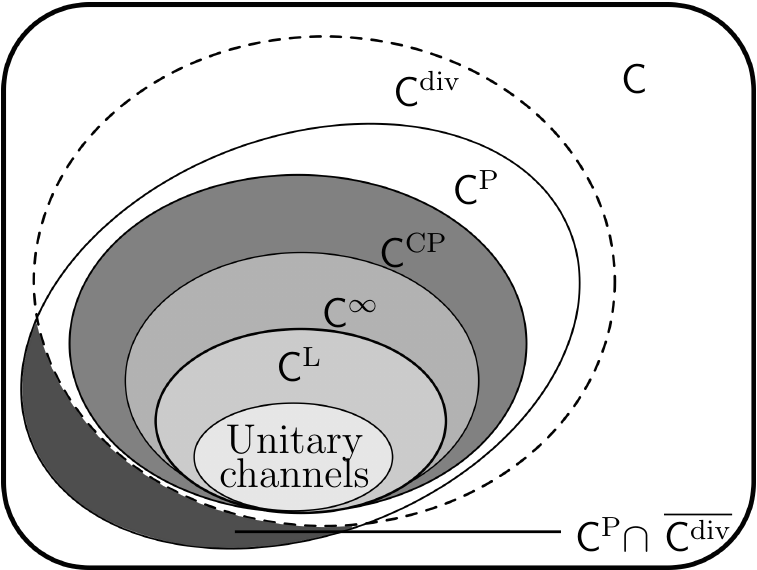}
\caption{
Scheme illustrating the different sets of quantum channels for a given dimension. In particular, the inclusion relations
presented in \eref{eq:set_relations} are depicted. \cptp{} is the set of completely positive trace preserving operations. The divisibility sets depicted are the ones containing channels infinitesimal divisible in CPTP (\cpDiv{}), infinitesimal divisible in PTP (\pDiv{}), infinite divisible (\InftyDiv{}), implementable with Lindblad equations (\LDiv{}), and unitary channels.
\label{fig:setscheme}}
\end{figure} 

%
%
%

\section{Characterization of L-divisibility} 

Deciding L-divisibility is equivalent to proving the existence of a hermiticity
preserving generator, which additionally fulfills the ccp condition, see
proposition~\ref{prop:ccp}.
To prove hermiticity preserving we recall that every \hp{} operator has
a real matrix representation when choosing an hermitian basis, see subsection~\ref{sec:herm_and_trace_less}. Since quantum
channels preserve hermiticity, the problem is reduced to find a real
logarithm $\log\hat\mcE$ given a real matrix $\hat \mcE$, where the hat means that $\mcE$ is written using an hermitian basis. This problem was
already solved by Culver~\cite{Culver1966} who characterized
completely the existence of real logarithms of real matrices. In this work we
restrict the analysis to  diagonalizable channels. The results can be
summarized as follows.
\begin{theorem}[\textbf{Existence of hermiticity preserving generator}]
A non-singular matrix with real entries $\hat \mcE$ has a real generator
(\ie{} a $\log \hat \mcE$ with real entries) if and only if
the spectrum fulfills the following conditions:
i) negative eigenvalues are even-fold degenerate;
ii) complex eigenvalues come in complex conjugate pairs.
\label{thm:culver}
\end{theorem} 

We now discuss the multiplicity of the solutions of $\log \hat \mcE$ and its
parametrization, as finding an appropriate one   is essential to test for the
ccp condition.  If $\hat \mcE$ has positive degenerate, negative, or complex
eigenvalues, its real logarithms are not unique, and are spanned by
\textit{real logarithm branches}~\cite{Culver1966}. The latter are defined
using the real quaternion, which coincides with $\rmi \sigma_y$, using the fact
that $\one=\exp\left( \rmi \sigma_y 2 \pi k\right)$, with $k\in \mathbb{Z}$. In
case of having negative eigenvalues, it turns out that real logarithms always
have a continuous parametrization, in addition to real branches due to the
freedom of the Jordan normal form transformation matrices~\cite{Culver1966}.

To compute the logarithm given a real representation of $\mcE$, \ie{} $\hat
\mcE$, we calculate its Jordan normal form, $J$, such that $\hat \mcE=w J
w^{-1}=\tilde w J \tilde w^{-1}$, where $w = \tilde w K$ and $K$ belongs to a
continuum of matrices that commute with $J$~\cite{Culver1966}. In the case of
diagonalizable matrices, if there are no degeneracies, $K$ commutes with
$\log(J)$. In the case of having degeneracies, matrix $K$ is responsible of the
continuous parametrization of the logarithm.  We compute explicitly the
logarithms for the case of Pauli channels in section~\ref{sec:L_div_Pauli}.

\section{Divisibility of unital qubit channels} 
We will apply various of the results from the literature~\cite{Wolf2008} to
decide if a given unital qubit channel belongs to \LDiv{}, \cpDiv{} and/or
\pDiv{}. The non-unital case will be discussed later.

Before starting with the characterization let us point out the following. From
the definition of divisibility, the concatenation of a given channel with
unitary conjugations (which are infinitesimal divisible) do not change its
divisibility character, except for L-divisibility.
In addition to
this, since unitary conjugations are infinitesimal divisible, they do not
change the infinitesimal divisible character either. We can summarize this in
the following,
\begin{proposition}[\textbf{Divisibility of
special orthogonal normal forms}]
Let $\mcE$ a qubit quantum channel and $\mcD$ its special orthogonal normal
form, $\mcE$ belongs to ${\sf C}^X$ if and only if $\mcD$ does, where $X=\lbrace
\text{``Div'', ``P'', ``CP''} \rbrace$. 
\label{prop:divisibility_using_orthogonal_form}
\end{proposition}
This proposition is in fact a consequence of theorem 17 of Ref.~\cite{Wolf2008}. Notice that this result does not apply for \LDiv{} since conjugating with unitaries breaks the implementability by means of time-independent Lindblad master equations. Thus, if a channel belongs to \LDiv{}, unitary conjugations can bring it
to $\InfDiv{}\setminus\LDiv{}$ and vice versa.

Therefore, by proposition~\ref{prop:divisibility_using_orthogonal_form} and the theorem~\ref{thm:orthogonal_normal_form}, to study \pDiv{} and \cpDiv{} of unital qubit channels, it is enough to study Pauli channels.
\subsection{Channels belonging to \Div{}} 
Divisibility in CPTP of unital qubit channels is completely characterized by
means of theorem~\ref{thm:unital_indivisible}. Therefore the only indivisible
channels lie in the faces of the tetrahedron (without the edges),
see~\fref{fig:tetra}.
\subsection{Channels belonging to \pDiv{}} 
Recalling that all unital qubit channels belonging to \pDiv{} have non-negative
determinant~\cite{cirac}, and using special orthogonal normal forms,  see
theorem~\ref{thm:orthogonal_normal_form}, the condition in terms of its
parameters is given by
\begin{equation}
\lambda_1 \lambda_2 \lambda_3\geq 0.
\label{eq:pdiv_qubits}
\end{equation}
This set is the intersection
of the tetrahedron with the octants 
where the product  of all $\lambda$s is positive. In fact, it consists of 
four triangular bipyramids starting in each vertex of the tetrahedron and
meeting in its center, see~\fref{fig:tetra}.
\begin{figure} 
\centering
\includegraphics{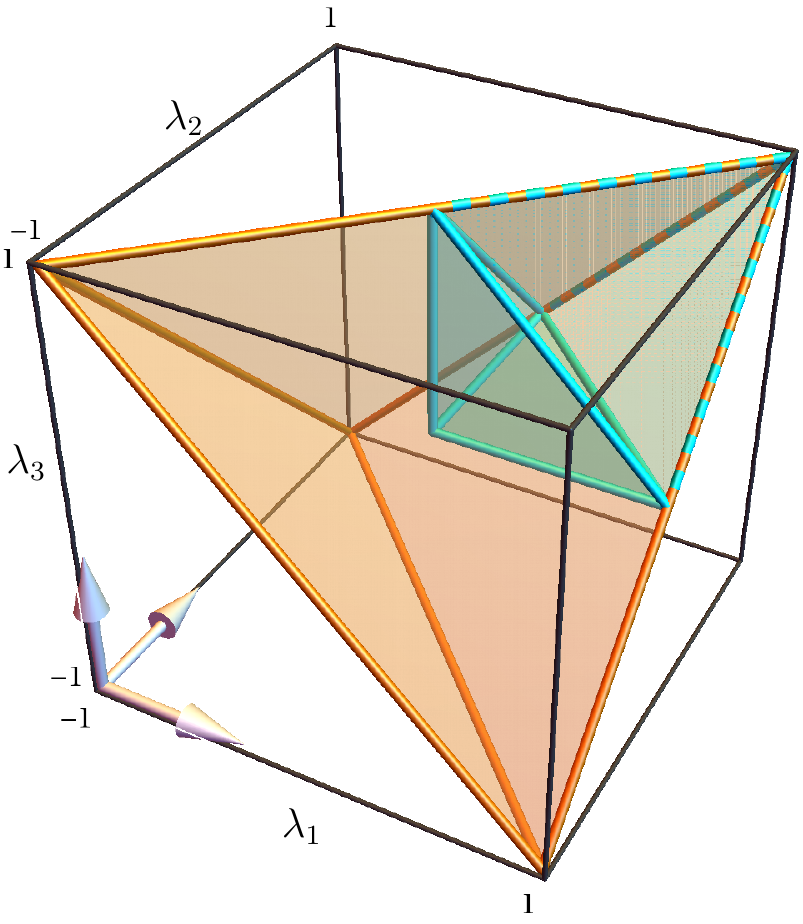}
\caption{
Tetrahedron of Pauli channels, see \fref{fig:simple_tetra}. 
The bipyramid in blue correspond to channels with $\lambda_i>0$ $\forall i$,
\ie{} channels of the positive octant belonging to \pDiv{}. The whole set
\pDiv{} includes other three bipyramids corresponding to the other vertexes of
tetrahedron. This implies that \pDiv{} enjoys the symmetries of the
tetrahedron, see \eref{eq:pdiv_qubits}.  The faces of the
bipyramids matching the corners of the tetrahedron are subsets of the faces
of the tetrahedron, \ie{} contain Kraus rank three channels. Such channels are
both \pDiv{} and \Ind{}, showing that the intersection shown in
\fref{fig:setscheme} is not empty.
\label{fig:tetra}
}
\end{figure} 
Let us study the intersection of this set with the set of unital
entanglement-breaking (EB) channels~\cite{Ziman2005}, see
definition~\ref{def:eb}. In the case of unital qubit channels, the set
entanglement-breaking channels is an octahedron that lie inside the tetrahedron
of unital qubit channels, see~\fref{fig:bypy}. The
inequalities that define such octahedron are the following,
\begin{equation}
\lambda_i \pm\left( \lambda_j + \lambda_k \right) \leq 1,
\label{eq:eb_inequalities}
\end{equation} 
with $i$, $j$ and $k$ all different~\cite{Ziman2005}, together~\eref{eq:complete_positivity_qubit}.
It follows that unital qubit channels that are not achieved by P-divisible
dynamical maps are necessarily entanglement-breaking (see~\fref{fig:bypy} and
\fref{fig:cut1}). In fact this holds for general qubit channels, see
section~\ref{subsec:generalqubitchannels}. 
\begin{figure} 
\centering
\includegraphics{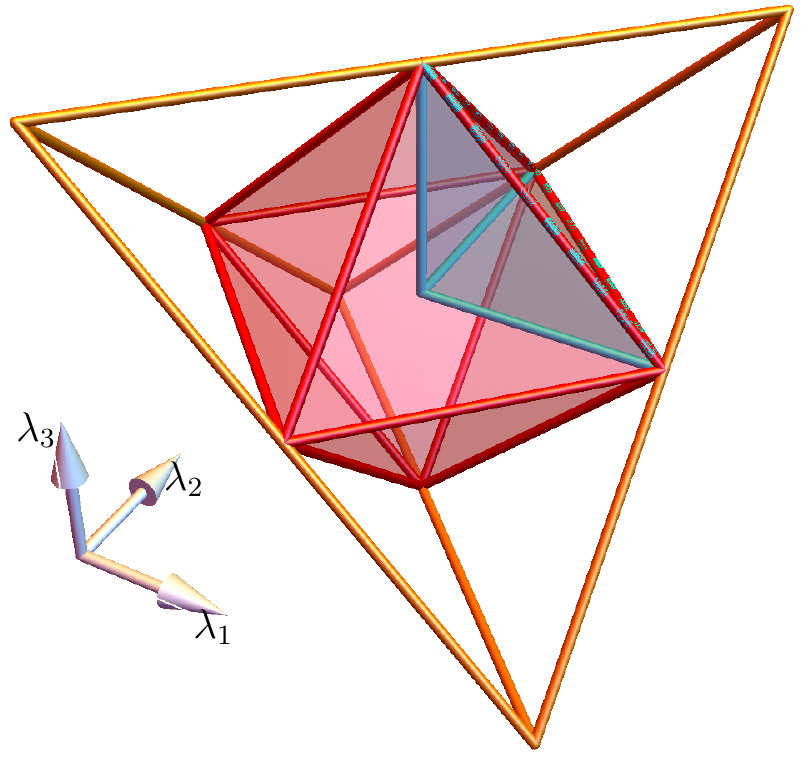}
\caption{Tetrahedron of Pauli channels with the octahedron
of entanglement
breaking channels shown in red, see \eref{eq:eb_inequalities}. The blue pyramid inside the octahedron is the
intersection of the bipyramid shown in \ref{fig:tetra}, with the octahedron.
The complement of the intersections of the four bipyramids forms the set of
divisible but not infinitesimal divisible channels in PTP. Thus, a central feature of
the figure is that the set \Div{}$\setminus$\pDiv{} is always
entanglement-breaking, but the converse is not true. \label{fig:bypy}}
\end{figure} 
\subsection{Channels belonging to \cpDiv{}} 

To characterize CP-divisible channels it is useful to consider the Lorentz
normal form for channels, see theorem~\ref{thm:Lorentz}. 
A remarkable property of the Lorentz normal
decomposition is that it preserves the infinitesimal divisible character of
$\mcE$, see Corollary 13 of~\cite{cirac}. To use it, we resort to theorem 24
of Ref.~\cite{cirac}. Due to the mentioned drawback of Lorentz normal forms,
see appendix~\ref{sec:normal_form}, we must modify such to
theorem to a restricted class of channels. 
\begin{theorem}[Restricted characterization of channels belonging to \cpDiv{}]
A qubit channel $\mcE$ with diagonal Lorentz normal form belongs to \cpDiv{} 
if and only if\\
i) the rank of the form is smaller than three or\\
ii) $s_{\min}^2\geq s_1 s_2 s_3> 0$,
where $s_{\min}$ is the smallest of $s_1$, $s_2$ and $s_3$, see
theorem~\ref{thm:Lorentz}. 
\label{thm:restricted_cp_div}
\end{theorem}
For non-unital Kraus deficient channels, the pertinent theorems are based on
non-diagonal Lorentz normal forms~\cite{Verstraete2002,cirac}. According to our
appendix~\ref{sec:normal_form} such results should
be reviewed and are out of the scope of this work.

Notice that the Lorentz normal form coincides with the special orthogonal
normal form for unital qubit channels. Therefore, by
theorem~\ref{thm:restricted_cp_div}, unital channels belonging to \cpDiv{} are
non-singular with
\begin{equation}
  0<\lambda_1 \lambda_2 \lambda_3\leq \lambda_{\min}^2\,,
\label{eq:cpdivqubit}
\end{equation}
or singular with a matrix rank less than three.
They determine a body
that is symmetric with respect to permutation of Pauli unitary channels (i.e.
in $\lambda_j$), hence, the set of \cpDiv{} of Pauli channels possesses the
symmetries of the tetrahedron. The set \cpDiv{}$\setminus$\LDiv{} is plotted
in~\fref{fig:cp}.
\begin{figure} 
\centering
\includegraphics{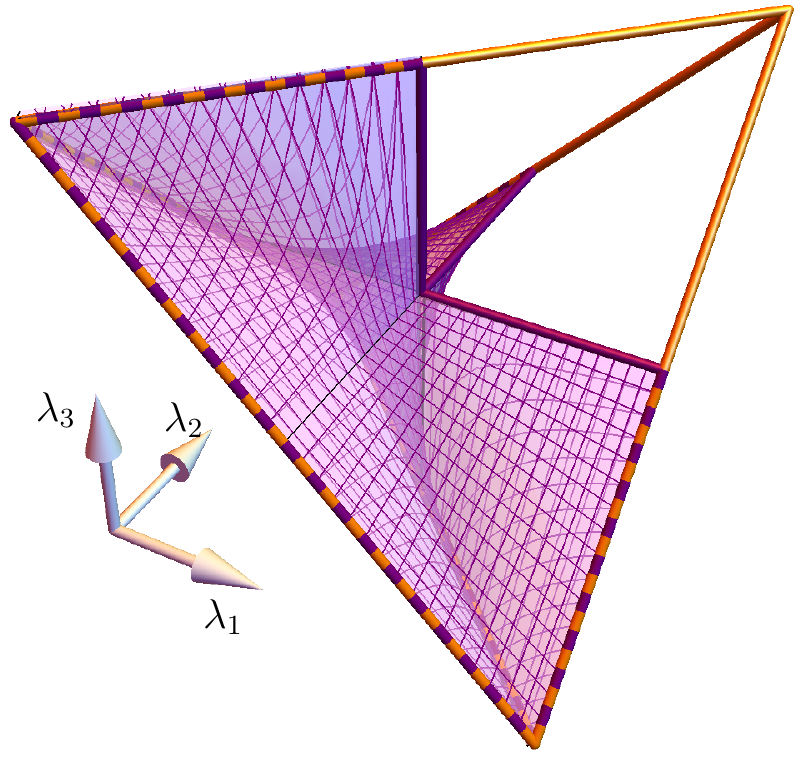}
\caption{Tetrahedron of Pauli channels with part of the set of CP-divisible, 
see \eref{eq:cpdivqubit}, but
not L-divisible channels (\cpDiv{}$\setminus$\LDiv{}) shown in purple. The
whole set \cpDiv{} 
is obtained applying the symmetry transformations of the tetrahedron to the 
purple volume. 
\label{fig:cp}}
\end{figure} 

\subsection{L-divisible unital qubit channels} 
\label{sec:L_div_Pauli}
We restrict our analysis of L-divisibility for two particular sets of unital channels, Pauli channels and a family with complex eigenvalues that will be introduced later. 

\subsubsection{Pauli channels with  non-degenerate positive eigenvalues} 
First let us now derive the conditions for L-divisibility of Pauli channels with
positive eigenvalues $\lambda_1,\lambda_2,\lambda_3$ ($\lambda_0=1$). 
The logarithm of $\mcD$, induced by the principal logarithm of its
eigenvalues is
\begin{equation}
L=K{\rm diag}(0,\log\lambda_1,\log\lambda_2,\log\lambda_3)K^{-1}\,,
\label{eq:logL_positive}
\end{equation}
which is real (hermiticity preserving). In case of no-degeneration
the dependency on $K$ vanishes and $L$ is unique. In such
case the ccp conditions, see theorem~\ref{prop:ccp}, read

\begin{align}
  \log \lambda_i-\log \lambda_j-\log \lambda_k\geq 0 \Rightarrow \frac{\lambda_i}{\lambda_j \lambda_k} \geq 1
\label{eq:hermpres}
\end{align}
for all combinations of mutually different $i,j,k$. This set (channels
belonging to \LDiv{} with positive eigenvalues)  forms a three dimensional
manifold, see \fref{fig:markov}. 
\subsubsection{Pauli channels with degenerate positive  eigenvalues} 
In case of degeneration, let us label the eigenvalues $\eta$, $\lambda$ and $\lambda$. In this case, the real solution for $L$ is not unique
and is parametrized by real branches in the degenerate 
subspace and by the continuous parameters of $K$~\cite{Culver1966}. 
Let us study the principal branch with $K=\one$. \Eref{eq:hermpres} is
then reduced to 
\begin{equation}
\lambda^2\leq \eta \leq 1\;.
\label{eq:ccp_degenerated}
\end{equation}
Therefore, if this inequalities are fulfilled, the generator has Lindblad
form. If not, then \textit{a priori} other branches can fulfill ccp condition and
consequently have a Lindblad form. Thus, Eq.~\eqref{eq:ccp_degenerated} provides
a sufficient condition for the channel to be in \LDiv{}. We will
prove it is also necessary. 

The complete positivity condition requires $\eta,\lambda\leq 1$, thus,
it remains to verify only the condition $\lambda^2\leq \eta$. It holds trivially
for the case $\lambda\leq\eta$. If $\eta\leq\lambda$, then this condition
coincides with the CP-divisibility condition from~\eref{eq:cpdivqubit}.
Since \LDiv{} implies \cpDiv{} the proof is completed. In conclusion,
the condition in~\eref{eq:hermpres} is a necessary and sufficient for a given
Pauli channel with positive eigenvalues to belong to \LDiv{}.

Let us stress that the obtained subset of L-divisible channels
does not possess the tetrahedron symmetries. In fact, composing
$\mcD$ with a $\sigma_z$ rotation $$\mcU_z={\rm diag}(1,-1,-1,1)$$
results in the Pauli channel
$\mcD^\prime={\rm diag}(1,-\lambda_1,-\lambda_2,\lambda_3)$.
Clearly, if $\lambda_j$ are positive ($\mcD$ is L-divisible),
then $\mcD^\prime$ has non-positive eigenvalues. Moreover, if all $\lambda_j$
are different, then $\mcD^\prime$ does not have any
real logarithm, therefore, it cannot be L-divisible.
In conclusion, the set of L-divisible unital qubit channel is
not symmetric with respect to tetrahedron symmetries.
\subsubsection{Pauli channels with negative eigenvalues} 
In what follows we will investigate the case of negative eigenvalues.
Theorem~\ref{thm:culver} implies that that eigenvalues have the form (modulo permutations)
$\eta,-\lambda,-\lambda$, where $\eta,\lambda> 0$. The corresponding
Pauli channels are 
$$\mcD_x={\rm diag}(1,\eta,-\lambda,-\lambda), \ \
\mcD_y={\rm diag}(1,-\lambda,\eta,-\lambda), \ \
\mcD_z={\rm diag}(1,-\lambda,-\lambda,\eta),$$ thus
forming three two-dimensional regions inside the tetrahedron.
Take, for instance, $\mcD_z$ that specifies a plane (inside the tetrahedron)
containing $I$, $\sigma_z$ and completely depolarizing channel
$\mcN={\rm diag}(1,0,0,0)$. The real logarithms for this case are given by
\begin{equation}
L=K\left(
\begin{array}{cccc}
0 & 0 & 0 & 0 \\ 
0 & \log( \lambda ) & (2 k +1) \pi & 0 \\ 
0 & -(2k +1)\pi & \log( \lambda ) & 0 \\ 
0 & 0 & 0 & \log( \eta )
\end{array} 
\right) K^{-1},
\label{eq:Lfornegative}
\end{equation}
where $k\in \mathbb{Z}$ and $K$, as mentioned above, belongs to a continuum of
matrices that commute with $\mcD_z$. Note that $L$ is always non-diagonal. For
this case (similarly for $\mcD_x$ and $\mcD_y$) the ccp condition reduces
again to conditions specified in Eq.~\eqref{eq:ccp_degenerated}. Using the same
arguments one arrives to more general conclusion: 
\begin{theorem}[L-divisibility of Pauli channels]
Let $\mcE$ be a non-singular Pauli channel. It belongs to \LDiv{} if and only if its non-trivial eigenvalues fulfill
\begin{equation}
\frac{\lambda_i}{\lambda_j \lambda_k}\geq 0
\label{eq:L_div_definitive}
\end{equation}
for $i,j$ and $k$ mutually different.
\end{theorem}
This is one of the central results of this work, and it implies that for
testing L-divisibility of Pauli channels, it is enough to consider the
principal real logarithm branch and $K=\one$. The singular cases are included
in the closure of channels fulfilling~\eref{eq:L_div_definitive}.  The set of
L-divisible Pauli channels is illustrated in~\fref{fig:markov}.
\begin{figure} 
\centering
\includegraphics{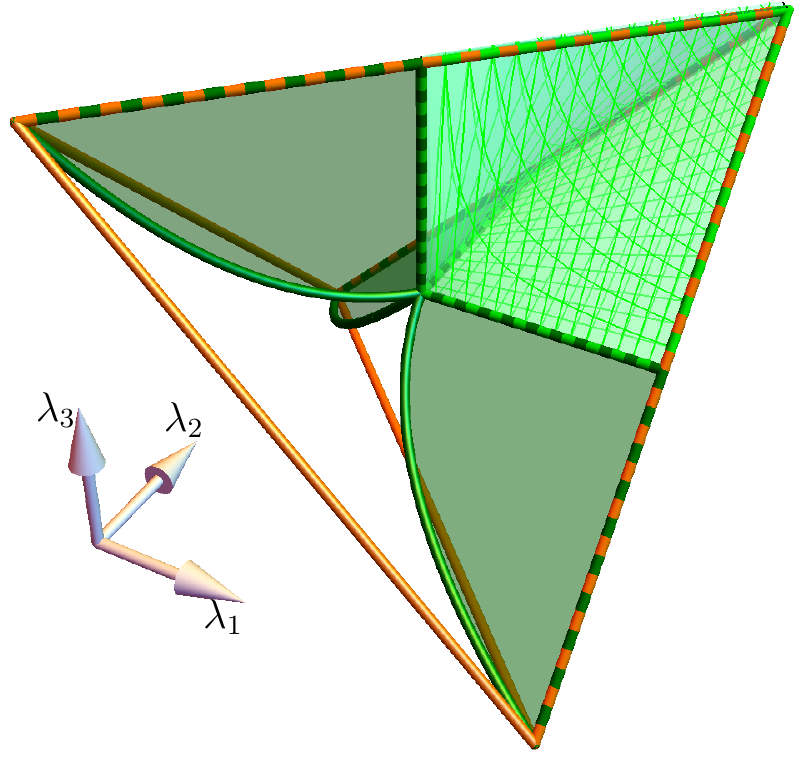}
\caption{Tetrahedron of Pauli channels with the set of L-divisible channels
(or equivalently infinitely divisible, see Theorem~\ref{thm:pauli_infinity})
shown in green, see  equations (\ref{eq:hermpres}) and
(\ref{eq:ccp_degenerated}). The solid set corresponds to channels with
positive eigenvalues, and the 2D sets correspond to the negative eigenvalue
case. The point where the four sets meet correspond to the \textit{total
depolarizing} channel. Notice that this set {\it does not} have the symmetries
of the tetrahedron. \label{fig:markov}}
\end{figure} 
To get a detailed picture of the position and inclusions of the divisibility
sets, we illustrate in \fref{fig:cut1} two slices of the tetrahedron where
different types of divisibility are visualized. Notice the non-convexity of the
considered divisibility sets.
\definecolor{colora}{rgb}{1., 0.5, 0.}
\definecolor{colorb}{rgb}{0., 1., 1.}
\definecolor{colorc}{rgb}{1., 0., 0.}
\definecolor{colord}{rgb}{0., 1., 0.}
\definecolor{colore}{rgb}{0., 0.5, 0.}
\definecolor{colorf}{rgb}{0.5, 0., 0.5}
\definecolor{colorg}{rgb}{0.5, 0.5, 0.5}
\definecolor{colorh}{rgb}{0., 0., 1.}

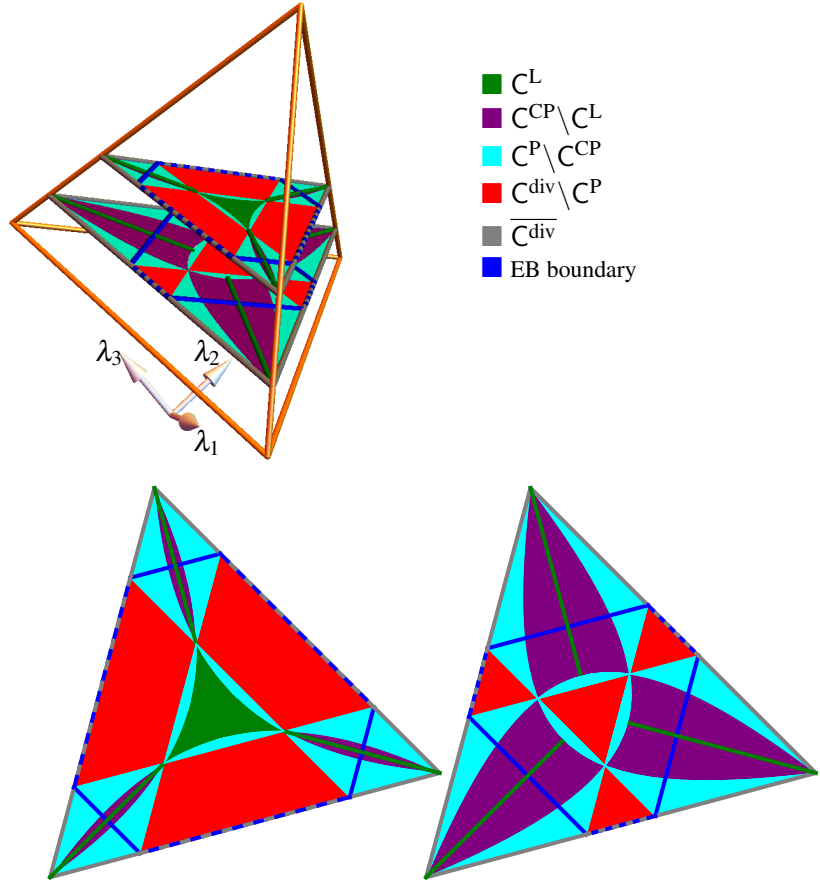
\begin{figure*} 
\centering
\begin{tikzpicture}
\coordinate (beg_1) at (7,8);
\coordinate (delta) at (0,-.5);
\node at (4,6) {\includegraphics{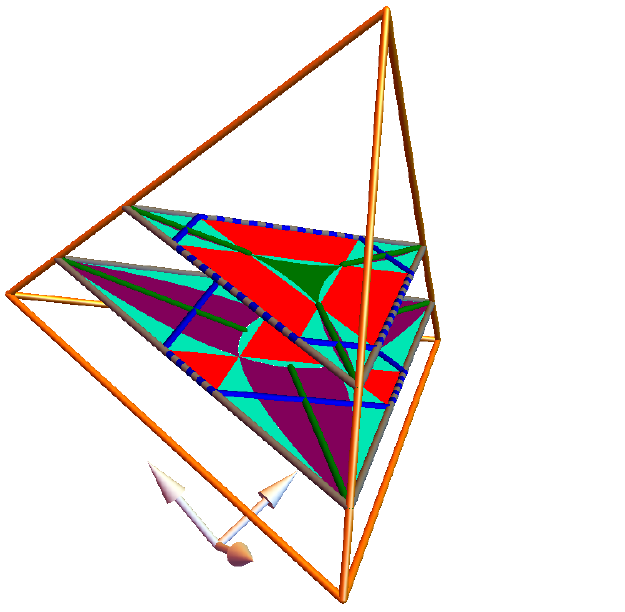}};
\node at (4,0) {\includegraphics[scale=0.9]{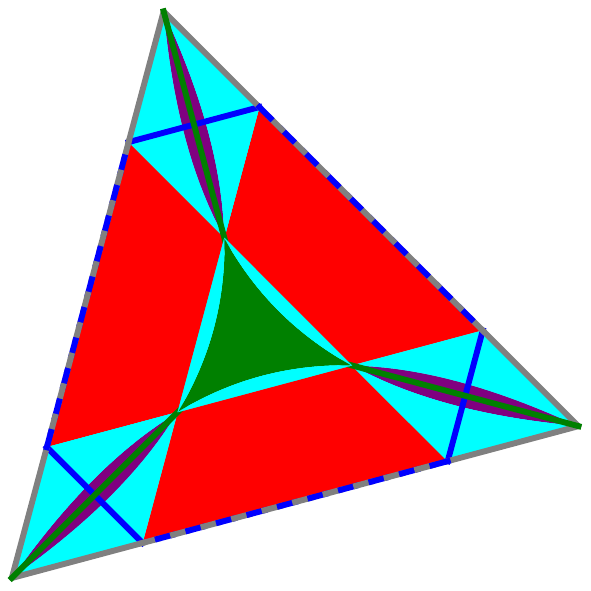}};
\node at (9,0) {\includegraphics[scale=0.9]{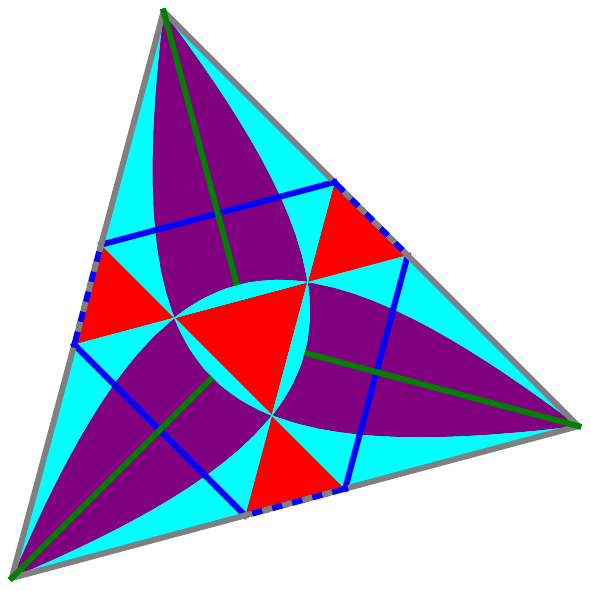}};
\node at (-0.5+4,-2.8+6) {$\lambda_1$};
\node at (-0.5+4,-1.6+6) {$\lambda_2$};
\node at (-1.8+4,-1.6+6) {$\lambda_3$};
\node[anchor=west]  at (beg_1) {{\color{colore} $\blacksquare$} \LDiv};
\node[anchor=west]  at ($(beg_1)+1*(delta)$) {{\color{colorf} $\blacksquare$} \cpDiv$\setminus$\LDiv};
\node[anchor=west]  at ($(beg_1)+2*(delta)$) {{\color{colorb} $\blacksquare$} \pDiv$\setminus$\cpDiv};
\node[anchor=west]  at ($(beg_1)+3*(delta)$) {{\color{colorc} $\blacksquare$} \Div$\setminus$\pDiv};
\node[anchor=west]  at ($(beg_1)+4*(delta)$) {{\color{colorg} $\blacksquare$} \Ind};
\node[anchor=west]  at ($(beg_1)+5*(delta)$) {{\color{colorh} $\blacksquare$} \footnotesize EB boundary};
\end{tikzpicture}
\caption{We show two slices of the unitary tetrahedron (figure in the left) determined by 
$\sum_i \lambda_i=0.4$ (shown in the center) and $\sum_i \lambda_i=-0.4$ (shown in the right). The non-convexity of the
divisibility sets  can be seen, including the set of indivisible channels. The convexity
of sets ${\sf C}$ and entanglement breaking channels can also be noticed in the slices. A central
feature is that the set $\Div{}\setminus\pDiv{}$ is always inside the
octahedron of entanglement breaking channels.  \label{fig:cut1}}
\end{figure*} 
\subsubsection{Family of unital channels with complex eigenvalues} 
To give an insight to the case of complex eigenvalues, consider the following
family of channels with real logarithm, written in the Pauli basis,
\begin{equation}
\mcE_\text{complex}=\left( \begin{array}{cccc}
1 & 0 & 0 & 0\\
0 & c & 0 & 0 \\
0 & 0 & a & -b \\  
0 & 0 & b & a
\end{array}  \right).
\label{eq:real_jordan_form}
\end{equation}
The latter has complex eigenvalues $a\pm i b$ and a real one $c>0$, together
with the trivial eigenvalue equal to $1$. Its real logarithm is given by,
\begin{align*}
 L &=K\log \left(\mcE_\text{complex} \right)_kK^{-1} \\
&=K\left( \begin{array}{cccc}
0 & 0 & 0 & 0\\
0 & \log(c) & 0 & 0 \\
0 & 0 & \log(\vert z \vert) & \arg(z) +2 \pi k \\  
0 & 0 & -\arg(z) -2 \pi k & \log(\vert z \vert)
\end{array}  \right)K^{-1}
\end{align*}
with $z=a+ib$.
The non-diagonal block of the logarithm has the same
structure of $\mcE_\text{complex}$, so $K$ also commutes with
$\log(\mcE_\text{complex})_k$, leading to a countable parametric space of hermitian preserving
generators. The ccp condition, see proposition~\ref{prop:ccp}, is reduced to
\begin{equation}
a^2+b^2\leq c \leq 1.
\label{eq:ccp_complex}
\end{equation}
Note that it does not depend on $k$ and the second inequality is always
fulfilled for CPTP channels.
The set containing them is shown
in~\fref{fig:tetra_complex_eigenvalues}.
\begin{figure} 
\centering
\includegraphics{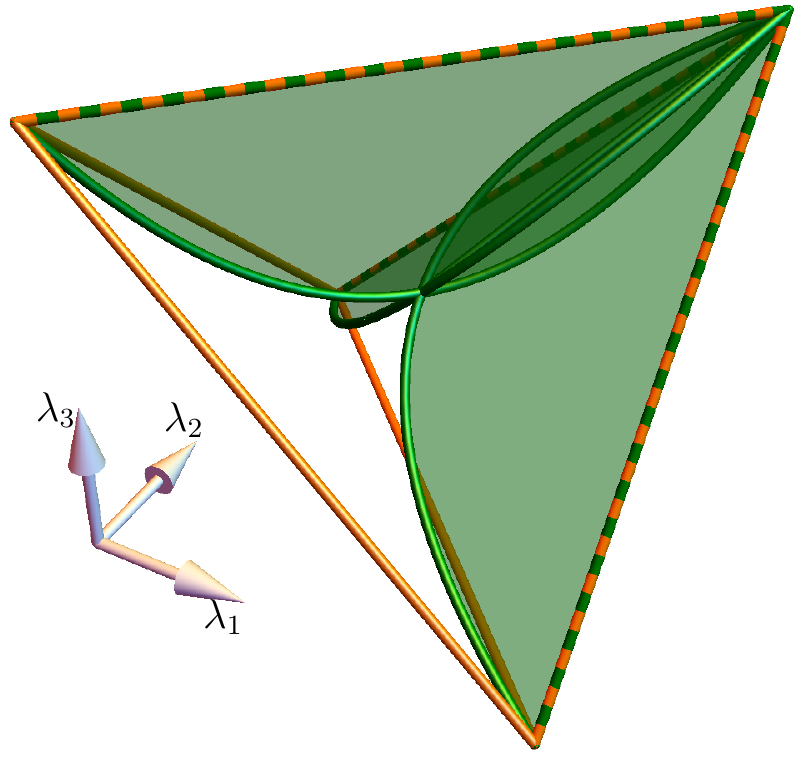}
\caption{
Tetrahedron of Pauli channels, with qubit unital L-divisible
channels of the form $\hat\mcE_\text{complex}$ (see main text). Note that the set
does not have the symmetries of the tetrahedron.
\label{fig:tetra_complex_eigenvalues}}
\end{figure} 

\subsection{Relation of L-divisibility with other divisibility classes} 
Consider a Pauli channel with
$0<\lambda_{\min}=\lambda_1\leq\lambda_2\leq\lambda_3<1$, thus the condition
$\lambda_1\lambda_2\leq \lambda_3$ trivially holds. Since
$\lambda_1\lambda_2\leq \lambda_1\lambda_3\leq \lambda_2\lambda_3\leq\lambda_2$,
it follows that $\lambda_1\lambda_3\leq\lambda_2$, thus, two (out of three)
L-divisibility conditions hold always for Pauli channels with
positive eigenvalues. Moreover, one may observe that CP-divisibility
condition \eref{eq:cpdivqubit} reduces to one
of L-divisibility conditions $\lambda_2\lambda_3\leq\lambda_1$.
In conclusion, the conditions of CP-divisibility and L-divisibility
for Pauli channels with positive eigenvalues coincide, thus,
in this case \cpDiv{} implies \LDiv{}.

Concatenating (positive-eigenvalues) L-divisible Pauli channels with
$\mcD_{x,y,z}$, one can generate the whole set of \cpDiv{} Pauli channels. In
other words, $\mcD_{x,y,z}$ brings the body (with vertex in $\id$) shown
in~\fref{fig:markov} to the bodies shown in \fref{fig:cp} (with vertexes
$x,y,z$). Therefore we can formulate the following theorem:
\begin{theorem}[Infinitesimal divisible unital channels]
Let $\mcE^{\text{CP}}_{\text{unital}}$ be an arbitrary infinitesimal divisible unital qubit channel. There exists at least one L-divisible Pauli channel $\tilde \mcE$, and two unitary conjugations $\mcU_1$ and $\mcU_2$, such that
$$\mcE^\text{CP}_{\text{unital}}=\mcU_1  \tilde \mcE \mcU_2\,.$$
Notice that if $\mcE^\text{CP}_{\text{unital}}$ is invertible, $\tilde \mcE=e^L$.
\end{theorem}

Let us continue with another equivalence relation valid for Pauli channels.
In general,
$\text{\LDiv{}}\subset\text{\InftyDiv{}}$; however, for Pauli channels these
two 
subsets coincide.

\begin{theorem}[Infinitely divisible Pauli channels]
The set of L-divisible Pauli channels is
equivalent to the set of infinitely divisible Pauli channels.
\label{thm:pauli_infinity}
\end{theorem}
\begin{proof}
A channel
is infinitely divisible if and only if it can be written as $\mcE_0 e^L$, where $\mcE_0$
is an idempotent channel satisfying $\mcE_0 L \mcE_0=\mcE_0 L$ and $L$ has
Lindblad form, see definition~\ref{def:infinitely_divisible}.
The only idempotent qubit channels are contractions of the Bloch
sphere into single points, diagonalization channels $\mcE_\diag$ transforming
Bloch sphere into a line connecting a pair of basis states, and the identity
channel. 
Among the single-point contractions, the only one that is a Pauli channel is
the contraction of the Bloch sphere into the complete mixture; let us call it
$\mcN$. Notice that
$\mcE = \mcN e^L=\mcN$ for all $L$. The channel $\mcN$
belongs to the closure of \LDiv{}, because a sequence of channels $e^{L_n}$
with $\hat{L}_n=\diag \left(0,-n,-n,-n \right)$ converges to $\hat\mcN$ in the
limit $n\to\infty$. 
For the case of $\mcE_0$ being the identity channel we have $\mcE=e^L$, thus,
trivially such infinitely divisible channel $\mcE$ is in \LDiv{} too. 
It
remains to analyze the case of diagonalization channels. First, let us
note that the matrix of $e^{\hat{L}}$ is necessarily of full rank,
since ${\rm det}\hat{\mcE}\neq 0$. It follows that the matrix
$\hat{\mcE}=\hat\mcE_{\diag} e^{\hat{L}}$ has rank two as $\hat\mcE_{\diag}$
is a rank two matrix, thus, it takes one of the following forms
$\hat{\mcE}_{x}^\lambda=\diag \left(1,\lambda,0,0\right)$,
$\hat{\mcE}_{y}^\lambda=\diag \left(1,0,\lambda,0\right)$,
$\hat{\mcE}_{z}^\lambda=\diag \left(1,0,0,\lambda\right)$.
The infinitely divisibility implies $\lambda>0$ in order to keep
the roots of $\lambda$ real. In what follows we will show that
$\hat{\mcE}_z$ belongs to (the closure of) \LDiv{}. Let us
define the channels $\hat\mcE_z^{\lambda,
\epsilon}=\diag\left(1,\epsilon,\epsilon,\lambda \right)$ with $\epsilon>0$.
The complete positivity and ccp conditions
translate into the inequalities $\epsilon\leq \frac{1+\lambda}{2}$
and $\epsilon^2\leq \lambda$, respectively; therefore one can always find an
$\epsilon>0$ such that $\hat \mcE_z^{\lambda, \epsilon}$ is a L-divisible
channel. 
If we choose $\epsilon=\sqrt{\lambda}/n$ with $n\in
\mathbb{Z}^+$, the
channels $\hat
\mcE_{z,n}=\diag\left(1,\sqrt{\lambda}/n,\sqrt{\lambda}/n,\lambda \right)$ form
a sequence of L-divisible channels converging to $\hat\mcE_z^\lambda$ when
$n\to\infty$. The analogous reasoning implies that $\hat{\mcE}_x^\lambda,
\hat{\mcE}_y^\lambda\in\LDiv{}$ too. Let us note that one parameter
family $\mcE_z$ are convex combinations of the complete
diagonalization channel
$\hat{\mcE}_z^1=\diag \left(1,0,0,1\right)$ 
and the complete mixture contraction $\hat \mcN$.
This completes the proof.
\end{proof}

Finally, let us remark that using the theorem~\ref{thm:unital_indivisible} we
conclude that the intersection $\pDiv{}\cap \Ind{}$ depicted
in~\fref{fig:setscheme} is not empty. To show this, notice that there are channels with positive
determinant inside the faces (\ie{} \pDiv{} but not \Div{}), for example $\diag
\left(1,\frac{4}{5},\frac{4}{5},\frac{3}{5}\right)$. Therefore we conclude that
up to unitaries, $\pDiv{}\cap \Ind{}$ corresponds to the union of the four faces
faces of the tetrahedron minus the faces of the octahedron that intersect with
the faces of the tetrahedron, see \fref{fig:bypy}. We have to remove such
intersection since it corresponds to channels with negative determinant, and
thus not in \pDiv{}.

\section{Non-unital qubit channels} 
\label{subsec:generalqubitchannels}
Similar to unital channels, using
theorem~\ref{prop:divisibility_using_orthogonal_form} we are able to
characterize \Div{}, \pDiv{} and \cpDiv{} by studying special orthogonal normal
forms. Such channels are characterized by $\vec \lambda$ and
$\vec \tau$, see~\eref{eq:orthogonalform}. Thus, we can study if a channel is
\Div{} by computing the rank of its Choi matrix, see theorem~\ref{thm:divisible}. For this case algebraic
equations are in general fourth order polynomials. In fact, 
in Ref.~\cite{Lukasz} a condition in terms of the eigenvalues and $\vec \tau$
is given.  For special cases, however, we can
obtain compact expressions, see~\fref{fig:cutnonunital2}. The characterization
of \pDiv{} is given by again by the condition $\lambda_1 \lambda_2 \lambda_3\geq 0$. \cpDiv{} is tested, for full Kraus rank non-unital channels,
using theorem~\ref{thm:restricted_cp_div}, the calculation of $s_i$'s is done using the algorithm presented in Ref.~\cite{Verstraete2001}. For the characterization of \LDiv{} we use theorem~\ref{thm:culver} and evaluate numerically the cpp condition.

We can plot illustrative pictures even though the whole space of qubit
channels has $12$ parameters. This can be done using special orthogonal normal forms
and fixing $\vec \tau$, exactly in the same way as the unital case.
Recall that unitaries only modify \LDiv{}, leaving the shape of other sets
unchanged. CPTP channels are represented as a volume inside the
tetrahedron presented in \fref{fig:tetra}, see \fref{fig:cutnonunital2}.  In
the later figure we show a slice corresponding to $\vec \tau=\left(1/2, 0,0
\right)^\text{T}$. Indeed, it has the same structure of the slices for the
unital case, but deformed, see \fref{fig:cut1}. 
A difference with respect to the unital case is that L-divisible
channels with negative eigenvalues (up to unitaries) are not completely inside
CP-divisible channels. A part of them are inside the \pDiv{}
channels.
\begin{figure} 
\centering
\includegraphics{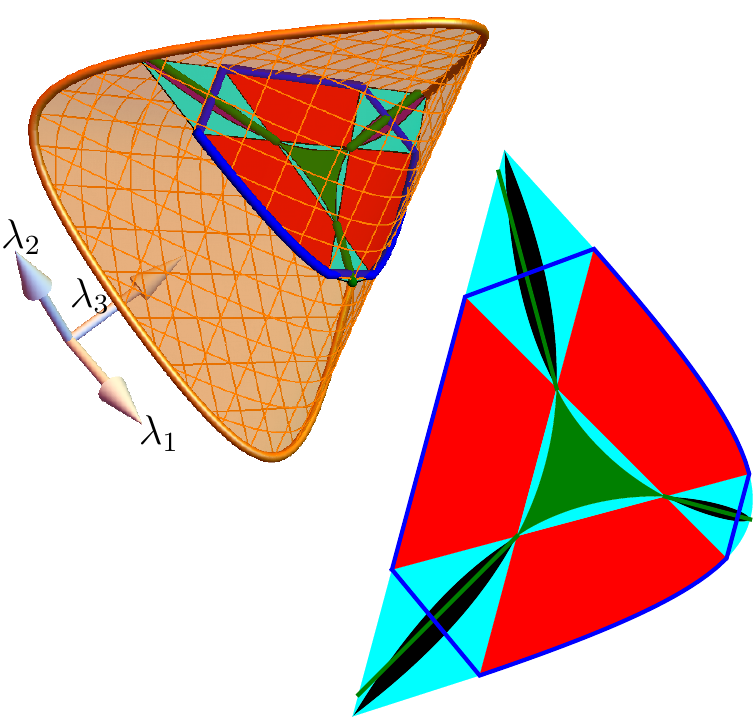}
\caption{(left) Set of non-unital unital channels up to unitaries, defined by
$\vec \tau=(1/2,0,0)$, see \eref{eq:orthogonalform}. This set lies inside the
tetrahedron. For this particular case the CP
conditions reduce to the two
inequalities $2 \pm 2\lambda_1 \ge \sqrt{1+4(\lambda_2 \pm \lambda_3)^2}$.
A cut corresponding to $\sum_i \lambda_i = 0.3$ is presented inside and in the right, see \fref{fig:cut1} for the color coding. The structure of divisibility sets presented here has basically the same structure as for the unital case except for \LDiv{}. A part of the channels with negative eigenvalues belonging to \LDiv{} lies outside $\cpDiv\setminus \LDiv$, see green lines. As for the unital case a central feature is that the channels in $\Div{}\setminus \pDiv{}$ are entanglement breaking channels. Channels in the boundary are not characterized due to the restricted character of Theorem~\ref{thm:Lorentz}.
 \label{fig:cutnonunital2}}
\end{figure} 

A central feature of Figs. \ref{fig:cut1} and \ref{fig:cutnonunital2} is that
the set $\Div{}\setminus\pDiv{}$ is inside the convex slice of the set of
entanglement breaking channels (deformed octahedron). Indeed, we can proof the
following theorem. 

\begin{theorem}[Entanglement-breaking channels and divisibility]
  Consider a qubit channel $\mcE$.
  If $\det\hat\mcE<0$, then $\mcE$ is entanglement-breaking,
  i.e. all qubit channels outside $\pDiv{}$ are entanglement breaking.
\label{thm:EB}
\end{theorem}

Before introducing the proof, let us first show that the proper orthochronous
Lorentz transformations present in the Lorentz normal decomposition for
channels, see \sref{sec:lorentz}, correspond to 1wSLOCC at the level of their
\Jami{} state. Consider a channel $\mcE$ and its Lorentz normal form $\tilde
\mcE$ given by
\begin{equation}
\tilde \mcE=\alpha \mcF_2 \mcE \mcF_1,
\end{equation}
where
$$\mcF_i:\rho \mapsto X_i\rho X_i^{\dagger}, \text{ with } X_i \in \text{SL}(2,\mathbb{C}), \ \ i=1,2,$$
and $\alpha$ is a constant that must be included for $\tilde \mcE$ to be trace preserving,
We showed already that $\text{SL}(2,\mathbb{C})$ is a double cover of $\text{SO}^+(3,1)$, \ie{} $\mcF_i$'s correspond to the proper orthochronous Lorentz transformations of the decomposition.

Now let us compute the \Jami{} state of $\tilde \mcE$, $\tilde \tau$, using the Kraus decomposition of $\mcE$~\cite{wolfnotes},
\begin{align}
\tilde \tau&=\alpha \left(\id_2\otimes \mcF_2 \mcE \mcF_1\right)[\omega] \nonumber \\
&=\alpha\sum_i \left(\one\otimes X_2\right)\left(\one\otimes K_i\right)\left(\one\otimes X_1\right)\proj{\Omega}{\Omega}  \left(\one\otimes X_1^\dagger{}\right)\left(\one\otimes K_i^\dagger{}\right)\left(\one\otimes X_2^\dagger\right) \nonumber \\
&=\alpha\sum_i \left(X_1^\text{T}\otimes X_2\right)\left(\one\otimes K_i\right)\proj{\Omega}{\Omega})\left(\one\otimes K_i^\dagger{}\right)\left(\overline{X_1}\otimes X_2^\dagger\right)\nonumber\\
&=\alpha (X_1^T \otimes X_2) \tau (X_1^T \otimes X_2)^{\dagger},
\label{eq:choi_locc}
\end{align}
where $\lbrace K_i \rbrace_i$ are a choice of Kraus operators of $\mcE$ and $$
\tau=\sum_i \left(\one\otimes K_i\right)\proj{\Omega}{\Omega}\left(\one\otimes
K_i^\dagger{}\right)$$ its \Jami{} matrix. Here, $\ket{\Omega}$ is the
Bell state between two copies of the system, in this case a qubit, for which
the identity $A\otimes \one \ket{\Omega}=\one \otimes
A^\text{T}\ket{\Omega}$ holds.
It can be observed that~\eref{eq:choi_locc} has exactly the form of the
normalized 1wSLOCC scheme, where $\alpha$ turns to be the normalization
constant, see~\eref{eq:det_press}, \ie{} $\tr \tilde \tau=1$. That's why we
have introduced it at the first place. Now let us proceed with the proof of
theorem~\ref{thm:EB},

\begin{proof}
Let $\mcE$ be a qubit channel with negative determinant and $\hat\mcE$ its
matrix representation using the Pauli basis, see
\eref{eq:qubit_channel_pauli_basis}. Recall that the matrix $R$ defining the
\Jami{} state of $\mcE$, 
$$\tau_{\mcE}=\frac{1}{4}\sum^{3}_{jk}R_{jk}\sigma_j\otimes \sigma_k,$$
and $\hat\mcE$ are related by
$$R=\hat\mcE\Phi_{\text{T}},$$
where $\Phi_{\text{T}}=\diag \left( 1,1,-1,1 \right)$.
It follows immediately that $R$ has positive determinant,
$$\det R=-\det\hat{\mcE}>0,$$
since $\det \Phi_\text{T}=-1$.
Using the aforementioned Lorentz normal decomposition for matrix $R$, we have
$$R=L^{\text{T}}_1 \tilde R L_2$$ where $\det L_{1,2}>0$ and $\det \tilde R>0$.
Stressing that transformations
$L_{1,2}$ correspond to 1wSLOCC (see~\eref{eq:choi_locc}), then
$\tilde R$ parametrizes an unnormalized two-qubit state.

Let us first discuss the case when $\tilde R$ is diagonal.
The channel corresponding to $\tilde R$ (in the Pauli basis)
is 
$$\hat {\mathcal G} = \tilde R \Phi_{\text{T}}/\tilde
R_{00},$$
where $R_{00}=\tr \tilde R=\tr \tau_{\mathcal G}$. Since $\tilde R$ is diagonal, then ${\mathcal G}$
is a Pauli channel with $\det\hat {\mathcal G}<0$. A Pauli channel
has a negative determinant, if either all $\lambda_j$ are negative,
or exactly one of them is negative. In Ref.~\cite{Ziman2005} it has been
shown that the set of channels with $\lambda_j<0 \ \ \forall j$ are
entanglement breaking channels. Now, using the symmetries of the
tetrahedron, one can generate all channels with negative determinant
by concatenating this set with the Pauli rotations. Therefore
every Pauli channel with negative determinant is entanglement breaking, thus,
$\tau_{\mathcal G}$ is separable. 
Given that LOCC operations
can not create entanglement~\cite{Horodecki}, we have that $\tau_\mcE$
is separable, therefore $\mcE$ is entanglement breaking.
%

The case when $\tilde R$ is non-diagonal corresponds to Kraus deficient channels (the matrix rank of \ref{eq:state_normal_form_singular} is at most $3$). This case can be analyzed as follows. 
Since the neighborhood of any Kraus deficient channel with negative determinant contains full Kraus rank channels, by continuity of the determinant such channels have negative determinant too. The last ones are entanglement breaking since full Kraus rank channels have diagonal Lorentz normal form. Therefore, by continuity of the concurrence~\cite{Ziman2005}, Kraus deficient channels with negative determinant are entanglement breaking. 
\end{proof}
\section{Divisibility transitions and examples with dynamical processes} 
\label{sec:jumps}
The aim of this section is to use illustrative examples of quantum dynamical
processes to show transitions between divisibility types of the instantaneous
channels. From the slices shown above (see figures \ref{fig:cut1} and
\ref{fig:cutnonunital2}) it can be noticed that every transition between the
studied divisibility types is permitted. This is due to the existence of common
borders between all combinations of divisibility sets; we can think of any
continuous line inside the tetrahedron~\cite{filippov2017} as describing some
quantum dynamical map. 

We analyze two examples. The first is an implementation of the approximate NOT
gate, $\mcA_\text{NOT}$ throughout a specific collision model~\cite{Rybar2012}. 
The second is the well known setting of a two-level atom interacting with a quantized mode of an optical cavity \cite{Haroche06}. We define a simple function that assigns a particular value to a channel $\mcE_t$ according to divisibility hierarchy, i.e.
\begin{equation}
\DivFunc[\mcE]=\left\{
\begin{array}{cl}
1&\ {\rm if}\ \mcE\in{\sf C}^{\rm L}\,,\\
2/3&\ {\rm if}\ \mcE\in{\sf C}^{\rm CP}\setminus{\sf C}^{\rm L}\,,\\
1/3&\ {\rm if}\ \mcE\in{\sf C}^{\rm P}\setminus{\sf C}^{\rm CP}\,,\\
0&\ {\rm if}\ \mcE\in{\sf C}\setminus{\sf C}^{\rm P}\,.
\end{array}
\right.
\label{eq:delta_function}
\end{equation}
A similar function can be defined
to study the transition to/from the set of entanglement-breaking channels, \ie{}
\begin{equation}
\chi[\mcE]=\left\{
\begin{array}{cl}
1&\ {\rm if}\ \mcE\text{ is entanglement breaking}\,,\\
0&\ {\rm if}\ \mcE\text{ if not}.
\end{array}
\right.
\label{eq:chi_function}
\end{equation}

The quantum NOT gate is defined as $\text{NOT}:\rho \mapsto \one-\rho$, \ie{}
it maps pure qubit states to its orthogonal state. Although this map transforms
the Bloch sphere into itself it is not a CPTP map, and the closest CPTP map
is $\mcA_{\rm NOT}:\rho \mapsto
(\one-\rho)/3$. This is a rank-three qubit unital channel, thus,
it is indivisible \cite{cirac}. Moreover, $\det \mcA_{\rm NOT}=-1/27$ implies
that this channel is not achievable by a P-divisible dynamical map.  It is worth
noting that $\mcA_\text{NOT}$ belongs to \Ind{}.

A specific collision model was designed in Ref.~\cite{Rybar2012} simulating
stroboscopically a quantum dynamical map that implements the approximate quantum NOT
gate, $\mcA_{\rm NOT}$, in finite time. It is constructed in the following way, any stroboscopically simulable channel can be written as 
$$\mcE_n=\tr_\text{E} \left[\left( U_1 \dots U_n \right) \rho \otimes \omega_n \left( U_1 \dots U_n \right)^\dagger{} \right],$$
where $U_j=U\otimes \one_{\overline{j}}$ is the unitary corresponding to the bipartite collision with the $j$th particle, the identity $\one_{\overline{j}}$ is applicated in all particles except particle $j$. The density matrix $\omega_n$ is the state of the particles that ``collide'' with the central system, they are though as the environment. It can be shown that in the limit $n\to \infty$, the change of the central system from the $j$th to the $(j+1)$th interaction can be made arbitrarily small~\cite{Rybar2012}. Thus, substituting the integer index $j$ by the continuous parameter $t$, we have,
\begin{equation}
\mcE_t[\varrho]=\cos^2( t)\varrho+\sin^2( t)\mcA_{\rm NOT}[\varrho]
+\frac{1}{2}\sin(2 t)\mcF[\varrho]\,,
\label{eq:collision:model}
\end{equation}
where $\mcF[\varrho]=i\frac{1}{3}\sum_j [\sigma_j,\varrho]$.
It achieves the desired gate
$\mcA_{\rm NOT}$ at $t=\pi/2$.

Let us stress that this dynamical map is unital, i.e. $\mcE_t[\one]=\one$ for
all $t$, thus, its special orthogonal normal form can be illustrated inside the
tetrahedron of Pauli channels, see  \fref{fig:tray}. In
\fref{fig:evolnot} we plot $\DivFunc{}[\mcE_t]$, $\chi[\mcE_t]$ and the
value of the $\det\mcE_t$.  We see the transitions ${\sf C}^{\rm L}
\rightarrow{\sf C}^{\rm P}\setminus{\sf C}^{\rm CP} \rightarrow{\sf C}^{\rm
div}\setminus{\sf C}^{\rm P} \rightarrow \Ind{}$ and back. Notice that
in both plots the trajectory never goes through the $\cpDiv{}\setminus \LDiv{}$
region. This means that when the parametrized channels, up to rotations, belong
to \LDiv{}, so do the original ones. The transition
between P-divisible and divisible channels, i.e. \pDiv{}$\setminus$\cpDiv{} and
\Div{}$\setminus$\pDiv{}, occurs at the discontinuity in the yellow curve in 
\fref{fig:tray}. Let us note that this discontinuity 
only occurs in the space of $\vec \lambda$; it is a consequence of the
special orthogonal normal decomposition, see \eref{eq:orthogonalform}. The complete
channel is continuous in the full convex space of qubit CPTP maps.
 The transition from
$\pDiv{}\setminus \Div{}$ and back occurs at times $\pi/3$ and $2 \pi/3$.  It
can also be noted that the transition to entanglement breaking channels 
occurs shortly before the channel enters in the $\Div{} \setminus \pDiv{}$
region; likewise, the channel stops being entanglement breaking shortly after
it leaves the $\Div{} \setminus \pDiv{}$ region, see theorem~\ref{thm:EB}.

\begin{figure} 
\centering
\includegraphics[]{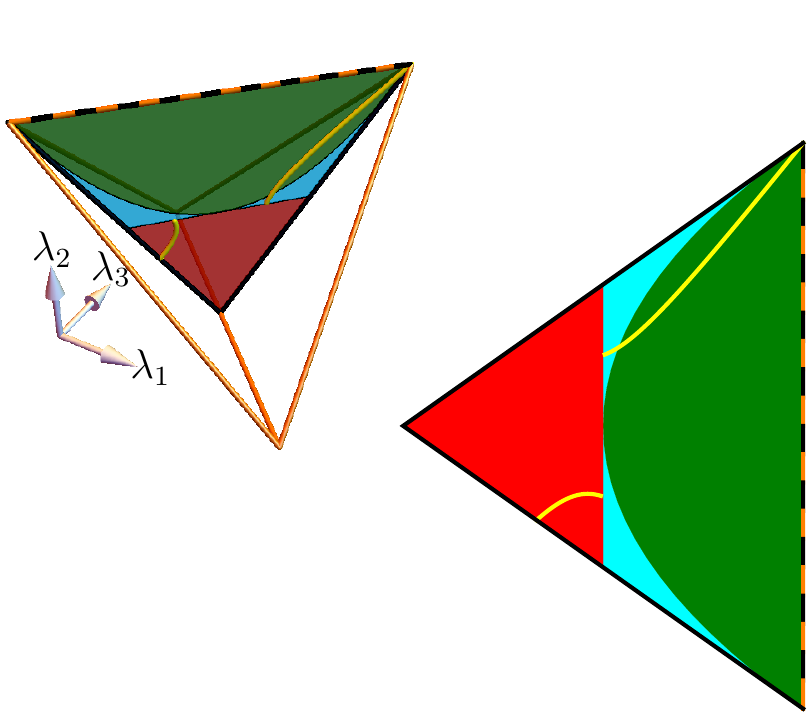}
\caption{(top left) Tetrahedron of Pauli channels with the trajectory, up to
rotations, of the quantum dynamical map~\eref{eq:collision:model} leading to
the $\mcA_\text{NOT}$ gate, as a yellow curve. (right) Cut along
the plane that contains the trajectory; there one can see the different regions
where the channel passes. For this
case, the characterization of the \LDiv{} of the channels induced gives the
same conclusions as for the corresponding Pauli channel, see
\eref{eq:orthogonalform}.
The discontinuity in the trajectory is due to the reduced representation of the
dynamical map, see \eref{eq:collision:model}; the trajectory is continuous
in the space of channels. 
See \fref{fig:cut1} for the color coding.
\label{fig:tray}}
\end{figure} 

\begin{figure} 
\centering
\includegraphics[width=\columnwidth]{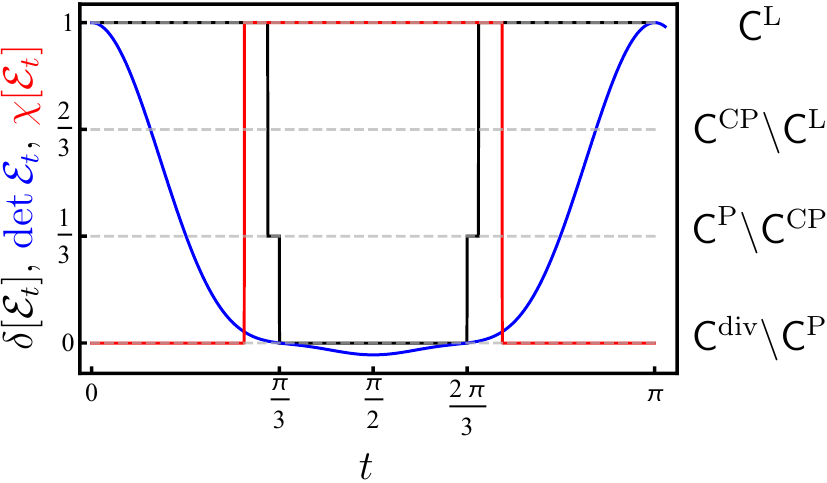}
\caption{
Evolution of divisibility, determinant, and entanglement breaking properties
of the map induced by \eref{eq:collision:model},
see \eref{eq:delta_function} and \eref{eq:chi_function}.
Notice that the channel $\mcA_\text{NOT}$, implemented at $t=\pi/2$, has minimum determinant. The horizontal gray dashed lines show the image
of the function $\delta$, with the divisibility types in the right side. It can
be seen that the dynamical map explores the divisibility sets as ${\sf C}^{\rm
L} \rightarrow{\sf C}^{\rm P}\setminus{\sf C}^{\rm CP} \rightarrow{\sf C}^{\rm
div}\setminus{\sf C}^{\rm P} \rightarrow \Ind{}$ and back. The channels are
entanglement breaking in the expected region.\label{fig:evolnot}}
\end{figure} 

Consider now the dynamical map induced by a two-level atom interacting with a mode of a
boson field. This model serves as a workhorse to explore a great variety of
phenomena in quantum optics~\cite{0953-4075-46-22-220201}. Using
the well known \textit{rotating wave approximation} one arrives to the
Jaynes-Cummings model~\cite{jaynescummings}, whose Hamiltonian is
\begin{equation}
H = \frac{\omega_a}{2}\sigma_z
    +\omega_f \left( a^{\dagger}a+\frac{1}{2} \right)
    + g\left( \sigma_- a^ {\dagger}+ \sigma_+ a \right).
\label{eq:Ham_qb}
\end{equation}
By initializing the environment in a coherent state $\ket{\alpha}$,
one gets the familiar \textit{collapse and revival} setting. Considering a particular set of parameters shown in~\fref{fig:revival1}, we
constructed the channels parametrized by time numerically, and studied their
divisibility and entanglement-breaking properties. In the same figure we plot
functions $\DivFunc{}[\mcE_t]$ and $\chi[\mcE_t]$, together with the probability of
finding the atom in its excited state $p_e(t)$, to study and compare the
divisibility properties with the features of the collapses and revivals. The
probability $p_e(t)$ is calculated choosing the ground state of the free
Hamiltonian  $\omega_a /2 \sigma_z$ of the qubit, and it is given
by~\cite{klimovbook}:
\begin{equation}
p_e(t)=\frac{\langle\sigma_z(t)\rangle+1}{2},
\end{equation}
where 
\begin{equation*}
\langle \sigma_z(t)\rangle=-\sum_{n=0}^{\infty}P_n \left( \frac{\Delta^2}{4 \Omega_n^2}+\left(1-\frac{\Delta^2}{4 \Omega_n^2}\right)\cos \left(2 \Omega_n t \right)\right),
\end{equation*}
with $P_n=e^{-|\alpha|^2}|\alpha|^{2n}/n!$, $\Omega_n=\sqrt{\Delta^2/4+g^2n}$ and $\Delta=\omega_f-\omega_a$ the detuning.
\begin{figure} 
\centering
\includegraphics[width=\columnwidth]{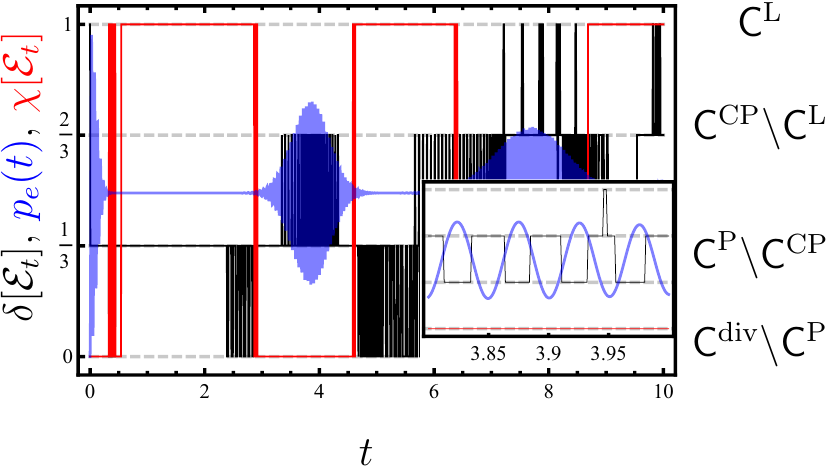}
%
\caption{Black and red curves show functions $\delta$ and $\chi$ of the
channels induced by the Jaynes-Cummings model over a two-level system,
see~\eref{eq:Ham_qb} with the environment initialized in a coherent state
$\ket{\alpha}$. The blue curve shows the probability of finding the
two-level atom in its excited state, $p_e(t)$.
The figure shows that the fast oscillations in $\delta$ occur
roughly at the same frequency as the ones of $p_e(t)$, see the inset. Notice
that there are fast transitions between $\pDiv{}\setminus\cpDiv{}$ and
$\cpDiv{}\setminus\LDiv{}$ occurring in the region of revivals, with a few
transitions between $\cpDiv{}\setminus\pDiv{}$ and $\LDiv{}$ in the second
revival. The function $\chi$ shows that during revivals channels are not
entanglement breaking, but we find that channels belonging to
$\Div{}\setminus\pDiv{}$ are always entanglement breaking, in agreement with
theorem~\ref{thm:EB}. The particular chosen set of parameters are
$\alpha=6$, $g=10$, $\omega_a=5$, and $\omega_f=20$.  \label{fig:revival1}}
\end{figure} 

The divisibility indicator function $\delta$ exhibits
an oscillating behavior, roughly at the same frequency of $p_e(t)$, see inset
in \fref{fig:revival1}.  The figure shows fast periodic transitions
between $\pDiv{}\setminus\cpDiv{}$ and $\cpDiv{}\setminus\LDiv{}$ occurring in
the region of revivals. There are also few transitions among
$\cpDiv{}\setminus\pDiv{}$ and $\LDiv{}$ in the second revival. Respect to the
entanglement breaking and the function $\chi$, there are no fast transitions in
the former, and during revivals, channels are not entanglement breaking. We
also observe that channels belonging to $\Div{}\setminus \pDiv{}$ are
entanglement breaking, which agrees with theorem~\ref{thm:EB} for the non-unital case.


\chapter{Singular Gaussian quantum channels} \label{chap:singular}
\begin{flushright}
\textit{Self-education is, I firmly believe, the only kind of education there is.}\\ 
Isaac Asimov
\end{flushright}
In this chapter we derive the conditions for \dgqc{} to be singular,
see~\sref{sec:ptp}. In particular we will show that only the functional form
involving one Dirac delta can be singular, together with the Gaussian form.
Additionally we derive, for the non-singular cases, the conditions for the
existence of master equations that parametrize channels that have always the
same functional form. We do this by letting the channels parameters to depend
on time.
\section{Allowed singular forms} 
\label{sec:singular_forms}
There are two classes of Gaussian singular
channels. Since the inverse of a Gaussian channel $\mcG\left(
\mathbf{T},\mathbf{N}, \vec \tau \right)$ is $\mcG\left( \mathbf{T}^{-1},
-\mathbf{T}^{-1}\mathbf{N} \mathbf{T}^{-T}, -\mathbf{T}^{-1} \vec \tau
\right)$, its existence rests on the invertibility of $\mathbf{T}$. Therefore,
studying the rank of the latter we are able to explore singular forms. 
We are going to use the classification of one-mode channels developed by 
Holevo~\cite{Holevo2007}.
For singular channels there are two classes
characterized by its \textit{canonical form}~\cite{Holevo2008}, \ie{} any 
channel can be obtained by applying Gaussian unitaries before and
after the canonical form. The class called ``$A_1$'' corresponds to singular
channels with $\text{Rank} \left(\mathbf{T}\right)=0$ and coincide with the
family of \textit{total depolarizing channels}. The class ``$A_2$'' is
characterized by $\text{Rank}\left(\mathbf{T}\right)=1$. Both channels are
entanglement-breaking~\cite{Holevo2008}.

Before analyzing the functional forms constructed in this work, let us study
channels with \gf{}. The tuple of the affine transformation, corresponding to
the propagator $J_\text{G}$, \eref{eq:gf}, were introduced in
Ref.~\cite{PazSupplementary} up to some typos. Our calculation for this tuple,
following~\eref{eq:changeofrepresentation}, is:
\begin{align}
\mathbf{T}_\text{G}&=\left(
\begin{array}{cc}
 -\frac{b_4}{b_3} & \frac{1}{b_3} \\
 \frac{b_1 b_4}{b_3}-b_2 & -\frac{b_1}{b_3} \\
\end{array}
\right),\nonumber \\
\mathbf{N}_\text{G}&=\left(
\begin{array}{cc}
 \frac{2 a_3}{b_3^2} & \frac{a_2}{b_3}-\frac{2 a_3 b_1}{b_3^2} \\
 \frac{a_2}{b_3}-\frac{2 a_3 b_1}{b_3^2} & -2 \left(-\frac{a_3 b_1^2}{b_3^2}+\frac{a_2 b_1}{b_3}-a_1\right) \\
\end{array}
\right),\nonumber \\
\vec \tau_\text{G}&= \left( -\frac{c_2}{b_3},\frac{b_1 c_2}{b_3}-c_1 \right)^\text{T}.
\label{eq:singular_tupleII}
\end{align}
It is straightforward to check that for $b_2=0$, $\mathbf{T}_\text{G}$ is
singular with $\text{Rank}\left(\mathbf{T}_\text{G}\right)=1$, \ie{} it belongs
to class $A_2$. Due to the full support of Gaussian functions, it was
surprising that Gaussian channels with \gf{} have singular limit. In this case
the singular behavior arises from the lack of a Fourier factor for $x_f r_i$,
see~\eref{eq:gf}.  This is the only singular case for \gf{}.

Now we analyze functional forms derived in~\sref{sec:ptp}. The complete
positivity conditions of the form $\tilde J_\text{III}$, presented in
\eref{eq:cpII}, have no solution for $\alpha\to 0$ and/or $\gamma \to 0$, thus,
this form cannot lead to singular channels. This is not the case for $\tilde
J_\text{I}$, \eref{eq:formcharI}, which leads to singular operations belonging
to class $A_2$ for 
\begin{equation}
\alpha e_2=0,
\label{eq:reduction_for_A2}
\end{equation}
 and to class $A_1$ for 
\begin{equation}
e_2 =\alpha=b_2= 0.
\label{eq:reduction_for_A1}
\end{equation}
For the latter, the complete
positivity conditions, see~\eref{eq:ccp}, read: 
\begin{equation}
e_1\leq a_1.
\label{eq:reducedCP}
\end{equation}

By using an initial state characterized by $\sigma_i$ and $\vec d_i$ we
can compute the explicit dependence of the final states on the initial
parameters.
The final states for channels of class $A_2$ with the functional form involving one delta, see~\eref{eq:deltaop2}, and with $e_2=0$, are
\begin{align}
\left(\sigma_f\right)_{11}&=\frac{1}{2 e_1},\nonumber\\
\left(\sigma_f\right)_{22} &= \left(\frac{\alpha}{\beta}\right)^2\left(\frac{b_3^2}{2 e_1}+2 a_3\right)+
\frac{\alpha}{\beta } \left(2 a_2+\frac{b_1 b_3}{e_1}\right)+ 2 a_1+\frac{b_1^2}{2 e_1}+s_1,\nonumber\\
\left(\sigma_f\right)_{12} &= -\frac{\alpha  }{\beta}\frac{b_3}{2 e_1}-\frac{b_1}{2 e_1},\nonumber\\
\vec d_f \left( s_3\right) &= \left(0, -\frac{\alpha}{\beta }c_2-c_1+s_2 \right)^\text{T},
\end{align}
where
\begin{align}
s_1&=\left(b_2^2+2 \frac{\alpha}{\beta} b_2 b_4+\left(\frac{\alpha}{\beta}\right)^2b_4^2\right)\left(\sigma_i\right)_{11}
-2\left(\frac{\alpha}{\beta}b_2+\left(\frac{\alpha}{\beta}\right)^2 b_4\right)\left(\sigma_i\right)_{12}\nonumber\\
&\qquad \qquad+ \left(\frac{\alpha}{\beta}\right)^2\left(\sigma_i\right)_{22},
\nonumber\\
s_2&=\left(\frac{\alpha}{\beta} b_4+ b_2\right)(d_i)_1 -\frac{\alpha
}{\beta }(d_i)_2.
\end{align}
For the same functional form but now with $\alpha=0$, the final states are
\begin{align}
\left(\sigma_f\right)_{11}&=\frac{e_2^2 }{4 e_1^2}\left(\sigma_i\right)_{11}+\frac{1}{2 e_1},\nonumber\\
\left(\sigma_f\right)_{12}&=\left(\frac{b_2 e_2}{2 e_1}-\frac{b_1 e_2^2 }{4 e_1^2}\right)\left(\sigma_i\right)_{11}-\frac{b_1}{2 e_1},\nonumber\\
\left(\sigma_f\right)_{22}&=2 a_1+\left(b_2-\frac{b_1 e_2}{2 e_1}\right)^2 \left(\sigma_i\right)_{11}+\frac{b_1^2}{2 e_1},
\end{align}
and
\begin{equation}
\vec d_f=\left(\frac{e_2}{2 e_1} \left(\vec d_i \right)_1, \left(b_2-\frac{b_1 e_2}{2 e_1}\right)\left(\vec d_i \right)_1-c_1\right)^\text{T}.
\end{equation}
The explicit formulas of the final states for channels of class $A_2$ with Gaussian form
are
\begin{align}
\left(\sigma_f\right)_{11}\left( s_1\right)&=\frac{2 a_3}{b_3^2}+s_1,\nonumber\\
\left(\sigma_f\right)_{12}\left(
s_1\right)&=\frac{a_2}{b_3}-\frac{2 a_3 b_1}{b_3^2}-b_1 s_1,\nonumber\\
\left(\sigma_f\right)_{22}\left( s_1
\right)&=\frac{b_1 \left(b_3 \left(b_1 b_3 s_1-2 a_2\right)+2 a_3 b_1\right)}{b_3^2}+2 a_1,\nonumber\\
\vec
d_f \left( s_2\right)&=\left(s_2-\frac{c_2}{b_3},b_1
\left(\frac{c_2}{b_3}-s_2\right)-c_1\right)^\text{T},
\end{align}
where
\begin{align}
s_1&=\frac{b_4^2}{b_3^2} \left(\sigma_i\right)_{11}-\frac{2 b_4 }{b_3^2}\left(\sigma_i\right)_{12}+\frac{1}{b_3^2}\left(\sigma_i\right)_{22},\nonumber\\
s_2&=\frac{1}{b_3}\left(d_i\right)_2-\frac{b_4
}{b_3}\left(d_i\right)_1.
\end{align}
See~\fref{fig:2} for an schematic description of the final states. From such
combinations it is obvious that we cannot solve for the initial state
parameters given a final state as expected; this is because the parametric
space dimension is reduced from $5$ to at most $3$. The channel belonging to $A_1$ [see
\eref{eq:tupleI} with $e_2 =\alpha=b_2= 0$ and \eref{eq:reducedCP}] maps every
initial state to a single one characterized by $\sigma_f=\mathbf{N}$ and $\vec d_f=
\left(0,-c_1 \right)^\text{T}$, see \fref{fig:1} for a schematic description.
%

According to our ans\"atze [see equations (\ref{eq:deltaop1}) and
(\ref{eq:deltaop2})], we conclude that one-mode \sgqc{} can only have the functional forms
given in \eref{eq:gf} and \eref{eq:deltaop1}.
This is the central result of this chapter and can be stated as:
\begin{theorem}[One-mode singular Gaussian channels]
A one-mode Gaussian quantum channel is singular if and only if it
has one of the following functional forms in the position space representation:
\begin{enumerate}
\item $\frac{b_3}{2 \pi} \exp\Big[ \Imi{}\Big( b_1 x_f r_f+b_3x_ir_f 
+b_4x_ir_i+c_1x_f+c_2x_i \Big) -a_1 x_f^2-a_2x_fx_i-a_3x_i^2 \Big],$
\item $|\beta|\sqrt{e_1/\pi}\delta(\alpha x_f-\beta x_i) 
\exp\Big[-a_2 x_f x_i-a_1 x_f^2-a_3 x_i^2$\\
$+\Imi{} \Big(b_2 x_f r_i +b_3 r_f x_i+b_1 r_f x_f+b_4 r_i x_i+c_1 x_f+c_2 x_i\Big)
-e_1 r_f^2 -e_2 r_f  
               r_i-\frac{e_2^2 r_i^2}{4 e_1}\Big]$, with $e_2 \alpha=0$.
\end{enumerate}
\end{theorem}
\begin{corollary}[Singular classes]
A one-mode singular Gaussian channel belongs to class $A_1$ if and only if its position representation has the following form:
$$\sqrt{e_1/\pi}\delta(x_i) \exp\Big[-a_1 x_f^2+\Imi{} \Big(b_2 x_f r_i 
         +b_1 r_f x_f+c_1 x_f\Big)-e_1 r_f^2\Big].$$
Otherwise the channel belongs to class $A_2$.
\end{corollary}

Since channels on each class are connected each other by unitary
conjugations~\cite{Holevo2007}, a consequence of the theorem and the subsequent
corollary is that the set of allowed forms must remain invariant under unitary
conjugations. 
To show this we must know the possible functional forms of Gaussian unitaries. They are given by the following lemma for one mode:
\begin{lemma}[One-mode Gaussian unitaries]
Gaussian unitaries can have only \gf{} or the one given by \eref{eq:deltaop2}.
\end{lemma}
\begin{proof}
Recalling that for a unitary \gqc{},  $\mathbf{T}$  must be symplectic
($\mathbf{T}\Omega \mathbf{T}^\text{T}=\Omega$) and $\mathbf{N}=\mathbf{0}$. However,
an inspection to \eref{eq:deltaI_Ps} lead us to note that
$\mathbf{N}\neq\mathbf{0}$ unless $e_1$ diverges. Thus, Gaussian unitaries
cannot have the form $J_\text{I}$ [see~\eref{eq:deltaop1}]. An inspection of
matrices $\mathbf{T}$ and $\mathbf{N}$ of \gqc{} with \gf{} [see
\eref{eq:singular_tupleII}] and the ones for $J_\text{II}$ [see equations
(\ref{eq:deltaII_Ps}) and (\ref{eq:tupleII})] lead us to note the following 
two observations: (i) in
both cases we have
$\mathbf{N}=0$ for $a_n=0 \ \ \forall n$; (ii) the matrix $\mathbf{T}$ is
symplectic for \gf{} when $b_2=b_3$, and when $\alpha  \eta =\beta  \gamma $ for $J_\text{II}$. In particular the identity map has the last form. This completes
the proof.
\end{proof}

\newcommand{\deltaAa}{$\delta_{\text{A}_2}^\alpha$}
\newcommand{\deltaAe}{$\delta_{\text{A}_2}^{e_2}$}
\newcommand{\deltaAae}{$\delta_{\text{A}_2}^{\alpha,e_2}$}
\newcommand{\deltaA}{$\delta_{\text{A}_1}$}
\newcommand{\gaussA}{$\mcG_{\text{A}_2}$}
\newcommand{\deltaU}{$\delta_{\mcU}$}
\newcommand{\gaussU}{$\mcG_{\mcU}$}
One can now compute the concatenations of the \sgqc{}s with Gaussian unitaries.
This can be done straightforward using the well known formulas for Gaussian
integrals and the Fourier transform of the Dirac delta.
Given that the calculation is elementary, and for sake of brevity, we present
only the resulting forms of each concatenation. To show this compactly we
introduce the following abbreviations: Singular channels belonging to class
$A_2$ with form $J_\text{I}$ and with $\alpha=0$, $e_2=0$ and $\alpha=e_2=0$,
will be denoted as \deltaAa{}, \deltaAe{} and \deltaAae{}, respectively;
singular channels belonging to the same class but with \gf{} will be denoted as
\gaussA{}; channels belonging to class $A_1$ will be denoted as \deltaA{};
finally Gaussian unitaries with \gf{} will be denoted as \gaussU{} and the ones
with form $J_\text{II}$ as \deltaU{}. Writing the concatenation of two channels
in the position representation as
\begin{equation}
J^{(\text{f})}(x_f,r_f;x_i,r_i)=
\int_{\mathbb{R}^2} dx' dr' J^\text{(1)}\left(x_f, r_f; x',r' \right)J^{(2)}\left( x',r';x_i, r_i\right),
\label{eq:concat}
\end{equation}
the resulting functional forms for $J^{(\text{f})}$ are given in 
table~\ref{tab:concats}. As expected, the table shows that the integral has
only the forms stated by our theorem. Additionally it shows the cases when
unitaries change the functional form of class $A_2$, while for class $A_1$
$J^{(\text{f})}$ has always the unique form enunciated by the corollary.
\begin{table}[h]
\centering
{\large
\begin{tabular}{|c|c|c|}
\hline 
$J^{(1)}$ & $J^{(2)}$ & $J^{(\text{f})}$ \\ 
\hline 
\deltaAa{} & \gaussU{} & \gaussA{} \\ 
\hline 
\gaussU{} & \deltaAa{} & \deltaAa{} \\ 
\hline 
\deltaAa{} & \deltaU{} & \deltaAa{} \\ 
\hline 
\deltaU & \deltaAa{} & \deltaAa{} \\ 
\hline 
\deltaAe{} & \gaussU{} & \deltaAe{} \\ 
\hline 
\gaussU{} & \deltaAe{} & \gaussA{} \\ 
\hline 
\deltaAe{} & \deltaU{} & \deltaAe{} \\ 
\hline 
\deltaU{} & \deltaAe{} & \deltaAe{} \\ 
\hline 
\gaussU{},\deltaU{} & \deltaAae{} & \deltaAae{} \\ 
\hline 
\deltaAae{} & \gaussU{},\deltaU{} & \deltaAae{} \\ 
\hline
\deltaU{},\gaussU{} & \deltaA{} & \deltaA{} \\ 
\hline 
\deltaA{} & \deltaU{},\gaussU{} & \deltaA{} \\ 
\hline 
\end{tabular}
}
\caption{The first and second columns show the functional forms of $J^{(1)}$
and $J^{(2)}$, respectively. The last column shows the resulting form of the
concatenation of them, see~\eref{eq:concat}. See main text for symbol coding.
\label{tab:concats}}
\end{table} 

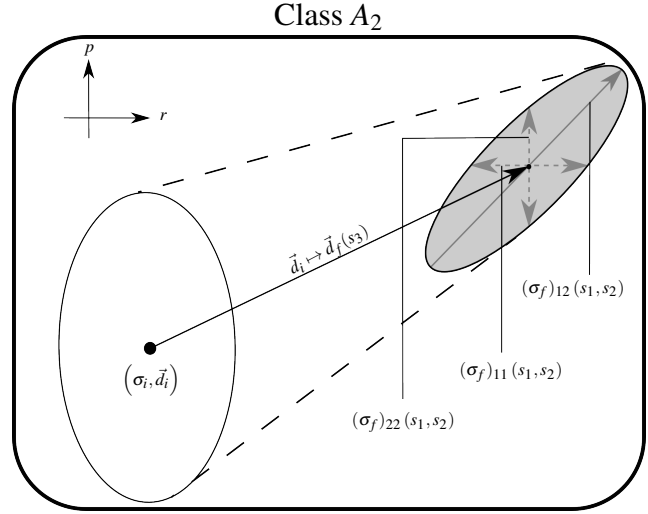
\begin{figure} 
\centering
\scalebox{1.1}{
\begin{tikzpicture}
\node at (0,3.1) {Class $A_2$};
\node at (0,0) {\includegraphics[scale=0.13]{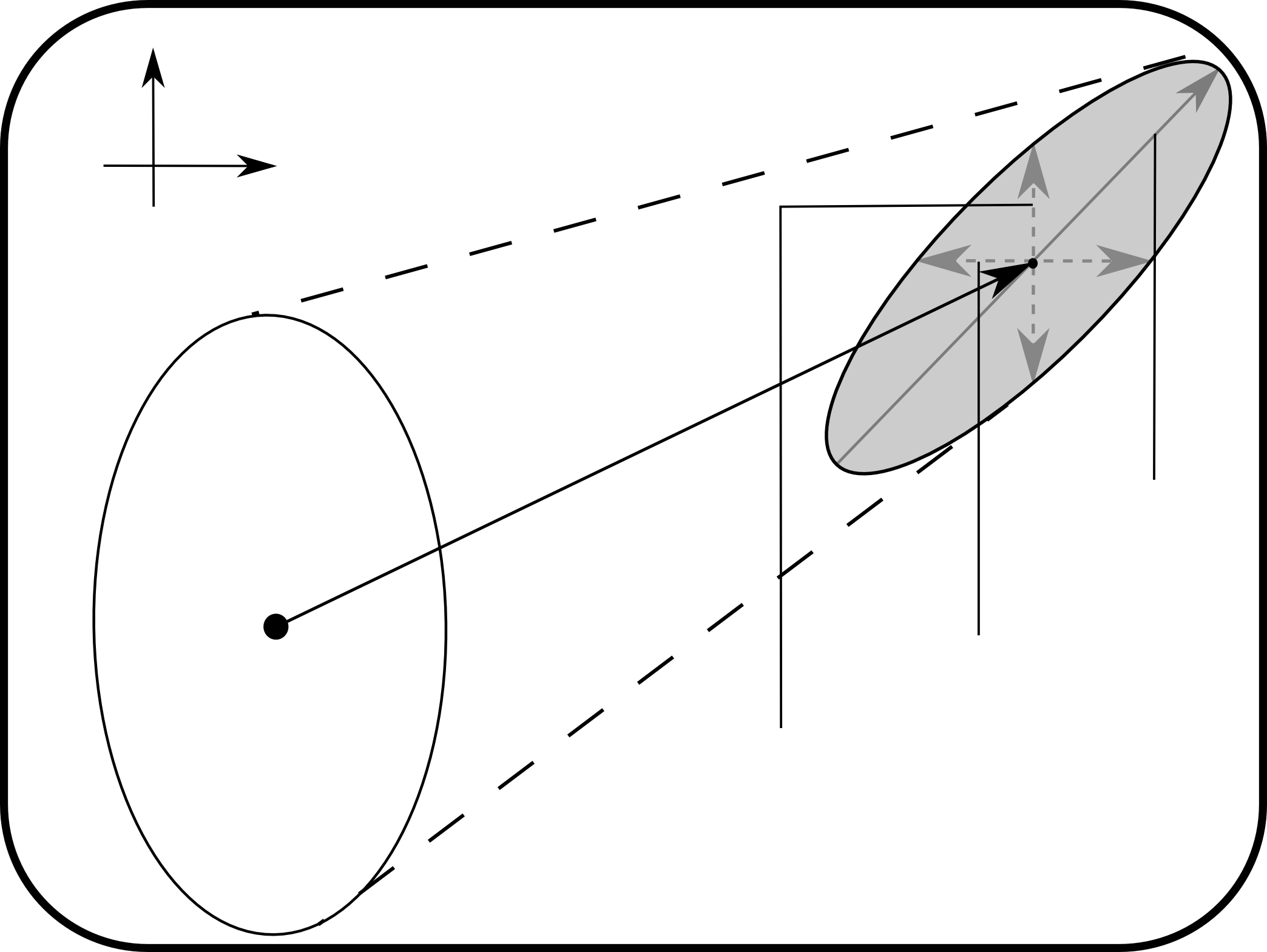}};
\node at (-2,1.9) {{\tiny $r$}};
\node at (-2.9,2.7) {{\tiny $p$}};
\node at (2.2,-1.2) {{\tiny $ (\sigma_f)_{11}\left( s_1,s_2\right)$}};
\node at (0.9,-1.8) {{\tiny $ (\sigma_f)_{22}\left( s_1,s_2\right)$}};
\node at (2.95,-0.2) {{\tiny $ (\sigma_f)_{12}\left( s_1,s_2\right)$}};
\node at (0,0.29) {\rotatebox{27.5}{{\tiny $\vec d_i \mapsto \vec d_f (s_3)$}} };
\node at (-2.15,-1.3) {{\tiny $\left(\sigma_i, \vec d_i \right)$}};
\end{tikzpicture}
}
\caption{Schematic picture of the channels belonging to class $A_2$. The explicit dependence of the final state in terms of the combinations $s_1$, $s_2$ and $s_3$ are presented in the appendix. As well the formulas for $s_i$ depending on the form of the channel.
\label{fig:2}
}
\end{figure} 
\begin{figure} 
\centering
\scalebox{1.1}{
\begin{tikzpicture}
\node at (0,3.1) {Class $A_1$};
\node at (0,0) {\includegraphics[scale=0.13]{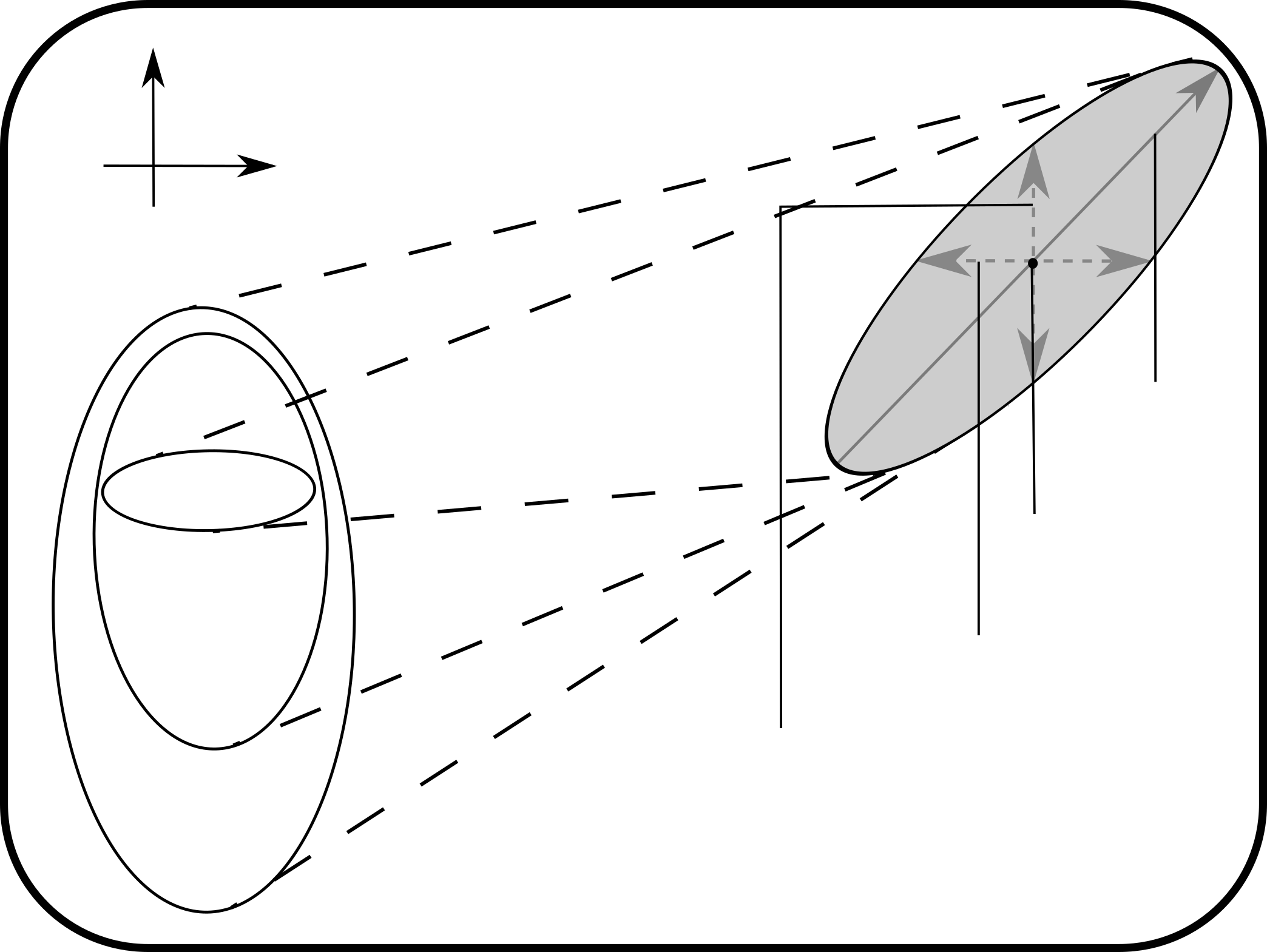}};
\node at (-2,1.9) {{\tiny $r$}};
\node at (-2.9,2.7) {{\tiny $p$}};
\node at (2.2,-1.1) {{\tiny $ \frac{1}{2 e_1}$}};
\node at (0.8,-1.8) {{\tiny $ 2 a_1+\frac{b_1^2}{2 e_1}$}};
\node at (3.15,0.3) {{\tiny $ -\frac{b_1}{2 e_1}$}};
\node at (2.8,-0.4) {{\tiny $\left(0,-c_1 \right)^\text{T}$}};
\node at (-2.5,-0.8) {{\tiny $\left(\sigma_i, \vec d_i \right)$}};
\end{tikzpicture}
}
\caption{Schematic picture of the class $A_1$. Every channel of this class maps every initial quantum state, in particular \gs{}s characterized by $\left(\sigma_i, \vec d_i \right)$, to a Gaussian state that depends only on the channel parameters. We indicate in the figure the values of the corresponding components of the first and second moments of the final Gaussian state.
\label{fig:1}
}
\end{figure} 

\section{Existence of master equations}  
\label{sec:master_equations}
In this section we show the conditions under which master equations, 
associated with the channels derived in~\sref{sec:ptp}, exist.
To be more precise, we study if
the functional forms derived above
parametrize channels belonging to one-parameter differentiable families of
\gqc{}s. As a first step, we let the coefficients of forms presented in equations
(\ref{eq:deltaop1}) and (\ref{eq:deltaop2}) to depend on time. Later we derive the conditions under which they bring
any quantum state $\rho(x,r;t)$ to $\rho(x,r;t+\epsilon)$ (with $\epsilon>0$
and $t\in [0,\infty)$) smoothly, while holding the specific functional form of
the channel, \ie{} 
\begin{equation}
\rho(x,r;t+\epsilon)=\rho(x,r;t)+\epsilon \mcL_t
\left[\rho(x,r;t)\right]+\mcO(\epsilon^2),
\end{equation}
where both $\rho(x,r;t)$ and $\rho(x,r;t+\epsilon)$ are propagated from $t=0$ with channels either with the form $J_\text{I}$ or $J_\text{II}$, and $\mcL_t$ is a bounded superoperator in the state subspace.
This is basically the problem of \textit{the existence of a master equation}
\begin{equation}
\partial_t \rho(x,r;t)=\mathcal{L}_t\left[\rho(x,r;t)\right],
\label{eq:master_equation}
\end{equation} 
for such functional forms. Thus, the problem is reduced to prove the existence
of the linear generator $\mcL_t$, also known as \textit{Liouvillian}.

To do this we use an ansatz proposed in Ref.~\cite{Karrlein1997} to investigate the
existence and derive the master equation for \gf{}s,
\begin{equation}
\mathcal{L}=\mathcal{L}_c(t)+(\partial_x,\partial_r) 
\mathbf{X}(t) 
\Bigg(\begin{matrix} \partial_x \\ \partial_r \end{matrix}\Bigg)
+(x,r) 
\mathbf{Y}(t) 
\Bigg(\begin{matrix}
\partial_x \\ \partial_r
\end{matrix}\Bigg) 
+(x,r) 
\mathbf{Z}(t) 
\Bigg(\begin{matrix}
x \\ r
\end{matrix}\Bigg)
\label{eq:ansatz}
\end{equation}
where $\mathcal{L}_c(t)$ is a complex function and 
\begin{equation}
\mathbf{X}(t) =
\Bigg(\begin{matrix}
X_{xx}(t) & X_{xr}(t) \\
X_{rx}(t) & X_{rr}(t)
\end{matrix}\Bigg) 
\label{eq:defX}
\end{equation}
is a complex matrix as well as 
$\mathbf{Y}(t)$ and $\mathbf{Z}(t)$, whose entries are defined in a similar
way as in \eref{eq:defX}. Note that
$\mathbf{X}(t)$ and $\mathbf{Z}(t)$ can always be chosen symmetric, \ie{}
$X_{xr}=X_{rs}$ and $Z_{xr}=Z_{rx}$. Thus, we must determine $11$ time-dependent
functions from~\eref{eq:ansatz}. This ansatz is
also appropriate to study the functional forms introduced in this work, given
that the left hand side of \eref{eq:master_equation} only involves quadratic
polynomials in $x$,
$r$, $\partial/\partial x$ and $\partial/\partial r$, as in the \gf{} case.

Notice that singular channels do not admit a master equation since its
existence implies that channels with the functional form involved can be found
arbitrarily close from the identity channel. This is not possible for singular
channels due to the continuity of the determinant of the matrix $\mathbf{T}$.

For the non-singular cases presented in equations
(\ref{eq:deltaop1}) and (\ref{eq:deltaop2}), the condition for the
existence of a master equation is obtained as follows. (i) Substitute
the ansatz of \eref{eq:ansatz} in the right hand side of the
\eref{eq:master_equation}. (ii) Define $\rho(x,r;t)$ using \eref{eq:propagacion}, given an initial condition
$\rho(x,r;0)$, for each functional form $J_\text{I,II}$. (iii) Take
$\rho_f(x_f,r_f)\to\rho(x,r;t)$ and $\rho_i(x_i,r_i)\to \rho(x,r;0)$. Finally, (iv) compare both sides of
\eref{eq:master_equation}. Defining $A(t)=\alpha(t)/\beta(t)$ and
$B(t)=\gamma(t)/\eta(t)$, the conclusion is that for both $J_\text{I}$ and
$J_\text{II}$, a master equations exist if
\begin{equation}
c(t)\propto A(t)
\end{equation}
holds, where $c(t)=
c_1(t)+A(t)c_2(t)$. Additionally, for the form $J_I$ the solutions for the
matrices $\mathbf{X}(t)$, $\mathbf{Y}(t)$ and $\mathbf{Z}(t)$ are given by 
\begin{equation}
\begin{gathered}
X_{xx}=X_{xr}=Y_{rx}=Z_{rr}=0, \\
Y_{xx}=\frac{\dot{A}}{A}, \\
\mathcal{L}_c=Y_{rr}=\frac{\dot{e}_1}{e_1}-\frac{\dot{e}_2}{e_2}, \\
X_{rr}=\frac{\dot{e}_1}{4e_1^2}-\frac{\dot{e}_2}{2e_1e_2}, \\
Y_{xr}=\Imi \left(\frac{\lambda_1\dot{e}_2}{e_1e_2}+\frac{\lambda_2\dot{A}}{e_2A}-\frac{\lambda_1\dot{e}_1}{2e_1^2}-\frac{\dot{\lambda}_2}{e_2}\right),\\
Z_{xx}=\frac{\lambda_1^2}{2}\Bigg(\frac{\dot{e}_2}{e_1e_2}-\frac{\dot{e}_1}{2e_1^2}\Bigg)+\frac{\lambda_1}{e_2}\Bigg(\frac{\lambda_2\dot{A}}{A}-\dot{\lambda}_2\Bigg)+2\lambda_3\frac{\dot{A}}{A}-\dot{\lambda}_3,  \\
Z_{xr}=\Imi\left(\frac{\dot{A}}{A}\Bigg(\frac{e_1\lambda_2}{e_2}-\frac{\lambda_1}{2}\Bigg)+\frac{\dot{\lambda}_1}{2}-\frac{\dot{\lambda}_2e_1}{e_2}+\frac{\lambda_2}{2}\Bigg(\frac{\dot{e}_2}{e_2}-\frac{\dot{e}_1}{e_1}\Bigg)\right),
\end{gathered}
\end{equation}
where we have defined the following coefficients: $\lambda_1=b_1+Ab_3$, $\lambda_2=b_2+Ab_4$ and $\lambda_3=a_1+Aa_2+A^2a_3$.

For the form $J_\text{II}$ the solutions are the following
\begin{equation}
\begin{gathered}	
\mathcal{L}_c=X_{xx}=X_{xr}=X_{rr}=Z_{rr}=0, \\
Y_{rx}=Y_{xr}=0, \\
Y_{xx}=\frac{\dot{A}}{A}, \, Y_{rr}=\frac{\dot{B}}{B}. \\ 
Z_{xx}=a_2(t) \dot A(t)+\frac{2 a_1(t) \dot A(t)}{A(t)}-A(t)^2\\-\dot a_3(t)-A(t) \dot{a}_2(t)-\dot a_1(t),  \\
Z_{xr}=\Imi\left(\frac{1}{2} \dot{\lambda}-\frac{\lambda}{2}\Bigg(\frac{\dot{A}}{A}+\frac{\dot{B}}{B}\Bigg)\right), \\
\end{gathered}
\end{equation}
where $\lambda=b_1+Ab_3+B(b_2+Ab_4)$.

\chapter{Summary and conclusions}
\label{chap:summary}
\begin{flushright}
\textit{Living is worthwhile if one can contribute in some small way to this
endless chain of progress.}\\
Paul A.M. Dirac
\end{flushright}
In this thesis we have introduced two works developed during my PhD. The first
one was devoted to study quantum channels from the point of view of their
divisibility properties. We made use of several results from the literature,
specially from the seminal work by M. M. Wolf and J. I. Cirac~\cite{Wolf2008},
and completed and fixed some results of Ref.~\cite{cirac}. This led to the
construction of a tool to decide whether a quantum channel can be implemented
using time-independent Markovian master equations or not, for the finite
dimensional case. We additionally proved three theorems relating some of the
studied divisibility types. Some of the tools introduced in
chapter~\ref{chap:reps} are results from other paper developed during my PhD,
where I am a secondary author, see Ref.~\cite{montesgorin}. In the second work
we have studied one-mode Gaussian channels without Gaussian functional form in
the position state representation. We performed a characterization based on the
universal properties that quantum channels must fulfill; in particular we
studied the case of singular channels. We showed that the transition from
unitarity to non-unitarity can correspond directly to a change in the
functional form of the channel, in particular it turns out that functional form
with one Dirac delta factor do not parametrize unitary channels. Additionally
in this project we derived the conditions under which master equations for
particular functional
forms exist.\par
Let us summarize the results for the first project in more detail. We
implemented the known conditions to decide the compatibility of channels with
time-independent master equations (the so called L-divisibility) for the
general
diagonalizable case, and a discussion of the parametric space of Lindblad
generators was given. We additionally clarified one of the results of the paper
\cite{Wolf2008}. There, the authors arrived to erroneous conclusions for the case of
channels with negative eigenvalues. In our work we handled this case carefully.
For unital qubit channels it was shown that every
infinitesimal divisible map can be written as a concatenation of one L-divisible
channel and two unitary conjugations. For the particular case of Pauli channels
case, we have shown that the sets of infinitely divisible and L-divisible
channels coincide. We made an interesting observation, connecting the concept
of divisibility with the quantum information concept of entanglement-breaking
channels: we found that divisible but not infinitesimal divisible qubit
channels (in positive but not necessarily completely positive maps)  are
necessarily entanglement-breaking. We also noted
that the intersection of indivisible and P-divisible channels is not empty.
This allows us to implement indivisible channels with infinitesimal positive and trance preserving
maps.
Finally, we studied the possibility of dynamical transitions between different
classes of divisibility channels. We argued that all the transitions are, in
principle, possible, given that every divisibility set appears connected in our plots. We exploited two simple models of dynamical maps to
demonstrate that these transitions exist. They clearly illustrate how the channels
evolutions change from being implementable by Markovian dynamical maps (infinitesimal divisible in complete positive maps and/or L-divisible) to
non-Markovian (divisible but not infinitesimal divisible or infinitesimal divisible in positive but not complete positive maps), and vice versa. \par
For the second project we have critically reviewed the deceptively natural idea that
Gaussian quantum channels always admit a Gaussian functional form. To this end,
we went beyond the pioneering characterization of Gaussian channels with
Gaussian form presented in Ref.~\cite{PazSupplementary} in two new directions.
First we have shown that, starting from their most general definition 
(a quantum operation that takes Gaussian states to Gaussian states), a more general parametrization of
the coordinate representation of the one-mode case exists, that admits
non-Gaussian functional forms. Second, we were able to provide a black-box
characterization of such new forms by imposing complete positivity (not
considered in Ref.~\cite{PazSupplementary}) and trace preserving
conditions. While our parametrization connects with the analysis done by
Holevo~\cite{Holevo2008} in the particular cases where besides having a
non-Gaussian form the channel is also singular, it also allows the study of
Gaussian unitaries, thus providing similar classification schemes. We completed
the classification of the studied types of channels by deriving the form of the
Liouvillian Liouvillian superoperator that generates their time evolution in the form of a
master equation. Surprisingly, Gaussian quantum channels without Gaussian form
can be experimentally addressed by means of the celebrated Caldeira-Legget
model for the quantum damped harmonic oscillator~\cite{Gert}, where the new types of channels
described here naturally appear in the sub-ohmic regime. 
\par
We are interested in several directions to continue the investigation. From the
project of divisibility of quantum channels, an extension of this analysis to
larger-dimensional systems could give a deeper sight to the structure of
quantum channels. In particular we are interested on proving if the equivalence
of infinitely divisible channels and L-divisible channels is present also in
the general qubit case. Additionally a plethora of interesting questions are
related to design of efficient verification procedures of the divisibility
classes for channels and dynamical maps. For instance, \textit{can we define an extension of the Lorentz normal decomposition to systems composed of many qubits?}, this would be useful to characterize infinitesimal divisibility of many particle systems; or \textit{Is the non-countable parametrization of channels with negative eigenvalues relevant on deciding L-divisibility?}. Finally the area of channel
divisibility contains several open structural questions, e.g. the existence of
at most $n$-divisible channels.
From the project concerning one-mode Gaussian channels, a natural direction to
follow is to extend the analysis for other types of channels (or more modes) by following the
classification introduced by Holevo, see Ref.~\cite{Holevo2007}. The latter is
based on the form of a canonical form of one-mode Gaussian channels. Therefore a
connection of this classification with ours could be useful to assess quantum information features, in particular for systems for which position state representation is advantageous. 

\chapter{Appendices}
\appendix
\chapter[Exact dynamics with Lindblad master equation]{Proof of theorem ``Exact dynamics with Lindblad master equation''}
\label{sec:unbounded}
The theorem announced in chapter~\ref{chap:open_quantum_systems} is,

\fxnote{Francois: Una observación algo boba: un Hamiltoniano de dimensión finita es a la vez acotado y también tiene espectro discreto. Aquí la prueba va para Hamiltonianos  e dimensión finito, pero el resultado se afirma para Hamiltonianos acotados por debajo. Esto incluiría Hamiltonianos como el movimiento libre, con espectro continuo pero acotado por abajo. Seguramente no se ha mostrado este último, y no estoy seguro de que sea cierto. Lo más ssencillo es reformular el resultado aquí diciendo ``para Hamiltonianos de dimensión finita''}
\fxwarning{Gracias, ya lo edite}
\textbf{Theorem 2} (Exact dynamics with Lindblad master equation)
\textit{Let $\mcE_t=e^{tL}$ a quantum process generated by a Lindblad operator
$L$. The equation 
$$
\mcE_t[\rho_\text{S}]
  =\tr_\text{E} \left[ 
   e^{-\rmi H t} \left(\rho_\text{S} \otimes \rho_\text{E}\right) e^{\rmi H t}
\right],
$$ 
{\blue where $H$ has finite dimension}, holds if and only if $\mcE_t$ is a unitary conjugation for
every $t$.}
\begin{proof}
To prove this theorem, we will compute $\rho_\text{S}(t)$ to first order in
$t$, see \eref{eq:second_iteration_Neumann}. Following the master equation
of~\eref{eq:von_neumann_reduced} and taking $t=\epsilon\ll 1$, we have
\begin{align*}
\rho_\text{S}(\epsilon)&\approx\rho_\text{S}+\tr_\text{E}\int_0^\epsilon dt \left\{\rmi \left[\rho_\text{S}\otimes \rho_\text{E},H \right] \right\} \\
&= \rho_\text{S}+ \tr_\text{E} \left\{\rmi \left[\rho_\text{S}\otimes \rho_\text{E},H \right]\right\}\epsilon\\
&=\rho_\text{S}+L_\text{Exact}[\rho_\text{S}] \epsilon.
\end{align*}
where $L_\text{Exact}=\tr_\text{E} \left\{\rmi \left[\rho_\text{S}\otimes \rho_\text{E},H \right]\right\}$.
Since $\mcE_t$ is generated by a Lindblad master equation, $L_\text{Exact}$ must coincide  with the Lindblad generator since the process is homogeneous in time, \ie{} $L_\text{Exact}$ is time-independent.
Writing the global Hamiltonian as
$$H=\sum_{k,l=0} h_{kl} F^{(S)}_k\otimes F^{(E)}_l,$$
where 
$h_{kl}\in \mathbb{R}$, and $\left\{ F^{(S)}_k\right\}_k$ and $\left\{ F^{(E)}_k\right\}_k$ are orthogonal
hermitian bases of $\mcB\left( \mcH^{(S)} \right)$ and $\mcB\left( \mcH^{(E)} \right)$, respectively, with $\mcH^\text{(S)}$
and $\mcH^\text{(E)}$ are the Hilbert spaces of the central system S and the
environment E. 
We have,
\begin{align*}
L_\text{Exact}[\rho_\text{S}]&=\rmi \tr_\text{E} \left\{\sum_{k,l} h_{kl}[\rho_\text{S}\otimes \rho_\text{E},F^{(S)}_k\otimes F^{(E)}_l]\right\}\\
&= \rmi \sum_{k,l} h_{kl} \left\{ \rho_\text{S} F^{(S)}_k\tr[\rho_\text{E} F^{(E)}_l]- F^{(S)}_k\rho_\text{S}\tr[F^{(E)}_l\rho_\text{E}]\right\}\\
&=\rmi \sum_{k,l}h_{kl}\tr[F^\text{(E)}_l\rho_\text{E}]\left\{ \rho_\text{S} F^\text{(S)}_k-F^\text{(S)}_k \rho_\text{S} \right\}\\
&=\rmi [\rho_\text{S},\tilde H],
\end{align*}
where $\tilde H= \sum_{k,l} h_{kl}\tr[F^\text{(E)}_l\rho_\text{E}] F_k^{\text{(S)}}$ is an hermitian operator. Therefore $L_\text{Exact}$ is the generator of Hamiltonian dynamics with Hamiltonian $\tilde H$, thus $\mcE_t$ is unitary for all $t$.
\end{proof}

\chapter{On Lorentz normal forms of Choi-Jamiolkowski state} 
\label{sec:normal_form}
In this appendix we compute the Lorentz normal decomposition of a channel for
which one gets $b\neq 0$, supporting our observation that Lorentz normal
decomposition does not take \Jami{} states to something proportional to a \Jami{}
state. Consider the following Kraus rank three channel and its
$R_{\mcE}$ matrix, both written in the Pauli basis:
\begin{equation}
\hat \mcE=\left(
\begin{array}{cccc}
 1 & 0 & 0 & 0 \\
 0 & -\frac{1}{3} & 0 & 0 \\
 0 & 0 & -\frac{1}{3} & 0 \\
 \frac{2}{3} & 0 & 0 & \frac{1}{3} \\
\end{array}
\right),
\end{equation}
and 
\begin{equation}
R_{\mcE}=\left(
\begin{array}{cccc}
 1 & 0 & 0 & 0 \\
 0 & -\frac{1}{3} & 0 & 0 \\
 0 & 0 & \frac{1}{3} & 0 \\
 \frac{2}{3} & 0 & 0 & \frac{1}{3} \\
\end{array}
\right).
\end{equation}
Using the algorithm introduced in Ref.~\cite{Verstraete2001} to 
calculate $R_{\mcE}$'s  Lorentz
decomposition into orthochronous proper Lorentz transformations
we obtain
\begin{align}
L_1 &= 
\frac{1}{\gamma_1}
\begin{pmatrix}
 4 & 0 & 0 & 1 \\
 0 & -\gamma_1 & 0 & 0 \\
 0 & 0 & -\gamma_1 & 0 \\
 1 & 0 & 0 & 4
\end{pmatrix}, 
\end{align}
\begin{align*}
L_2 & = 
\frac{1}{\gamma_2}
\begin{pmatrix}
 89+9\sqrt{97} & 0 & 0 & -8 \\
 0 & -\gamma_2 & 0 & 0 \\
 0 & 0 & -\gamma_2 & 0 \\
 -8 & 0 & 0 & 89+9\sqrt{97}
\end{pmatrix},
\end{align*}
and
\begin{align*}
\Sigma_\mcE&=
\frac{1}{\gamma_3}
\begin{pmatrix}
 \sqrt{11+\frac{109}{\sqrt{97}} } & 0 & 0 & -\frac{\sqrt{97}+1}{\sqrt{ 89 \sqrt{97}+873}} \\
 0 & -\frac{\gamma_3}{3} & 0 & 0 \\
 0 & 0 & \frac{\gamma_3}{3} & 0 \\
 \sqrt{1+\frac{49}{\sqrt{97} }} & 0 & 0 & \sqrt{-1+\frac{49}{\sqrt{97} }} \\
\end{pmatrix}
\end{align*}
with $\gamma_1=\sqrt{15}$, $\gamma_2=3\sqrt{178 \sqrt{97}+1746}$, and 
$\gamma_3=\sqrt{30}$.
Although the central matrix $\Sigma_\mcE$ is not exactly of the form
\eref{eq:state_normal_form_singular}, it is equivalent. To see this notice
that the derivation of the theorem 2 in~\cite{Verstraete2001} considers
only decompositions into proper orthochronous Lorentz transformations. But 
to obtain the desired form, the authors change signs until they get
\eref{eq:state_normal_form_singular}; this cannot be done without changing
Lorentz transformations. If we relax the condition over $L_{1,2}$ of being proper
and orthochronous, we can bring $\Sigma_\mcE$ to the desired form by
conjugating $\Sigma_\mcE$ with $G=\diag\left(1,1,1,-1 \right)$:
\begin{equation*}
G^{-1} \Sigma_\mcE G=
\frac{1}{\gamma_3}
\begin{pmatrix}
 \sqrt{11+\frac{109}{\sqrt{97} }} & 0 & 0 & \frac{\sqrt{97}+1}{\sqrt{ 89 \sqrt{97}+873}} \\
 0 & -\frac{\gamma_3}{3} & 0 & 0 \\
 0 & 0 & \frac{\gamma_3}{3} & 0 \\
- \sqrt{1+\frac{49}{\sqrt{97} }} & 0 & 0 & \sqrt{-1+\frac{49}{\sqrt{97} }} \\
\end{pmatrix}.
\end{equation*}
In both cases (taking $\Sigma_\mcE$ or $G^{-1} \Sigma_\mcE G$ as the normal form of $R_\mcE$), the corresponding channel is not proportional to a trace-preserving
one since $b\neq 0$, see \eref{eq:state_normal_form_singular}.
This completes the counterexample.
\addcontentsline{toc}{chapter}{Bibliography}
\bibliographystyle{alpha}
\bibliography{labibliografia}
\end{document}